\documentclass[12pt]{article}
\usepackage{amsmath,amsthm}
\usepackage{eufrak}

\newcommand{\ar}{\renewcommand{\arraystretch}{1}} 
\DeclareMathAlphabet{\bb}{U}{msb}{m}{n} \gdef\C{\bb C} \gdef\dZ{\bb
Z}   \gdef\dS{\bb S} \gdef\R{\bb R}
\gdef\K{\bb K} \gdef\BH{\bb H} \gdef\F{\bb F} 

\DeclareMathOperator{\End}{End} \DeclareMathOperator{\spin}{{\bf
Spin}}

\DeclareMathOperator{\Sym}{Sym}

 \DeclareMathOperator{\SL}{SL}
\DeclareMathOperator{\SO}{SO}\DeclareMathOperator{\SU}{SU}

\newcommand{\scr}{\scriptstyle}

\newcommand{\cA}{\mathcal{A}}

\newcommand{\sA}{{\sf A}}
\newcommand{\sB}{{\sf B}}

\newcommand{\sH}{{\sf H}}

\newcommand{\sX}{{\sf X}}
\newcommand{\sY}{{\sf Y}}

\newcommand{\bsH}{{\boldsymbol{\sf H}}}
\newcommand{\sL}{\Lambda}

\newcommand{\bp}{{\bf p}}
\newcommand{\bx}{{\bf x}}
\newcommand{\by}{{\bf y}}
\newcommand{\bz}{{\bf z}}
\newcommand{\bB}{{\bf B}}
\newcommand{\bJ}{{\bf J}}

\newcommand{\bE}{{\bf E}}

\newcommand{\fM}{\mathfrak{M}}

\newcommand{\fZ}{\mathfrak{Z}}

\newcommand{\fg}{\mathfrak{g}}

\newcommand{\cl}{C\kern -0.2em \ell}

\newcommand{\bO}{\mbox{\bf O}}

\newcommand{\hypergeom}[5]{\mbox{$
_#1 F_#2\left. \!\! \left( \!\!
\begin{array}{c}
\multicolumn{1}{c}{\begin{array}{c} #3
\end{array}}\\[1mm]
\multicolumn{1}{c}{\begin{array}{c} #4
\end{array}}\end{array}
\!\! \right|\displaystyle{#5}\right) $} }
\newcommand{\ld}{\left[}
\newcommand{\rd}{\right]}

\newtheorem{thm}{Theorem}
\begin{document}
\title{Spinor Structure and Matter Spectrum}
\author{V.~V. Varlamov\thanks{Siberian State Industrial University,
Kirova 42, Novokuznetsk 654007, Russia, e-mail:
varlamov@sibsiu.ru}}
\date{}
\maketitle
\begin{abstract}
Classification of relativistic wave equations is given on the ground of interlocking representations of the Lorentz group. A system of interlocking representations is associated with a system of eigenvector subspaces of the energy operator. Such a correspondence allows one to define matter spectrum, where the each level of this spectrum presents a some state of elementary particle. An elementary particle is understood as a superposition of state vectors in nonseparable Hilbert space. Classification of indecomposable systems of relativistic wave equations is produced for bosonic and fermionic fields on an equal footing (including Dirac and Maxwell equations). All these fields are equivalent levels of matter spectrum, which differ from each other by the value of mass and spin. It is shown that a spectrum of the energy operator, corresponding to a given matter level, is non-degenerate for the fields of type $(l,0)\oplus(0,l)$, where $l$ is a spin value, whereas for arbitrary spin chains we have degenerate spectrum. Energy spectra of the stability levels (electron and proton states) of the matter spectrum are studied in detail. It is shown that these stability levels have a nature of threshold scales of the fractal structure associated with the system of interlocking representations of the Lorentz group.
\end{abstract}
\section{Introduction}
At present time there is vast observation material on spectroscopy of elementary particles (see Particle Data Group). Discovery of dynamical symmetries ($\SU(3)$- and $\SU(6)$-symmetries and so on) allows one to partially regulate and systematize these observation data, mainly for the baryon spectrum. However, a general structure of particle spectrum is unclear so far. In Heisenberg's opinion \cite{Heisen1}, the main reason of this situation is misunderstanding of elementary particle nature. Quark models, based on the approximate dynamical $\SU(N)$-symmetries, do not give answer to this question, since these models do not include lepton sector. In the standard model we see division of quark and lepton sectors, plus to these two we have a gauge sector. Such triple division of particle spectrum on the three classes `fundamental particles' is the most peculiar feature of the standard model. Over a long period of time many physicists attempted to go out from this scheme (beyond the standard model).

One of the most interesting and promising alternative schemes is a research programme proposed by Heisenberg \cite{Heisen}. The main idea of this programme is a representation of the all huge quantity of elementary particles as a \textit{\textbf{matter spectrum}}, where the each elementary particle presents a some energy level of this spectrum. At this point, an essentially important feature of such description is an absence of fundamental particles, since all the levels of matter spectrum are equivalent (enjoying equal rights). Heisenberg claimed that a notion `consists of' does not work in particle physics. Applying this notion, we obtain that the each particle consists of the all known particles. Thus, physical knowledge comes near frontiers of the area, where the notion `consists of' does not have a sense. For that reason among all elementary particles we cannot to separate one as a fundamental particle \cite{Heisen}. In Heisenberg's programme instead fundamental particles we have fundamental symmetries. Basic law, stipulating spectrum of elementary particles, is defined by fundamental symmetries.

The present paper is continuation of the works \cite{Var03,Var07}, where mapping of Bhabha equations \cite{Bha45} on the bivector space $\R^6$ (a parameter space of the Lorentz group) has been proposed. This mapping makes more clear group-theoretical structure of relativistic wave equations. Moreover, it allows one to essentially simplify a finding of solutions of these equations for any spin and mass. Solutions of obtained systems of equations have been found in the form of series in hyperspherical functions defined on the surface of two-dimensional complex
sphere\footnote{The surface of $S^c_2:$ $z^2_1+z^2_2+z^2_3=\bx^2-\by^2+2i\bx\by=r^2$ is homeomorphic to
the space of the pairs $(z_1,z_2)$, where $z_1$ and $z_2$ are the
points of a \emph{complex projective line}, $z_1\neq z_2$. This
space is a homogeneous space of the Lorentz group \cite{GGV62}. It is well-known that both quantities $\bx^2-\by^2$, $\bx\by$ are
invariant with respect to the Lorentz transformations, since a
surface of the complex sphere is invariant (Casimir operators of the
Lorentz group are constructed from such quantities). Moreover, since
the real and imaginary parts of the complex two-sphere transform
like the electric and magnetic fields, respectively, the invariance
of $\bz^2\sim(\bE+i\bB)^2$ under proper Lorentz transformations is
evident. At this point, the quantities $\bx^2-\by^2$, $\bx\by$ are
similar to the well known electromagnetic invariants $E^2-B^2$,
$\bE\bB$. This intriguing relationship between the Laplace-Beltrami
operators, Casimir operators of the Lorentz group and
electromagnetic invariants $E^2-B^2\sim\bx^2-\by^2$,
$\bE\bB\sim\bx\by$ leads naturally to a Riemann-Silberstein
representation of the electromagnetic field (see, for example,
\cite{Web01,Sil07,Bir96}).} $S^c_2$ for the fields of type $(l,0)\oplus(0,l)$ \cite{Var03} and for the fields of more general type (arbitrary spin chains) in the form of series in generalized hyperspherical functions \cite{Var07}. In the present paper, using a complex envelope of the group algebra of $\SL(2,\C)$, we associate a system of interlocking representations of this group with the system of eigenvector subspaces of the energy operator $H$ (in virtue of commutativity of energy operator and operators of complex momentum there exists a common system of eigenfunctions for these operators). Obviously, in this case fundamental symmetry of the matter spectrum is defined by the Lorentz group. In this paper we choose $\spin_+(1,3)\simeq\SL(2,\C)$ as a \emph{generating kernel} of spinor structure. However, the group $\spin_+(2,4)\simeq\SU(2,2)$ (universal covering of the conformal group $\SO_0(2,4)$) can be also chosen as such a kernel (in virtue of similarity of complex envelopes for the group algebras $\mathfrak{sl}(2,\C)$ and $\mathfrak{su}(2,2)$). Moreover, conformal group plays an important role in many areas of theoretical physics (see \cite{Ulr14} and references therein).  Further, matter spectrum is identified with the system of eigenvector subspaces of the energy operator $H$, and the each spectrum level is identified with an actualized state of some elementary particle with definite values of spin and mass. At this point, a level of the matter spectrum does not describe an elementary particle in its entirety, it describes only one state  from the variety of all its possible states at the given value of energy. An elementary particle in itself is understood as a superposition of state vectors in nonseparable Hilbert space $\bsH^S\otimes\bsH^Q\otimes\bsH_\infty$ (for more details about structure of this space\footnote{As is known, fundamental duality of nature, was first discovered in Stern-Gerlach experiment, led to introduction of spin in quantum mechanics. In 1927, Pauli \cite{Pau27} introduced spin in quantum mechanics via the doubled Hilbert space $\bsH_2\otimes\bsH_\infty$. In its turn, $\bsH^S\otimes\bsH^Q\otimes\bsH_\infty$ is a natural generalization of the Pauli space $\bsH_2\otimes\bsH_\infty$.} see \cite{Var15,Var12,Var15a}). At the reduction of superposition one from the possible states (level of the matter spectrum) is realized in accordance with superselection rules and coherent subspaces of $\bsH^S\otimes\bsH^Q\otimes\bsH_\infty$. On the other hand, according to \cite{Bha45} the each elementary particle is described by a some relativistic wave equation. In the section 4 we study structure of indecomposable relativistic wave equations related by a general interlocking scheme of the Lorentz group. Such interlocking scheme corresponds to some nonseparable state in the space $\bsH^S\otimes\bsH^Q\otimes\bsH_\infty$. So-defined matter spectrum serves as an energetic scale with respect to which we produce classification of relativistic wave equations and analysis of energy operator spectrum for the each matter level. At this point, Dirac and Maxwell equations are considered on an equal footing, that is, these equations describe particle states (electron and photon states) which are the levels of matter spectrum. For that reason Dirac and Maxwell equations have similar spinor form. It is shown that for the fields of type $(l,0)\oplus(0,l)$ and related equations we have simple non-degenerate spectrum, whereas for the arbitrary spin chains this spectrum is degenerate (at least twofold). Structure of the two stability levels (electron and proton states) of matter spectrum is studied in detail. It is shown that a place of these levels on the energetic scale of matter spectrum is not arbitrary. Namely, these levels have a nature of threshold scales of the fractal structure associated with the system of interlocking representations of the Lorentz group.
\section{Complex momentum and spinor structure}
In this section we consider basic facts concerning representations of the group $\spin_+(1,3)\simeq\SL(2,\C)$ with respect to underlying spintensor substrate (spinor structure)\footnote{Here we follow to terminology of \cite{Weyl}, where Weyl wrote that symmetric and antisymmetric tensor products are \textit{substrates} for all representations of the group $\mathfrak{c}_n$ ($\mathfrak{c}_n$ is the group of all non-singular linear transformations in $n$ dimensions).}. Complex envelope of $\mathfrak{sl}(2,\C)$ and system of eigenvector subspaces of the energy operator are introduced in sections 2.1 and 2.2. In the sections 2.3-2.4, following to Heisenberg-Fock conception, we define elementary particle as a \textbf{\textit{nonseparable (entangled) state}} in spin-charge Hilbert space $\bsH^S\otimes\bsH^Q\otimes\bsH_\infty$. The space $\bsH^S\otimes\bsH^Q\otimes\bsH_\infty$ allows us to take into account \textbf{\textit{spin}}, \textbf{\textit{charge}} and \textbf{\textit{mass}} of elementary particle. At this point, all these three characteristics are attributes of underlying spintensor substrate associated with the each state vector of $\bsH^S\otimes\bsH^Q\otimes\bsH_\infty$. Charge conjugation is interpreted as a pseudoautomorphism of the spinor structure.
\subsection{Complex momentum and energy operator}
As is known, a universal covering of the proper orthochronous Lorentz group $\SO_0(1,3)$
(rotation group of the Minkowski space-time $\R^{1,3}$)
is the spinor group
\[\ar
\spin_+(1,3)\simeq\left\{\begin{pmatrix} \alpha & \beta \\ \gamma &
\delta
\end{pmatrix}\in\C_2:\;\;\det\begin{pmatrix}\alpha & \beta \\ \gamma & \delta
\end{pmatrix}=1\right\}=\SL(2,\C).
\]

Let $\fg\rightarrow T_{\fg}$ be an arbitrary linear
representation of the proper orthochronous Lorentz group
$\SO_0(1,3)$ and let $\sA_i(t)=T_{a_i(t)}$ be an infinitesimal
operator corresponding to the rotation $a_i(t)\in\SO_0(1,3)$.
Analogously, let $\sB_i(t)=T_{b_i(t)}$, where $b_i(t)\in\SO_0(1,3)$ is
the hyperbolic rotation. The elements $\sA_i$ and $\sB_i$ form a basis of the group algebra
$\mathfrak{sl}(2,\C)$ and satisfy the relations
\begin{equation}\label{Com1}
\left.\begin{array}{lll} \ld\sA_1,\sA_2\rd=\sA_3, &
\ld\sA_2,\sA_3\rd=\sA_1, &
\ld\sA_3,\sA_1\rd=\sA_2,\\[0.1cm]
\ld\sB_1,\sB_2\rd=-\sA_3, & \ld\sB_2,\sB_3\rd=-\sA_1, &
\ld\sB_3,\sB_1\rd=-\sA_2,\\[0.1cm]
\ld\sA_1,\sB_1\rd=0, & \ld\sA_2,\sB_2\rd=0, &
\ld\sA_3,\sB_3\rd=0,\\[0.1cm]
\ld\sA_1,\sB_2\rd=\sB_3, & \ld\sA_1,\sB_3\rd=-\sB_2, & \\[0.1cm]
\ld\sA_2,\sB_3\rd=\sB_1, & \ld\sA_2,\sB_1\rd=-\sB_3, & \\[0.1cm]
\ld\sA_3,\sB_1\rd=\sB_2, & \ld\sA_3,\sB_2\rd=-\sB_1. &
\end{array}\right\}
\end{equation}
Defining the operators
\begin{gather}
\sX_l=\frac{1}{2}i(\sA_l+i\sB_l),\quad\sY_l=\frac{1}{2}i(\sA_l-i\sB_l),
\label{SL25}\\
(l=1,2,3),\nonumber
\end{gather}
we come to a \textit{complex envelope} of the group algebra $\mathfrak{sl}(2,\C)$.
Using the relations (\ref{Com1}), we find
\begin{equation}\label{Com2}
\ld\sX_k,\sX_l\rd=i\varepsilon_{klm}\sX_m,\quad
\ld\sY_l,\sY_m\rd=i\varepsilon_{lmn}\sY_n,\quad \ld\sX_l,\sY_m\rd=0.
\end{equation}
From the relations (\ref{Com2}) it follows that each of the sets of
infinitesimal operators $\sX$ and $\sY$ generates the group $\SU(2)$
and these two groups commute with each other. Thus, from the
relations (\ref{Com2}) and (\ref{Waerden}) it follows that the group
$\SL(2,\C)$, in essence, is equivalent locally to the group
$\SU(2)\otimes\SU(2)$.


As is known, the most important \textit{observable} in quantum mechanics is the \textbf{\textit{energy}}. Let $H$ be a Hermitian \textit{energy operator} $H$\footnote{\textit{Spectrum} of this operator forms a some set of real numbers.}, defined on the Hilbert space $\sH_\infty$\footnote{We suppose that this space is separable.}. Then all the possible values of energy are eigenvalues of the operator $H$. At this point, if $E^\prime\neq E^{\prime\prime}$ are eigenvalues of $H$, and $|\psi^\prime\rangle$, $|\psi^{\prime\prime}\rangle$ are their eigenvectors in the space $\sH_\infty$, then $\langle\psi^\prime|\psi^{\prime\prime}\rangle=0$. All the eigenvectors, belonging to a given eigenvalue $E$, form (together with the null vector) an \textit{eigenvector subspace} $\sH_E$ of the Hilbert space $\sH_\infty$. All the eigenvector subspaces $\sH_E\in\sH_\infty$ are finite-dimensional. A dimension $r$ of $\sH_E$ is called a \textit{multiplicity} of the eigenvalue $E$. When $r>1$ the eigenvalue $E$ is \textit{$r$-fold degenerate}. As is known \cite{BHJ26}, the energy operator $H$ commutes with the all operators in $\sH_\infty$, representing a Lie algebra of the group $\SU(2)$. Therefore, $H$ commutes also with the all operators in $\sH_\infty$, which represent a group algebra of $\SU(2)\otimes\SU(2)$.

Let us consider an arbitrary eigenvector subspace $\sH_E$ of the energy operator $H$. Since the operators $\sX_l$, $\sY_l$ and $H$ commute with the each other, then, as is known \cite{Dir}, for these operators we can build a common system of eigenfunctions. It means that the subspace $\sH_E$ is invariant with respect to operators $\sX_l$, $\sY_l$\footnote{Moreover, the operators $\sX_l$, $\sY_l$ can be considered \textit{only on} $\sH_E$.}. The each operator from $\sX_l$, $\sY_l$, for example, $\sX_3$ (or $\sY_3$) is a Hermitian operator on $\sH_E$ and, therefore, $\sX_3$ has $r_E$ independent eigenvectors on the space $\sH_E$.

Further, let us suppose that there is a some \textit{local representation} of the group $\SU(2)\otimes\SU(2)$, defined by the operators acting in the space $\sH_\infty$. At this point, we assume that all the representing operators commute with $H$. Then the each eigenvector subspace $\sH_E$ of the energy operator is invariant with respect to operators of complex momentum $\sX_l$, $\sY_l$. On the Fig.\,1 eigenvector subspaces $\sH_E\simeq\Sym_{(k,r)}$ are placed in accordance with interlocking representation scheme of $\SL(2,\C)$. Here we have the first cell of spinorial chessboard of second order (for more details see \cite{Var15a}).
\begin{figure}[ht]
\unitlength=1.5mm
\begin{center}
\begin{picture}(100,50)
\put(50,0){$\overset{(0,0)}{\bullet}$}\put(47,5.5){\line(1,0){10}}\put(52.25,2.75){\line(0,1){7.25}}
\put(55,5){$\overset{(\frac{1}{2},0)}{\bullet}$}
\put(45,5){$\overset{(0,\frac{1}{2})}{\bullet}$}
\put(40,10){$\overset{(0,1)}{\bullet}$}\put(42,10.5){\line(1,0){10}}\put(47.25,7.75){\line(0,1){7.25}}
\put(50,10){$\overset{(\frac{1}{2},\frac{1}{2})}{\bullet}$}
\put(52,10.5){\line(1,0){10}}\put(57.25,7.75){\line(0,1){7.25}}
\put(60,10){$\overset{(1,0)}{\bullet}$}
\put(35,15){$\overset{(0,\frac{3}{2})}{\bullet}$}\put(37,15.5){\line(1,0){10}}\put(42.25,12.75){\line(0,1){7.25}}
\put(45,15){$\overset{(\frac{1}{2},1)}{\bullet}$}\put(47,15.5){\line(1,0){10}}\put(52.25,12.75){\line(0,1){7.25}}
\put(55,15){$\overset{(1,\frac{1}{2})}{\bullet}$}\put(57,15.5){\line(1,0){10}}\put(62.25,12.75){\line(0,1){7.25}}
\put(65,15){$\overset{(\frac{3}{2},0)}{\bullet}$}
\put(30,20){$\overset{(0,2)}{\bullet}$}\put(32,20.5){\line(1,0){10}}\put(37.25,17.75){\line(0,1){7.25}}
\put(40,20){$\overset{(\frac{1}{2},\frac{3}{2})}{\bullet}$}
\put(42,20.5){\line(1,0){10}}\put(47.25,17.75){\line(0,1){7.25}}
\put(50,20){$\overset{(1,1)}{\bullet}$}\put(52,20.5){\line(1,0){10}}\put(57.25,17.75){\line(0,1){7.25}}
\put(60,20){$\overset{(\frac{3}{2},\frac{1}{2})}{\bullet}$}
\put(62,20.5){\line(1,0){10}}\put(67.25,17.75){\line(0,1){7.25}}
\put(70,20){$\overset{(2,0)}{\bullet}$}
\put(25,25){$\overset{(0,\frac{5}{2})}{\bullet}$}\put(27,25.5){\line(1,0){10}}\put(32.25,22.75){\line(0,1){7.25}}
\put(35,25){$\overset{(\frac{1}{2},2)}{\bullet}$}\put(37,25.5){\line(1,0){10}}\put(42.25,22.75){\line(0,1){7.25}}
\put(45,25){$\overset{(1,\frac{3}{2})}{\bullet}$}\put(47,25.5){\line(1,0){10}}\put(52.25,22.75){\line(0,1){7.25}}
\put(55,25){$\overset{(\frac{3}{2},1)}{\bullet}$}\put(57,25.5){\line(1,0){10}}\put(62.25,22.75){\line(0,1){7.25}}
\put(65,25){$\overset{(2,\frac{1}{2})}{\bullet}$}\put(67,25.5){\line(1,0){10}}\put(72.25,22.75){\line(0,1){7.25}}
\put(75,25){$\overset{(\frac{5}{2},0)}{\bullet}$}
\put(20,30){$\overset{(0,3)}{\bullet}$}\put(22,30.5){\line(1,0){10}}\put(27.25,27.75){\line(0,1){7.25}}
\put(30,30){$\overset{(\frac{1}{2},\frac{5}{2})}{\bullet}$}
\put(32,30.5){\line(1,0){10}}\put(37.25,27.75){\line(0,1){7.25}}
\put(40,30){$\overset{(1,2)}{\bullet}$}\put(42,30.5){\line(1,0){10}}\put(47.25,27.75){\line(0,1){7.25}}
\put(50,30){$\overset{(\frac{3}{2},\frac{3}{2})}{\bullet}$}
\put(52,30.5){\line(1,0){10}}\put(57.25,27.75){\line(0,1){7.25}}
\put(60,30){$\overset{(2,1)}{\bullet}$}\put(62,30.5){\line(1,0){10}}\put(67.25,27.75){\line(0,1){7.25}}
\put(70,30){$\overset{(\frac{5}{2},\frac{5}{2})}{\bullet}$}
\put(72,30.5){\line(1,0){10}}\put(77.25,27.75){\line(0,1){7.25}}
\put(80,30){$\overset{(3,0)}{\bullet}$}
\put(15,35){$\overset{(0,\frac{7}{2})}{\bullet}$}\put(17,35.5){\line(1,0){10}}\put(22.25,32.75){\line(0,1){7.25}}
\put(25,35){$\overset{(\frac{1}{2},3)}{\bullet}$}\put(27,35.5){\line(1,0){10}}\put(32.25,32.75){\line(0,1){7.25}}
\put(35,35){$\overset{(1,\frac{5}{2})}{\bullet}$}\put(37,35.5){\line(1,0){10}}\put(42.25,32.75){\line(0,1){7.25}}
\put(45,35){$\overset{(\frac{3}{2},2)}{\bullet}$}\put(47,35.5){\line(1,0){10}}\put(52.25,32.75){\line(0,1){7.25}}
\put(55,35){$\overset{(2,\frac{3}{2})}{\bullet}$}\put(57,35.5){\line(1,0){10}}\put(62.25,32.75){\line(0,1){7.25}}
\put(65,35){$\overset{(\frac{5}{2},1)}{\bullet}$}\put(67,35.5){\line(1,0){10}}\put(72.25,32.75){\line(0,1){7.25}}
\put(75,35){$\overset{(3,\frac{1}{2})}{\bullet}$}\put(77,35.5){\line(1,0){10}}\put(82.25,32.75){\line(0,1){7.25}}
\put(85,35){$\overset{(\frac{7}{2},0)}{\bullet}$}
\put(10,40){$\overset{(0,4)}{\bullet}$}\put(12,40.5){\line(1,0){10}}
\put(20,40){$\overset{(\frac{1}{2},\frac{7}{2})}{\bullet}$}\put(22,40.5){\line(1,0){10}}
\put(30,40){$\overset{(1,3)}{\bullet}$}\put(32,40.5){\line(1,0){10}}
\put(40,40){$\overset{(\frac{3}{2},\frac{5}{2})}{\bullet}$}\put(42,40.5){\line(1,0){10}}
\put(50,40){$\overset{(2,2)}{\bullet}$}\put(52,40.5){\line(1,0){10}}
\put(60,40){$\overset{(\frac{5}{2},\frac{3}{2})}{\bullet}$}\put(62,40.5){\line(1,0){10}}
\put(70,40){$\overset{(3,1)}{\bullet}$}\put(72,40.5){\line(1,0){10}}
\put(80,40){$\overset{(\frac{7}{2},\frac{1}{2})}{\bullet}$}\put(82,40.5){\line(1,0){10}}
\put(90,40){$\overset{(4,0)}{\bullet}$}
\put(11.5,45){$\vdots$}
\put(21.5,45){$\vdots$}
\put(31.5,45){$\vdots$}
\put(41.5,45){$\vdots$}
\put(51.5,45){$\vdots$}
\put(61.5,45){$\vdots$}
\put(71.5,45){$\vdots$}
\put(81.5,45){$\vdots$}
\put(91.5,45){$\vdots$}
\put(10,0.5){\line(1,0){42}}\put(50,0.5){\vector(1,0){42}}
\put(16.5,32){$\vdots$}
\put(16.5,29){$\vdots$}
\put(16.5,26){$\vdots$}
\put(16.5,23){$\vdots$}
\put(16.5,20){$\vdots$}
\put(16.5,17){$\vdots$}
\put(16.5,14){$\vdots$}
\put(16.5,11){$\vdots$}
\put(16.5,9){$\vdots$}
\put(16.5,6){$\vdots$}
\put(16.5,3){$\vdots$}
\put(16.5,1.5){$\cdot$}
\put(16.5,0){$\cdot$}
\put(14.5,-3){$-\frac{7}{2}$}
\put(21.5,27){$\vdots$}
\put(21.5,24){$\vdots$}
\put(21.5,21){$\vdots$}
\put(21.5,18){$\vdots$}
\put(21.5,15){$\vdots$}
\put(21.5,13){$\vdots$}
\put(21.5,9){$\vdots$}
\put(21.5,6){$\vdots$}
\put(21.5,3){$\vdots$}
\put(21.5,1.5){$\cdot$}
\put(21.5,0){$\cdot$}
\put(19.5,-3){$-3$}
\put(26.5,22){$\vdots$}
\put(26.5,19){$\vdots$}
\put(26.5,16){$\vdots$}
\put(26.5,13){$\vdots$}
\put(26.5,10){$\vdots$}
\put(26.5,7){$\vdots$}
\put(26.5,4){$\vdots$}
\put(26.5,1){$\vdots$}
\put(24.5,-3){$-\frac{5}{2}$}
\put(31.5,17){$\vdots$}
\put(31.5,14){$\vdots$}
\put(31.5,11){$\vdots$}
\put(31.5,8){$\vdots$}
\put(31.5,5){$\vdots$}
\put(31.5,2){$\vdots$}
\put(31.5,0.5){$\cdot$}
\put(29.5,-3){$-2$}
\put(36.5,12){$\vdots$}
\put(36.5,9){$\vdots$}
\put(36.5,6){$\vdots$}
\put(36.5,3){$\vdots$}
\put(36.5,1.5){$\cdot$}
\put(36.5,0){$\cdot$}
\put(34.5,-3){$-\frac{3}{2}$}
\put(41.5,7){$\vdots$}
\put(41.5,4){$\vdots$}
\put(41.5,1){$\vdots$}
\put(39.5,-3){$-1$}
\put(46.5,2){$\vdots$}
\put(46.5,0.5){$\cdot$}
\put(44.5,-3){$-\frac{1}{2}$}
\put(51.5,-3){$0$}
\put(56.5,2){$\vdots$}
\put(56.5,0.5){$\cdot$}
\put(56.5,-3){$\frac{1}{2}$}
\put(61.5,7){$\vdots$}
\put(61.5,4){$\vdots$}
\put(61.5,1){$\vdots$}
\put(61.5,-3){$1$}
\put(66.5,12){$\vdots$}
\put(66.5,9){$\vdots$}
\put(66.5,6){$\vdots$}
\put(66.5,3){$\vdots$}
\put(66.5,1.5){$\cdot$}
\put(66.5,0){$\cdot$}
\put(66.5,-3){$\frac{3}{2}$}
\put(71.5,17){$\vdots$}
\put(71.5,14){$\vdots$}
\put(71.5,11){$\vdots$}
\put(71.5,8){$\vdots$}
\put(71.5,5){$\vdots$}
\put(71.5,2){$\vdots$}
\put(71.5,0.5){$\cdot$}
\put(71.5,-3){$2$}
\put(76.5,22){$\vdots$}
\put(76.5,19){$\vdots$}
\put(76.5,16){$\vdots$}
\put(76.5,13){$\vdots$}
\put(76.5,10){$\vdots$}
\put(76.5,7){$\vdots$}
\put(76.5,4){$\vdots$}
\put(76.5,1){$\vdots$}
\put(76.5,-3){$\frac{5}{2}$}
\put(81.5,27){$\vdots$}
\put(81.5,24){$\vdots$}
\put(81.5,21){$\vdots$}
\put(81.5,18){$\vdots$}
\put(81.5,15){$\vdots$}
\put(81.5,13){$\vdots$}
\put(81.5,9){$\vdots$}
\put(81.5,6){$\vdots$}
\put(81.5,3){$\vdots$}
\put(81.5,1.5){$\cdot$}
\put(81.5,0){$\cdot$}
\put(81.5,-3){$3$}
\put(86.5,32){$\vdots$}
\put(86.5,29){$\vdots$}
\put(86.5,26){$\vdots$}
\put(86.5,23){$\vdots$}
\put(86.5,20){$\vdots$}
\put(86.5,17){$\vdots$}
\put(86.5,14){$\vdots$}
\put(86.5,11){$\vdots$}
\put(86.5,9){$\vdots$}
\put(86.5,6){$\vdots$}
\put(86.5,3){$\vdots$}
\put(86.5,1.5){$\cdot$}
\put(86.5,0){$\cdot$}
\put(86.5,-3){$\frac{7}{2}$}
\put(53.8,1.7){$\cdot$}\put(54.3,2.2){$\cdot$}\put(54.8,2.7){$\cdot$}\put(55.3,3.3){$\cdot$}\put(55.8,3.8){$\cdot$}
\put(56.3,4.3){$\cdot$}
\put(58.8,6.8){$\cdot$}\put(59.3,7.3){$\cdot$}\put(59.8,7.8){$\cdot$}\put(60.3,8.3){$\cdot$}\put(60.8,8.8){$\cdot$}
\put(61.3,9.3){$\cdot$}
\put(63.8,11.8){$\cdot$}\put(64.3,12.3){$\cdot$}\put(64.8,12.8){$\cdot$}\put(65.3,13.3){$\cdot$}\put(65.8,13.8){$\cdot$}
\put(66.3,14.3){$\cdot$}
\put(68.8,16.8){$\cdot$}\put(69.3,17.3){$\cdot$}\put(69.8,17.8){$\cdot$}\put(70.3,18.3){$\cdot$}\put(70.8,18.8){$\cdot$}
\put(71.3,19.3){$\cdot$}
\put(33.8,21.8){$\cdot$}\put(34.3,22.3){$\cdot$}\put(34.8,22.8){$\cdot$}\put(35.3,23.3){$\cdot$}\put(35.8,23.8){$\cdot$}
\put(36.3,24.3){$\cdot$}
\put(38.8,26.8){$\cdot$}\put(39.3,27.3){$\cdot$}\put(39.8,27.8){$\cdot$}\put(40.3,28.3){$\cdot$}\put(40.8,28.8){$\cdot$}
\put(41.3,29.3){$\cdot$}
\put(43.8,31.8){$\cdot$}\put(44.3,32.3){$\cdot$}\put(44.8,32.8){$\cdot$}\put(45.3,33.3){$\cdot$}\put(45.8,33.8){$\cdot$}
\put(46.3,34.3){$\cdot$}
\put(48.8,36.8){$\cdot$}\put(49.3,37.3){$\cdot$}\put(49.8,37.8){$\cdot$}\put(50.3,38.3){$\cdot$}\put(50.8,38.8){$\cdot$}
\put(51.3,39.3){$\cdot$}
\put(47.3,4.4){$\cdot$}\put(47.8,3.9){$\cdot$}\put(48.3,3.4){$\cdot$}\put(48.8,2.9){$\cdot$}\put(49.3,2.4){$\cdot$}
\put(49.8,1.9){$\cdot$}
\put(42.3,9.4){$\cdot$}\put(42.8,8.9){$\cdot$}\put(43.3,8.4){$\cdot$}\put(43.8,7.9){$\cdot$}\put(44.3,7.4){$\cdot$}
\put(44.8,6.9){$\cdot$}
\put(37.3,14.4){$\cdot$}\put(37.8,13.9){$\cdot$}\put(38.3,13.4){$\cdot$}\put(38.8,12.9){$\cdot$}\put(39.3,12.4){$\cdot$}
\put(39.8,11.9){$\cdot$}
\put(32.3,19.4){$\cdot$}\put(32.8,18.9){$\cdot$}\put(33.3,18.4){$\cdot$}\put(33.8,17.9){$\cdot$}\put(34.3,17.4){$\cdot$}
\put(34.8,16.9){$\cdot$}
\put(67.3,24.4){$\cdot$}\put(67.8,23.9){$\cdot$}\put(68.3,23.4){$\cdot$}\put(68.8,22.9){$\cdot$}\put(69.3,22.4){$\cdot$}
\put(69.8,21.9){$\cdot$}
\put(62.3,29.4){$\cdot$}\put(62.8,28.9){$\cdot$}\put(63.3,28.4){$\cdot$}\put(63.8,27.9){$\cdot$}\put(64.3,27.4){$\cdot$}
\put(64.8,26.9){$\cdot$}
\put(57.3,34.4){$\cdot$}\put(57.8,33.9){$\cdot$}\put(58.3,33.4){$\cdot$}\put(58.8,32.9){$\cdot$}\put(59.3,32.4){$\cdot$}
\put(59.8,31.9){$\cdot$}
\put(52.3,39.4){$\cdot$}\put(52.8,38.9){$\cdot$}\put(53.3,38.4){$\cdot$}\put(53.8,37.9){$\cdot$}\put(54.3,37.4){$\cdot$}
\put(54.8,36.9){$\cdot$}
\end{picture}
\end{center}
\vspace{0.3cm}
\begin{center}\begin{minipage}{30pc}{\small {\bf Fig.\,1:} Eigenvector subspaces $\sH_E\simeq\Sym_{(k,r)}$ of the energy operator $H$. The each subspace $\sH_E$ (level of matter spectrum) is a space of irreducible representation $\boldsymbol{\tau}_{k/2,r/2}$ belonging to a system of interlocking representations of the Lorentz group. The first cell of spinorial chessboard of second order is marked by dotted lines.}\end{minipage}\end{center}
\end{figure}

Further, introducing operators of the form (`rising' and `lowering' operators of the group $\SL(2,\C)$)
\begin{equation}\label{SL26}
\left.\begin{array}{cc}
\sX_+=\sX_1+i\sX_2, & \sX_-=\sX_1-i\sX_2,\\[0.1cm]
\sY_+=\sY_1+i\sY_2, & \sY_-=\sY_1-i\sY_2,
\end{array}\right\}
\end{equation}
we see that
\[
\ld\sX_3,\sX_+\rd=\sX_+,\quad\ld\sX_3,\sX_-\rd=-\sX_-,\quad\ld\sX_+,\sX_-\rd=2\sX_3,
\]
\[
\ld\sY_3,\sY_+\rd=\sY_+,\quad\ld\sY_3,\sY_-\rd=-\sY_-,\quad\ld\sY_+,\sY_-\rd=2\sY_3.
\]
In virtue of commutativity of the relations (\ref{Com2}) a space of an irreducible finite-dimensional representation of the group $\SL(2,\C)$ can be spanned on the totality of
$(2l+1)(2\dot{l}+1)$ basis ket-vectors $|l,m;\dot{l},\dot{m}\rangle$ and basis bra-vectors
$\langle l,m;\dot{l},\dot{m}|$, where $l,m,\dot{l},\dot{m}$ are integer
or half-integer numbers, $-l\leq m\leq l$, $-\dot{l}\leq
\dot{m}\leq \dot{l}$. Therefore,
\begin{eqnarray}
&&\sX_-|l,m;\dot{l},\dot{m}\rangle= \sqrt{(l+m)(l-m+1)}|l,m-1;\dot{l},\dot{m}\rangle
\;\;(m>-l),\nonumber\\
&&\sX_+|l,m;\dot{l},\dot{m}\rangle= \sqrt{(l-m)(l+m+1)}|l,m+1;\dot{l},\dot{m}\rangle
\;\;(m<l),\nonumber\\
&&\sX_3|l,m;\dot{l},\dot{m}\rangle=
m|l,m;\dot{l},\dot{m}\rangle,\nonumber\\
&&\langle l,m;\dot{l},\dot{m}|\sY_-=
\langle l,m;\dot{l},\dot{m}-1|\sqrt{(\dot{l}+\dot{m})(\dot{l}-\dot{m}+1)}\;\;(\dot{m}>-\dot{l}),\nonumber\\
&&\langle l,m;\dot{l},\dot{m}|\sY_+=
\langle l,m;\dot{l},\dot{m}+1|\sqrt{(\dot{l}-\dot{m})(\dot{l}+\dot{m}+1)}\;\;(\dot{m}<\dot{l}),\nonumber\\
&&\langle l,m;\dot{l},\dot{m}|\sY_3= \langle l,m;\dot{l},\dot{m}|\dot{m}.\label{Waerden}
\end{eqnarray}

Thus, a complex envelope of the group algebra $\mathfrak{sl}(2,\C)$, generating complex momentum, leads to a \textit{duality} which is mirrored in the appearance of the two spaces: a space of ket-vectors $|l,m;\dot{l},\dot{m}\rangle$ and a dual space of bra-vectors $\langle l,m;\dot{l},\dot{m}|$.
\subsection{Spinor structure}
It is well-known \cite{GMS,RF} that spintensor representations of the
group $\SL(2;\C)\simeq\spin_+(1,3)$ form a ground of
all  finite-dimensional representations of the Lorentz group. We will consider a relationship of these representations with the complex Clifford
algebras. From the each complex Clifford algebra
$\C_n=\C\otimes\cl_{p,q}\; (n=p+q)$ we obtain a spinspace
$\dS_{2^{n/2}}$ which is a complexification of the minimal left
ideal of the algebra $\cl_{p,q}$: $\dS_{2^{n/2}}=\C\otimes
I_{p,q}=\C\otimes\cl_{p,q} f_{pq}$, where $f_{pq}$ is a primitive
idempotent of the algebra $\cl_{p,q}$. Further, a spinspace, related
with the biquaternion algebra $\C_2$, has the form $\dS_2=\C\otimes
I_{2,0}=\C\otimes\cl_{2,0}f_{20}$ or $\dS_2=\C\otimes
I_{1,1}=\C\otimes\cl_{1,1}f_{11}(\C\otimes I_{0,2}=
\C\otimes\cl_{0,2}f_{02})$. Therefore, the tensor product of $k$
algebras $\C_2$ induces the tensor product of $k$ spinspaces
$\dS_2$: $\dS_2\otimes\dS_2\otimes\cdots\otimes\dS_2=\dS_{2^k}$.
Vectors of the spinspace $\dS_{2^k}$ (or elements of the minimal
left ideal of $\C_{2k}$) are spintensors of the following form:
\begin{equation}\label{6.16}
\boldsymbol{s}^{\alpha_1\alpha_2\cdots\alpha_k}=\sum
\boldsymbol{s}^{\alpha_1}\otimes
\boldsymbol{s}^{\alpha_2}\otimes\cdots\otimes
\boldsymbol{s}^{\alpha_k},
\end{equation}
where summation is produced on all the index collections
$(\alpha_1\ldots\alpha_k)$, $\alpha_i=1,2$.

Further, let $\overset{\ast}{\C}_2$ be the biquaternion algebra with the
coefficients which are complex conjugate to the coefficients of
$\C_2$. The tensor product
$\overset{\ast}{\C}_2\otimes\overset{\ast}{\C}_2\otimes\cdots\otimes
\overset{\ast}{\C}_2\simeq\overset{\ast}{\C}_{2r}$ of the $r$
algebras $\overset{\ast}{\C}_2$ induces tensor product of the
$r$ spinspaces $\dot{\dS}_2$:
$\dot{\dS}_2\otimes\dot{\dS}_2\otimes\cdots\otimes\dot{\dS}_2=\dot{\dS}_{2^r}$. Vectors of the spinspace $\dot{\dS}_{2^r}$ have the form
\begin{equation}\label{6.18}
\boldsymbol{s}^{\dot{\alpha}_1\dot{\alpha}_2\cdots\dot{\alpha}_r}=\sum
\boldsymbol{s}^{\dot{\alpha}_1}\otimes
\boldsymbol{s}^{\dot{\alpha}_2}\otimes\cdots\otimes
\boldsymbol{s}^{\dot{\alpha}_r}.
\end{equation}

In general case we have a tensor product of the $k$ algebras $\C_2$
and the $r$ algebras $\overset{\ast}{\C}_2$:
\begin{equation}\label{TenAlg}
\underbrace{\C_2\otimes\C_2\otimes\cdots\otimes\C_2}_{k\;\text{times}}\bigotimes
\underbrace{\overset{\ast}{\C}_2\otimes\overset{\ast}{\C}_2\otimes\cdots\otimes
\overset{\ast}{\C}_2}_{r\;\text{times}}\simeq
\C_{2k}\otimes\overset{\ast}{\C}_{2r},
\end{equation}
which induces a spinspace
\begin{equation}\label{SpinSpace}
\underbrace{\dS_2\otimes\dS_2\otimes\cdots\otimes\dS_2}_{k\;\text{times}}\bigotimes
\underbrace{\dot{\dS}_2\otimes\dot{\dS}_2\otimes\cdots\otimes\dot{\dS}_2}_{r\;\text{times}}
=\dS_{2^{k+r}}
\end{equation}
with the vectors
\[
\boldsymbol{S}=\boldsymbol{s}^{\alpha_1\alpha_2\ldots\alpha_k\dot{\alpha}_1\dot{\alpha}_2\ldots
\dot{\alpha}_r}=\sum \boldsymbol{s}^{\alpha_1}\otimes
\boldsymbol{s}^{\alpha_2}\otimes\cdots\otimes
\boldsymbol{s}^{\alpha_k}\otimes
\boldsymbol{s}^{\dot{\alpha}_1}\otimes
\boldsymbol{s}^{\dot{\alpha}_2}\otimes\cdots\otimes
\boldsymbol{s}^{\dot{\alpha}_r}.
\]
Further, any pair of substitutions
\[
\alpha=\begin{pmatrix} 1 & 2 & \ldots & k\\
\alpha_1 & \alpha_2 & \ldots & \alpha_k\end{pmatrix},\quad \beta=\begin{pmatrix} 1 & 2 & \ldots & r\\
\dot{\alpha}_1 & \dot{\alpha}_2 & \ldots & \dot{\alpha}_r\end{pmatrix}
\]
defines a transformation $(\alpha,\beta)$ mapping $\boldsymbol{S}$ to the following polynomial:
\[
P_{\alpha\beta}\boldsymbol{S}=\boldsymbol{s}^{\alpha\left(\alpha_1\right)\alpha\left(\alpha_2\right)\ldots
\alpha\left(\alpha_k\right)\beta\left(\dot{\alpha}_1\right)\beta\left(\dot{\alpha}_2\right)\ldots
\beta\left(\dot{\alpha}_r\right)}.
\]
The spintensor $\boldsymbol{S}$ is called a \emph{symmetric spintensor} if at any $\alpha$, $\beta$ the equality
\[
P_{\alpha\beta}\boldsymbol{S}=\boldsymbol{S}
\]
holds. The space $\Sym_{(k,r)}$ of symmetric spintensors has the dimensionality
\begin{equation}\label{Degree}
\dim\Sym_{(k,r)}=(k+1)(r+1).
\end{equation}
The dimensionality of $\Sym_{(k,r)}$ is called a \emph{degree of the representation} $\boldsymbol{\tau}_{l\dot{l}}$ of the group $\SL(2,\C)$. It is easy to see that for the group $\SL(2,\C)$ there are representations of \textbf{\emph{any degree}} (in contrast to $\SU(3)$, $\SU(6)$ and other groups of internal symmetries, see \cite{Var15}). For the each $A\in\SL(2,\C)$ we define a linear transformation of the spintensor $\boldsymbol{s}$ via the formula
\begin{equation}\label{SPT}
\boldsymbol{s}^{\alpha_1\alpha_2\ldots\alpha_k\dot{\alpha}_1\dot{\alpha}_2\ldots
\dot{\alpha}_r}\longrightarrow\sum_{\left(\beta\right)\left(\dot{\beta}\right)}A^{\alpha_1\beta_1}A^{\alpha_2\beta_2}
\cdots A^{\alpha_k\beta_k}\overline{A}^{\dot{\alpha}_1\dot{\beta}_1}\overline{A}^{\dot{\alpha}_2\dot{\beta}_2}
\cdots\overline{A}^{\dot{\alpha}_r\dot{\beta}_r}\boldsymbol{s}^{\beta_1\beta_2\ldots\beta_k\dot{\beta}_1\dot{\beta}_2\ldots
\dot{\beta}_r},
\end{equation}
where the symbols $\left(\beta\right)$ and $\left(\dot{\beta}\right)$ mean $\beta_1$, $\beta_2$, $\ldots$, $\beta_k$ and $\dot{\beta}_1$, $\dot{\beta}_2$, $\ldots$, $\dot{\beta}_r$. Spintensor representations of $\SL(2,\C)$, defined by (\ref{SPT}), act in the spinspace $\dS_{2^{k+r}}$ of dimension $2^{k+r}$. As a rule, these representations are reducible into symmetric and antisymmetric parts. We will separate subspaces $\Sym_{(k,r)}\subset\dS_{2^{k+r}}$ of symmetric spintensors. Representations of $\SL(2,\C)$ in $\Sym_{(k,r)}$ form a \textit{full} system of irreducible representations of this group. These representations of $\SL(2,\C)$ we denote as $\boldsymbol{\tau}_{\frac{k}{2},\frac{r}{2}}=\boldsymbol{\tau}_{l\dot{l}}$. The each \emph{irreducible} finite dimensional representation of $\SL(2,\C)$ is equivalent to one from $\boldsymbol{\tau}_{k/2,r/2}$. Further, tensor products (\ref{TenAlg}) and (\ref{SpinSpace}) generate
\begin{equation}\label{TenRep}
\underbrace{\boldsymbol{\tau}_{\frac{1}{2},0}\otimes\boldsymbol{\tau}_{\frac{1}{2},0}\otimes\cdots\otimes
\boldsymbol{\tau}_{\frac{1}{2},0}}_{k\;\text{times}}\bigotimes
\underbrace{\boldsymbol{\tau}_{0,\frac{1}{2}}\otimes\boldsymbol{\tau}_{0,\frac{1}{2}}\otimes\cdots\otimes
\boldsymbol{\tau}_{0,\frac{1}{2}}}_{r\;\text{times}}\simeq\boldsymbol{\tau}_{\frac{k}{2},0}\otimes
\boldsymbol{\tau}_{0,\frac{r}{2}}=\boldsymbol{\tau}_{\frac{k}{2},\frac{r}{2}}
\end{equation}
and
\begin{equation}\label{TenSym}
\underbrace{\Sym_{(1,0)}\otimes\Sym_{(1,0)}\otimes\cdots\otimes\Sym_{(1,0)}}_{k\;\text{times}}\bigotimes
\underbrace{\Sym_{(0,1)}\otimes\Sym_{(0,1)}\otimes\cdots\otimes\Sym_{(0,1)}}_{r\;\text{times}}=\Sym_{(k,r)}.
\end{equation}

When the matrices $A$ are unitary and unimodular we come to the subgroup $\SU(2)$ of $\SL(2,\C)$. Irreducible representations of $\SU(2)$ are equivalent to one from the mappings $A\rightarrow\boldsymbol{\tau}_{k/2,0}(A)$ with $A\in\SU(2)$, they are denoted as $\boldsymbol{\tau}_{k/2}$. The representation of $\SU(2)$, obtained at the contraction $A\rightarrow\boldsymbol{\tau}_{k/2,r/2}(A)$ onto $A\in\SU(2)$, is not irreducible. In fact, it is a direct product of $\boldsymbol{\tau}_{k/2}$ by $\boldsymbol{\tau}_{r/2}$, therefore, in virtue of the Clebsh-Gordan decomposition we have here a sum of representations
\[
\boldsymbol{\tau}_{\frac{k+r}{2}},\quad\boldsymbol{\tau}_{\frac{k+r}{2}-1},\quad\ldots,\quad
\boldsymbol{\tau}_{\left|\frac{k-r}{2}\right|}.
\]

\subsection{Description of charge. Pseudoautomorphism of spinor structure}
In the present form of quantum theory complex fields correspond to charged particles \cite{WWW52}.
As a spin, charge is an intrinsic notion belonging to the level of spinor structure (spintensor substrate). Pre-image of charge conjugation $C$ on this level is a pseudoautomorphism $\cA\rightarrow\overline{\cA}$ of spinor structure \cite{Var01a,Var04,Var14}. Tensor substrate (\ref{TenAlg}) corresponds to a \textit{complex spinor structure}.
In accordance with \cite{Var01a}, Clifford algebras over the field $\F=\C$ correspond to \textbf{\emph{charged particles}}, such as electron, proton and so on. Coming to a \textit{real spinor structure}, which, obviously, is a substructure of the substrate (\ref{TenAlg}), we see that transformation $C$ (pseudoautomorphism
$\cA\rightarrow\overline{\cA}$) for the algebras $\cl_{p,q}$ over the field $\F=\R$ and the ring $\K\simeq\BH$ (the types $p-q\equiv
4,6\pmod{8}$) is reduced to \textbf{\emph{particle-antiparticle interchange}} $C^\prime$ (see \cite{Var01a,Var04}). As is known, neutral particles are described within real representations of the Lorentz group. There are two classes of neutral particles: 1) particles which have antiparticles, such as neutrons, neutrino\footnote{However, it should be noted that the question whether neutrinos are Dirac or Majorana particles (truly neutral fermions) is still open (the last hypothesis being preferred by particle physicists).} and so on; 2) particles which coincide with their antiparticles (for example, photons, $\pi^0$-mesons and so on). The first class is described by the algebras
$\cl_{p,q}$ over the field $\F=\R$ with the rings $\K\simeq\BH$ and
$\K\simeq\BH\oplus\BH$ (the types $p-q\equiv 4,6\pmod{8}$ and
$p-q\equiv 5\pmod{8}$), and second class (\textbf{\emph{truly neutral particles}}) is described by the algebras $\cl_{p,q}$ over the field $\F=\R$ with the rings $\K\simeq\R$ and
$\K\simeq\R\oplus\R$ (the types $p-q\equiv 0,2\pmod{8}$ and
$p-q\equiv 1\pmod{8}$) (for more details see
\cite{Var01a,Var04,Var05,Var11,Var14}).

\subsection{Definition of elementary particle}
Hilbert spaces $\bsH^S_{2s+1}\otimes\bsH_\infty$ (spin multiplets) and $\bsH^Q\otimes\bsH_\infty$ (charge multiplets), considered in \cite{Var15}, lead naturally to the following generalization of the abstract Hilbert space. Let
\begin{equation}\label{SQHil}
\bsH^S\otimes\bsH^Q\otimes\bsH_\infty
\end{equation}
be a tensor product of $\bsH_\infty$ and a \emph{spin-charge space} $\bsH^S\otimes\bsH^Q$. State vectors of the space (\ref{SQHil}) describe actualized particle states\footnote{In accordance with Heisenberg-Fock conception \cite{Heisen,Fock}, reality has a two-level structure: \textit{potential reality} and \textit{actual reality}. Heisenberg claimed that any quantum object (for example, elementary particle) belongs to both sides of reality: at first, to potential reality as a superposition, and to actual reality after reduction of superposition (for more details see \cite{Heisen,Fock,Var15b}).} of the spin $s=|l-\dot{l}|$ and charge $Q$ with the mass $m$. Moreover, state vectors are grouped into spin lines: spin-0 line, spin-1/2 line, spin-1 line and so on (bosonic and fermionic lines). The each state vector of $\bsH^S\otimes\bsH^Q\otimes\bsH_\infty$ presents in itself an irreducible representation $\boldsymbol{\tau}_{l\dot{l}}$ of the group $\spin_+(1,3)$ in the space $\Sym_{(k,r)}\simeq\sH_E$ (an eigenvector subspace of the energy operator $H$). Charge $Q$ takes three values $-1$, $0$, $+1$, where the values $-1$, $+1$ correspond to actualized states of charged particles, and the value $0$ corresponds to neutral (or truly neutral) particle states. On the level of underlying spinor structure charged particles are described by the complex representations of $\spin_+(1,3)$, for which the pseudoatomorphism $\cA\rightarrow\overline{\cA}$ is not trivial ($\F=\C$) and an action of $\cA\rightarrow\overline{\cA}$ replaces complex representations (charge state $-1$) by complex conjugate representations (charge state $+1$). The neutral particles (charge state $0$) are described by the real representations of $\spin_+(1,3)$, for which the transformation $\cA\rightarrow\overline{\cA}$ is also not trivial ($\F=\R$, $\K\simeq\BH$), that is, we have here particle-antiparticle interchange. In turn, truly neutral particles are described by the real representations of $\spin_+(1,3)$ for which the action of the pseudoautomorphism $\cA\rightarrow\overline{\cA}$ is trivial ($\F=\R$, $\K\simeq\R$). With the aim to distinguish this case from the neutral particles (state 0) we denote this charge state as $\overline{0}$. Therefore, the spinor structure with the help of $\cA\rightarrow\overline{\cA}$ allows us to separate real representations for neutral (charge state 0) and truly neutral (charge state $\overline{0}$) particles.
Consideration of charge allows us to give the following definition of elementary particle. \textbf{\textit{Particle is a superposition of state vectors in nonseparable Hilbert space  $\bsH^S\otimes\bsH^Q\otimes\bsH_\infty$}}\footnote{A superposition of the state vectors in $\bsH^S\otimes\bsH^Q\otimes\bsH_\infty$ can be considered as a unitary  (infinite-dimensional) representation of some group (a group of fundamental symmetry). For example, in accordance with \cite{Wig39}, a quantum system, described by irreducible unitary representation of the Poincar\'{e} group, is an elementary particle (Wigner interpretation). Of course, the compact groups $\SU(N)$ (groups of secondary dynamical symmetries) can not be considered as groups of fundamental symmetries, since they do not have infinite-dimensional representations. As a fundamental symmetry, except Lorentz (Poincar\'{e}) group, we can take conformal group $\SO_0(2,4)$, and also a Rumer-Fet group $\SU(2,2)\otimes\SU(2)\otimes\SU(2)$ \cite{Fet}, where $\SU(2,2)$ is a universal covering of $\SO_0(2,4)$. All these groups contain Lorentz group as a subgroup.}.
State vectors of the space $\bsH^S\otimes\bsH^Q\otimes\bsH_\infty$ have the form
\begin{equation}\label{VectH}
\left|\boldsymbol{\Psi}\right\rangle=
\left.\left|\boldsymbol{\tau}_{l\dot{l}},\,\Sym_{(k,r)},\,\cl_{p,q},\,\dS_{(p+q)/2},\,C^{a,b,c,d,e,f,g},\,
\right.\ldots\right\rangle,
\end{equation}
where $\boldsymbol{\tau}_{l\dot{l}}$ is a representation of the
proper orthocronous Lorentz group, $\Sym_{(k,r)}$ is a
representation space of $\boldsymbol{\tau}_{l\dot{l}}$ with the degree (\ref{Degree}), $\cl_{p,q}$
is a Clifford algebra associated with
$\boldsymbol{\tau}_{l\dot{l}}$, $\dS_{(p+q)/2}$ is a spinspace
associated with $\cl_{p,q}$, $C^{a,b,c,d,e,f,g}$ is a $CPT$ group
defined within $\cl_{p,q}$. It is obvious that the main object,
defining the vectors of the type (\ref{VectH}), is the
representation $\boldsymbol{\tau}_{l\dot{l}}$. Objects
$\cl_{p,q}$, $\dS_{(p+q)/2}$ and $C^{a,b,c,d,e,f,g}$ belong to spinor structure. In
other words, the representation $\boldsymbol{\tau}_{l\dot{l}}$
corresponds to the each vector of $\bsH^S\otimes\bsH^Q\otimes\bsH_\infty$, and a set of all possible
representations $\boldsymbol{\tau}_{l\dot{l}}$ generates the all
spin-charge Hilbert space $\bsH^S\otimes\bsH^Q\otimes\bsH_\infty$ with the vectors (\ref{VectH}).
\subsection{Relativistic wave equations}
In 1945, Bhabha \cite{Bha45} introduced relativistic wave equations
\begin{equation}\label{Bha}
i\Gamma_\mu\frac{\partial\psi}{\partial x_\mu}+m\psi=0,\quad\mu=0,1,2,3
\end{equation}
that describe systems with many masses and spins. The $\Gamma_\mu$ satisfy different commutation rules depending on the spin of the particle described ($\Gamma_\mu$ are square matrices of different degrees in each case). Bhabha claimed that the fundamental equations of the elementary particles must be first-order equations of the form (\ref{Bha}) and that all properties of the particles must be derivable from these without use of any further subsidiary conditions. It should be noted that Majorana \cite{Maj32} introduced relativistic wave equations for arbitrary spins of the same type in 1932 (including infinite-component equations), see excellent review \cite{Esp12}. In 1948, Gel'fand and Yaglom \cite{GY48} developed a general theory of such (Majorana-Bhabha) equations including infinite-component equations of Majorana type \cite{Maj32}.

The Bhabha equations (\ref{Bha}) are defined in the Minkowski space-time $\R^{1,3}$. These equations present an attempt to describe elementary particles in terms of space-time, that is, within Schr\"{o}dinger picture. With the aim to obtain an analogue of (\ref{Bha}) in the underlying spinor structure (or, within Heisenberg picture) we use a mapping into bivector space $\R^6$ (for more details see \cite{Var03,Var07}).  After the mapping of (\ref{Bha}) onto $\R^6$ we obtain the following system:
\[
\sum^3_{j=1}\sL^{l,0}_j
\frac{\partial\boldsymbol{\psi}}{\partial a_j}+
i\sum^3_{j=1}\sL^{l,0}_j
\frac{\partial\boldsymbol{\psi}}{\partial a^\ast_j}+
m^{(s)}\boldsymbol{\psi}=0,
\]
\begin{equation}\label{BS1}
\sum^3_{j=1}\overset{\ast}{\sL}{}^{0,\dot{l}}_j\frac{\dot{\boldsymbol{\psi}}}
{\partial\widetilde{a}_j}-
i\sum^3_{j=1}\overset{\ast}{\sL}{}^{0,\dot{l}}_j\frac{\dot{\boldsymbol{\psi}}}
{\partial\widetilde{a}^\ast_j}+m^{(s)}
\dot{\boldsymbol{\psi}}=0,
\end{equation}
where $a_1=\fg_1$, $a_2=\fg_2$, $a_3=\fg_3$, $ia_1=\fg_4$,
$ia_2=\fg_5$, $ia_3=\fg_6=$, $a^\ast_1=-i\fg_4$, $a^\ast_2=-i\fg_5$,
$a^\ast_3=-i\fg_6$, and $\widetilde{a}_j$, $\widetilde{a}^\ast_j$
are the parameters corresponding to the dual basis, $\fg_i\in\spin_+(1,3)$. These equations describe a field of type $(l,0)\oplus(0,\dot{l})$ and mass $m$ (these fields belong to the edges of the representation cone shown on the Fig.\,1). Non-null elements of the matrices $\sL^{l,0}_j$ and $\overset{\ast}{\sL}{}^{0,\dot{l}}_j$ have the form \cite{Var03}
\begin{equation}\label{L3}
{\renewcommand{\arraystretch}{1.1}
\sL^{l,0}_3:\quad\left\{\begin{array}{ccc} c^{k^\prime k}_{l-1,l,m}&=&
c^{k^\prime k}_{l-1,l}\sqrt{l^2-m^2},\\
c^{k^\prime k}_{l,l,m}&=&c^{k^\prime k}_{ll}m,\\
c^{k^\prime k}_{l+1,l,m}&=& c^{k^\prime
k}_{l+1,l}\sqrt{(l+1)^2-m^2}.
\end{array}\right.}
\end{equation}
\begin{equation}\label{L1}
{\renewcommand{\arraystretch}{1.25}
\sL^{l,0}_1:\quad\left\{\begin{array}{ccc} a^{k^\prime k}_{l-1,l,m-1,m}&=&
-\dfrac{c_{l-1,l}}{2}\sqrt{(l+m)(l+m-1)},\\
a^{k^\prime k}_{l,l,m-1,m}&=&
\dfrac{c_{ll}}{2}\sqrt{(l+m)(l-m+1)},\\
a^{k^\prime k}_{l+1,l,m-1,m}&=&
\dfrac{c_{l+1,l}}{2}\sqrt{(l-m+1)(l-m+2)},\\
a^{k^\prime k}_{l-1,l,m+1,m}&=&
\dfrac{c_{l-1,l}}{2}\sqrt{(l-m)(l-m-1)},\\
a^{k^\prime k}_{l,l,m+1,m}&=&
\dfrac{c_{ll}}{2}\sqrt{(l+m+1)(l-m)},\\
a^{k^\prime k}_{l+1,l,m+1,m}&=&
-\dfrac{c_{l+1,l}}{2}\sqrt{(l+m+1)(l+m+2)}.
\end{array}\right.}
\end{equation}
\begin{equation}\label{L2}
{\renewcommand{\arraystretch}{1.25}
\sL^{l,0}_2:\quad\left\{\begin{array}{ccc} b^{k^\prime k}_{l-1,l,m-1,m}&=&
-\dfrac{ic_{l-1,l}}{2}\sqrt{(l+m)(l+m-1)},\\
b^{k^\prime k}_{l,l,m-1,m}&=&
\dfrac{ic_{ll}}{2}\sqrt{(l+m)(l-m+1)},\\
b^{k^\prime k}_{l+1,l,m-1,m}&=&
\dfrac{ic_{l+1,l}}{2}\sqrt{(l-m+1)(l-m+2)},\\
b^{k^\prime k}_{l-1,l,m+1,m}&=&
-\dfrac{ic_{l-1,l}}{2}\sqrt{(l-m)(l-m-1)},\\
b^{k^\prime k}_{l,l,m+1,m}&=&
-\dfrac{ic_{ll}}{2}\sqrt{(l+m+1)(l-m)},\\
b^{k^\prime k}_{l+1,l,m+1,m}&=& \dfrac{ic_{l+1,l}}{2}\sqrt{(l+m+1)(l+m+2)}.
\end{array}\right.}
\end{equation}
\begin{equation}\label{L1'}
{\renewcommand{\arraystretch}{1.25}
\overset{\ast}{\sL}{}^{0,\dot{l}}_1:\quad\left\{\begin{array}{ccc}
d^{\dot{k}^\prime\dot{k}}_{\dot{l}-1,\dot{l},\dot{m}-1,\dot{m}}&=&
-\dfrac{c_{\dot{l}-1,\dot{l}}}{2}
\sqrt{(\dot{l}+\dot{m})(\dot{l}-\dot{m}-1)},\\
d^{\dot{k}^\prime\dot{k}}_{\dot{l},\dot{l},\dot{m}-1,\dot{m}}&=&
\dfrac{c_{\dot{l}\dot{l}}}{2}
\sqrt{(\dot{l}+\dot{m})(\dot{l}-\dot{m}+1)},\\
d^{\dot{k}^\prime\dot{k}}_{\dot{l}+1,\dot{l},\dot{m}-1,\dot{m}}&=&
\dfrac{c_{\dot{l}+1,\dot{l}}}{2}
\sqrt{(\dot{l}-\dot{m}+1)(\dot{l}-\dot{m}+2)},\\
d^{\dot{k}^\prime\dot{k}}_{\dot{l}-1,\dot{l},\dot{m}+1,\dot{m}}&=&
\dfrac{c_{\dot{l}-1,\dot{l}}}{2}
\sqrt{(\dot{l}-\dot{m})(\dot{l}-\dot{m}-1)},\\
d^{\dot{k}^\prime\dot{k}}_{\dot{l},\dot{l},\dot{m}+1,\dot{m}}&=&
\dfrac{c_{\dot{l}\dot{l}}}{2}
\sqrt{(\dot{l}+\dot{m}+1)(\dot{l}-\dot{m})},\\
d^{\dot{k}^\prime\dot{k}}_{\dot{l}+1,\dot{l},\dot{m}+1,\dot{m}}&=&
-\dfrac{c_{\dot{l}+1,\dot{l}}}{2}
\sqrt{(\dot{l}+\dot{m}+1)(\dot{l}+\dot{m}+2)}.
\end{array}\right.}
\end{equation}
\begin{equation}\label{L2'}
{\renewcommand{\arraystretch}{1.25}
\overset{\ast}{\sL}{}^{0,\dot{l}}_2:\quad\left\{\begin{array}{ccc}
e^{\dot{k}^\prime\dot{k}}_{\dot{l}-1,\dot{l},\dot{m}-1,\dot{m}}&=&
-\dfrac{ic_{\dot{l}-1,\dot{l}}}{2}
\sqrt{(\dot{l}+\dot{m})(\dot{l}-\dot{m}-1)},\\
e^{\dot{k}^\prime\dot{k}}_{\dot{l},\dot{l},\dot{m}-1,\dot{m}}&=&
\dfrac{ic_{\dot{l}\dot{l}}}{2}
\sqrt{(\dot{l}+\dot{m})(\dot{l}-\dot{m}+1)},\\
e^{\dot{k}^\prime\dot{k}}_{\dot{l}+1,\dot{l},\dot{m}-1,\dot{m}}&=&
\dfrac{ic_{\dot{l}+1,\dot{l}}}{2}
\sqrt{(\dot{l}-\dot{m}+1)(\dot{l}-\dot{m}+2)},\\
e^{\dot{k}^\prime\dot{k}}_{\dot{l}-1,\dot{l},\dot{m}+1,\dot{m}}&=&
\dfrac{-ic_{\dot{l}-1,\dot{l}}}{2}
\sqrt{(\dot{l}-\dot{m})(\dot{l}-\dot{m}-1)},\\
e^{\dot{k}^\prime\dot{k}}_{\dot{l},\dot{l},\dot{m}+1,\dot{m}}&=&
\dfrac{-ic_{\dot{l}\dot{l}}}{2}
\sqrt{(\dot{l}+\dot{m}+1)(\dot{l}-\dot{m})},\\
e^{\dot{k}^\prime\dot{k}}_{\dot{l}+1,\dot{l},\dot{m}+1,\dot{m}}&=&
-\dfrac{ic_{\dot{l}+1,\dot{l}}}{2}
\sqrt{(\dot{l}+\dot{m}+1)(\dot{l}+\dot{m}+2)}.
\end{array}\right.}
\end{equation}
\begin{equation}\label{L3'}
{\renewcommand{\arraystretch}{1.25}
\overset{\ast}{\sL}{}^{0,\dot{l}}_3:\quad\left\{\begin{array}{ccc}
f^{\dot{k}^\prime\dot{k}}_{\dot{l}-1,\dot{l},\dot{m}}&=&
c^{\dot{k}^\prime\dot{k}}_{\dot{l}-1,\dot{l}}
\sqrt{\dot{l}^2-\dot{m}^2},\\
f^{\dot{k}^\prime\dot{k}}_{\dot{l}\dot{l},\dot{m}}&=&
c^{\dot{k}^\prime\dot{k}}_{\dot{l}\dot{l}}\dot{m},\\
f^{\dot{k}^\prime\dot{k}}_{\dot{l}+1,\dot{l},\dot{m}}&=&
c^{\dot{k}^\prime\dot{k}}_{\dot{l}+1,\dot{l}}
\sqrt{(\dot{l}+1)^2-\dot{m}^2}.
\end{array}\right.}
\end{equation}
Solutions of (\ref{BS1}) are defined via series in hyperspherical functions \cite{Var03,Var04e,Var06}
\begin{multline}
\fM^{-\frac{1}{2}+i\rho,l_0}_{mn}(\fg)=e^{-m(\epsilon+i\varphi)-n(\varepsilon+i\psi)}
\sqrt{\frac{\Gamma(l_0+m+1)\Gamma(i\rho-n+\frac{1}{2})}
{\Gamma(l_0-m+1)\Gamma(i\rho+n+\frac{1}{2})}}\times\\
\times\cos^{2l_0}\frac{\theta}{2}\cosh^{-1+2i\rho}\frac{\tau}{2}
\sum^{l_0}_{t=-l_0}
i^{m-t}\tan^{m-t}\frac{\theta}{2}\tanh^{t-n}\frac{\tau}{2}\times\\
\times\hypergeom{2}{1}{m-l_0,-t-l_0}{m-t+1}{-\tan^2\frac{\theta}{2}}
\hypergeom{2}{1}{t-i\rho+\frac{1}{2},-n-i\rho+\frac{1}{2}}{t-n+1}{\tanh^2\frac{\tau}{2}},
\quad m\geq n. \nonumber
\end{multline}

Further, equations for arbitrary spin chains
\[
\boldsymbol{\tau}_{l\dot{l}},\;\boldsymbol{\tau}_{l+\frac{1}{2},\dot{l}-\frac{1}{2}},\;
\boldsymbol{\tau}_{l+1,\dot{l}-1},\;\boldsymbol{\tau}_{l+\frac{3}{2},\dot{l}-\frac{3}{2}},\;\ldots,\;
\boldsymbol{\tau}_{\dot{l}l},
\]
have the form
\[
\sum^3_{j=1}\Lambda^{l\dot{l}}_j\frac{\partial\psi}{\partial a_j}-
i\sum^3_{j=1}\Lambda^{l\dot{l}}_j\frac{\partial\psi}{\partial a^\ast_j}+m^{(s)}\dot{\psi}=0,
\]
\[
\sum^3_{j=1}\Lambda^{l+\frac{1}{2},\dot{l}-\frac{1}{2}}_j\frac{\partial\psi}{\partial a_j}-
i\sum^3_{j=1}\Lambda^{l+\frac{1}{2}\dot{l}-\frac{1}{2}}_j\frac{\partial\psi}{\partial a^\ast_j}+m^{(s)}\dot{\psi}=0,
\]
\[
\sum^3_{j=1}\Lambda^{l+1,\dot{l}-1}_j\frac{\partial\psi}{\partial a_j}-
i\sum^3_{j=1}\Lambda^{l+1\dot{l}-1}_j\frac{\partial\psi}{\partial a^\ast_j}+m^{(s)}\dot{\psi}=0,
\]
\[
\sum^3_{j=1}\Lambda^{l+\frac{3}{2},\dot{l}-\frac{3}{2}}_j\frac{\partial\psi}{\partial a_j}-
i\sum^3_{j=1}\Lambda^{l+\frac{3}{2}\dot{l}-\frac{3}{2}}_j\frac{\partial\psi}{\partial a^\ast_j}+m^{(s)}\dot{\psi}=0,
\]
\[
\ldots\ldots\ldots\ldots\ldots\ldots\ldots\ldots\ldots\ldots\ldots\ldots
\]
\begin{equation}\label{BS2}
\sum^3_{j=1}\Lambda^{\dot{l}l}_j\frac{\partial\dot{\psi}}{\partial\widetilde{a}_j}+
i\sum^3_{j=1}\Lambda^{\dot{l}l}_j\frac{\partial\dot{\psi}}{\partial\widetilde{a}^\ast_j}+m^{(s)}\psi=0,
\end{equation}
where the spin $s=l-\dot{l}$ is changed as
\[
l-\dot{l},\;l-\dot{l}+1,\;l-\dot{l}+2,\;l-\dot{l}+3,\,\ldots,\;\dot{l}-l.
\]
The system (\ref{BS2}) describes particles with different masses and spins in $\R^6$. The fields of types $(l,0)$, $(0,\dot{l})$ and $(l,0)\oplus(0,\dot{l})$ are particular cases of these fields. The state mass $m^{(s)}$, corresponding to the energy level $\sH_E\simeq\Sym_{(k,r)}$, is defined by the formula\footnote{Mass spectrum and its dependence of the Lorentz group representations, characterizing by the pair $(l,\dot{l})$, was given in \cite{Var15}. Usually one consider the mass spectrum dependence of spin $s$ only (Lorentz group representation is fixed). As it is known, in this case non-physical mass spectrum arises, $m_i\sim 1/s_i$, that is, for very large $s_i$, masses are very small (the first mass formula of this type was given by Majorana \cite{Maj32} and similar mass formulas were  considered by Gel'fand and Yaglom \cite{GY48}, see detailed historical consideration in \cite{Esp12}). In the paper \cite{Var15} it has been shown that for representations $(l,\dot{l})$ of the Lorentz group the mass is proportional to $(l+1/2)(\dot{l}+1/2)$. Hence it immediately follows that particles (more precisely, particle states) with the \textit{same} spin but \textit{distinct} masses are described by \textit{different} representations of Lorentz group.}
\begin{equation}\label{MGY}
m^{(s)}=\mu^0\left(l+\frac{1}{2}\right)\left(\dot{l}+\frac{1}{2}\right),
\end{equation}
where $s=|l-\dot{l}|$. An explicit form of the elements of $\Lambda^{l\dot{l}}_j$ is derived via calculation of the commutators
\[
\ld\ld\sL^{l\dot{l}}_3,\sX^{l\dot{l}}_-\rd,\sX^{l\dot{l}}_+\rd =
2\sL^{l\dot{l}}_3,\quad
\ld\ld\sL^{l\dot{l}}_3,\sY^{l\dot{l}}_-\rd,\sY^{l\dot{l}}_+\rd=
2\sL^{l\dot{l}}_3,
\]
\begin{equation}\label{Commut}
\ld\sA^{l\dot{l}}_2,\sL^{l\dot{l}}_3\rd=\sL^{l\dot{l}}_1,\quad
\ld\sA^{l\dot{l}}_1,\sL^{l\dot{l}}_3\rd=-\sL^{l\dot{l}}_2
\end{equation}
with respect to ket- and bra-vectors of the helicity basis (for more details see \cite{Var07}), here $\sA^{l\dot{l}}_i=\sA^l_i\otimes\boldsymbol{1}_{2\dot{l}+1}-
\boldsymbol{1}_{2l+1}\otimes\sA^{\dot{l}}_i$. Non-null elements of the matrices $\Lambda^{l\dot{l}}_j$ have the form \cite{Var07}
\begin{equation}\label{L3T}
{\renewcommand{\arraystretch}{1.5}
\sL^{l\dot{l}}_3:\quad\left\{\begin{array}{ccc} c^{k^\prime
k;\dot{k}^\prime\dot{k}}_{l-1,l,m;\dot{l}-1,\dot{l},\dot{m}}
&=&c^{k^\prime k;\dot{k}^\prime\dot{k}}_{l-1,l;\dot{l}-1,\dot{l}}
\sqrt{(l^2-m^2)(\dot{l}^2-\dot{m}^2)},\\
c^{k^\prime
k;\dot{k}^\prime\dot{k}}_{ll,m;\dot{l}\dot{l},\dot{m}}&=&
c^{k^\prime k;\dot{k}^\prime\dot{k}}_{ll;\dot{l}\dot{l}}m\dot{m},\\
c^{k^\prime
k;\dot{k}^\prime\dot{k}}_{l+1,l,m;\dot{l}+1,\dot{l},\dot{m}}
&=&c^{k^\prime k;\dot{k}^\prime\dot{k}}_{l+1,l;\dot{l}+1,\dot{l}}
\sqrt{((l+1)^2-m^2)((\dot{l}+1)^2-\dot{m}^2)}.
\end{array}\right.}
\end{equation}
$\Lambda^{l\dot{l}}_1$:
\begin{eqnarray}
a^{k^\prime
k;\dot{k}^\prime\dot{k}}_{l-1,l,m-1,m;\dot{l}-1,\dot{l},\dot{m}
\dot{m}}&=& -\frac{1}{2}c^{k^\prime
k;\dot{k}^\prime\dot{k}}_{l-1,l;\dot{l}-1,\dot{l}}
\sqrt{(l+m)(l+m-1)(\dot{l}^2-\dot{m}^2)},\nonumber\\
a^{k^\prime
k;\dot{k}^\prime\dot{k}}_{ll,m-1,m;\dot{l}\dot{l},\dot{m}\dot{m}}
&=& \frac{1}{2}c^{k^\prime
k;\dot{k}^\prime\dot{k}}_{ll;\dot{l}\dot{l}}
\dot{m}\sqrt{(l+m)(l-m+1)},\nonumber\\
a^{k^\prime
k;\dot{k}^\prime\dot{k}}_{l+1,l,m-1,m;\dot{l}+1,\dot{l},\dot{m}
\dot{m}}&=& \frac{1}{2}c^{k^\prime
k;\dot{k}^\prime\dot{k}}_{l+1,l;\dot{l}+1,\dot{l}}
\sqrt{(l-m+1)(l-m+2)((\dot{l}+1)^2-\dot{m}^2)},\nonumber
\end{eqnarray}
\begin{eqnarray}
a^{k^\prime
k;\dot{k}^\prime\dot{k}}_{l-1,l,m+1,m;\dot{l}-1,\dot{l},\dot{m}
\dot{m}}&=& \frac{1}{2}c^{k^\prime
k;\dot{k}^\prime\dot{k}}_{l-1,l;\dot{l}-1,\dot{l}}
\sqrt{(l-m)(l-m-1)(\dot{l}^2-\dot{m}^2)},\nonumber\\
a^{k^\prime
k;\dot{k}^\prime\dot{k}}_{ll,m+1,m;\dot{l}\dot{l},\dot{m}
\dot{m}}&=& \frac{1}{2}c^{k^\prime
k;\dot{k}^\prime\dot{k}}_{ll;\dot{l}\dot{l}}
\dot{m}\sqrt{(l+m+1)(l-m)},\nonumber\\
a^{k^\prime
k;\dot{k}^\prime\dot{k}}_{l+1,l,m+1,m;\dot{l}+1,\dot{l},\dot{m}
\dot{m}}&=& -\frac{1}{2}c^{k^\prime
k;\dot{k}^\prime\dot{k}}_{l+1,l;\dot{l}+1,\dot{l}}
\sqrt{(l+m+1)(l+m+2)((\dot{l}+1)^2-\dot{m}^2)},\nonumber
\end{eqnarray}
\begin{eqnarray}
a^{k^\prime
k;\dot{k}^\prime\dot{k}}_{l-1,l,m,m;\dot{l}-1,\dot{l},\dot{m}-1,
\dot{m}}&=& \frac{1}{2}c^{k^\prime
k;\dot{k}^\prime\dot{k}}_{l-1,l;\dot{l}-1,\dot{l}}
\sqrt{(l^2-m^2)(\dot{l}+\dot{m})(\dot{l}+\dot{m}-1)},\nonumber\\
a^{k^\prime
k;\dot{k}^\prime\dot{k}}_{ll,m,m;\dot{l}\dot{l},\dot{m}-1,\dot{m}}
&=& -\frac{1}{2}c^{k^\prime
k;\dot{k}^\prime\dot{k}}_{ll;\dot{l}\dot{l}}
m\sqrt{(\dot{l}+\dot{m})(\dot{l}-\dot{m}+1)},\nonumber\\
a^{k^\prime
k;\dot{k}^\prime\dot{k}}_{l+1,l,m,m;\dot{l}+1,\dot{l},\dot{m}-1,
\dot{m}}&=& -\frac{1}{2}c^{k^\prime
k;\dot{k}^\prime\dot{k}}_{l+1,l;\dot{l}+1,\dot{l}}
\sqrt{((l+1)^2-m^2)(\dot{l}-\dot{m}+1)(\dot{l}-\dot{m}+2)},\nonumber
\end{eqnarray}
\begin{eqnarray}
a^{k^\prime
k;\dot{k}^\prime\dot{k}}_{l-1,l,m,m;\dot{l}-1,\dot{l},\dot{m}+1,
\dot{m}}&=& -\frac{1}{2}c^{k^\prime
k;\dot{k}^\prime\dot{k}}_{l-1,l;\dot{l}-1,\dot{l}}
\sqrt{(l^2-m^2)(\dot{l}-\dot{m})(\dot{l}-\dot{m}-1)},\nonumber\\
a^{k^\prime
k;\dot{k}^\prime\dot{k}}_{ll,m,m;\dot{l}\dot{l},\dot{m}+1,
\dot{m}}&=& -\frac{1}{2}c^{k^\prime
k;\dot{k}^\prime\dot{k}}_{ll;\dot{l}\dot{l}}
m\sqrt{(\dot{l}+\dot{m}+1)(\dot{l}-\dot{m})},\nonumber\\
a^{k^\prime
k;\dot{k}^\prime\dot{k}}_{l+1,l,m,m;\dot{l}+1,\dot{l},\dot{m}+1,
\dot{m}}&=& \frac{1}{2}c^{k^\prime
k;\dot{k}^\prime\dot{k}}_{l+1,l;\dot{l}+1,\dot{l}}
\sqrt{((l+1)^2-m^2)(\dot{l}+\dot{m}+1)(\dot{l}+\dot{m}+2)}.\label{L1T}
\end{eqnarray}
$\sL^{l\dot{l}}_2$:
\begin{eqnarray}
b^{k^\prime
k;\dot{k}^\prime\dot{k}}_{l-1,l,m-1,m;\dot{l}-1,\dot{l},\dot{m}
\dot{m}}&=& -\frac{i}{2}c^{k^\prime
k;\dot{k}^\prime\dot{k}}_{l-1,l;\dot{l}-1,\dot{l}}
\sqrt{(l+m)(l+m-1)(\dot{l}^2-\dot{m}^2)},\nonumber\\
b^{k^\prime
k;\dot{k}^\prime\dot{k}}_{ll,m-1,m;\dot{l}\dot{l},\dot{m}\dot{m}}
&=& \frac{i}{2}c^{k^\prime
k;\dot{k}^\prime\dot{k}}_{ll;\dot{l}\dot{l}}
\dot{m}\sqrt{(l+m)(l-m+1)},\nonumber\\
b^{k^\prime
k;\dot{k}^\prime\dot{k}}_{l+1,l,m-1,m;\dot{l}+1,\dot{l},\dot{m}
\dot{m}}&=& \frac{i}{2}c^{k^\prime
k;\dot{k}^\prime\dot{k}}_{l+1,l;\dot{l}+1,\dot{l}}
\sqrt{(l-m+1)(l-m+2)((\dot{l}+1)^2-\dot{m}^2)},\nonumber
\end{eqnarray}
\begin{eqnarray}
b^{k^\prime
k;\dot{k}^\prime\dot{k}}_{l-1,l,m+1,m;\dot{l}-1,\dot{l},\dot{m}
\dot{m}}&=& -\frac{i}{2}c^{k^\prime
k;\dot{k}^\prime\dot{k}}_{l-1,l;\dot{l}-1,\dot{l}}
\sqrt{(l-m)(l-m-1)(\dot{l}^2-\dot{m}^2)},\nonumber\\
b^{k^\prime
k;\dot{k}^\prime\dot{k}}_{ll,m+1,m;\dot{l}\dot{l},\dot{m}
\dot{m}}&=& -\frac{i}{2}c^{k^\prime
k;\dot{k}^\prime\dot{k}}_{ll;\dot{l}\dot{l}}
\dot{m}\sqrt{(l+m+1)(l-m)},\nonumber\\
b^{k^\prime
k;\dot{k}^\prime\dot{k}}_{l+1,l,m+1,m;\dot{l}+1,\dot{l},\dot{m}
\dot{m}}&=& \frac{i}{2}c^{k^\prime
k;\dot{k}^\prime\dot{k}}_{l+1,l;\dot{l}+1,\dot{l}}
\sqrt{(l+m+1)(l+m+2)((\dot{l}+1)^2-\dot{m}^2)},\nonumber
\end{eqnarray}
\begin{eqnarray}
b^{k^\prime
k;\dot{k}^\prime\dot{k}}_{l-1,l,m,m;\dot{l}-1,\dot{l},\dot{m}-1,
\dot{m}}&=& \frac{i}{2}c^{k^\prime
k;\dot{k}^\prime\dot{k}}_{l-1,l;\dot{l}-1,\dot{l}}
\sqrt{(l^2-m^2)(\dot{l}+\dot{m})(\dot{l}+\dot{m}-1)},\nonumber\\
b^{k^\prime
k;\dot{k}^\prime\dot{k}}_{ll,m,m;\dot{l}\dot{l},\dot{m}-1,\dot{m}}
&=& -\frac{i}{2}c^{k^\prime
k;\dot{k}^\prime\dot{k}}_{ll;\dot{l}\dot{l}}
m\sqrt{(\dot{l}+\dot{m})(\dot{l}-\dot{m}+1)},\nonumber\\
b^{k^\prime
k;\dot{k}^\prime\dot{k}}_{l+1,l,m,m;\dot{l}+1,\dot{l},\dot{m}-1,
\dot{m}}&=& -\frac{i}{2}c^{k^\prime
k;\dot{k}^\prime\dot{k}}_{l+1,l;\dot{l}+1,\dot{l}}
\sqrt{((l+1)^2-m^2)(\dot{l}-\dot{m}+1)(\dot{l}-\dot{m}+2)},\nonumber
\end{eqnarray}
\begin{eqnarray}
b^{k^\prime
k;\dot{k}^\prime\dot{k}}_{l-1,l,m,m;\dot{l}-1,\dot{l},\dot{m}+1,
\dot{m}}&=& \frac{i}{2}c^{k^\prime
k;\dot{k}^\prime\dot{k}}_{l-1,l;\dot{l}-1,\dot{l}}
\sqrt{(l^2-m^2)(\dot{l}-\dot{m})(\dot{l}-\dot{m}-1)},\nonumber\\
b^{k^\prime
k;\dot{k}^\prime\dot{k}}_{ll,m,m;\dot{l}\dot{l},\dot{m}+1,
\dot{m}}&=& \frac{i}{2}c^{k^\prime
k;\dot{k}^\prime\dot{k}}_{ll;\dot{l}\dot{l}}
m\sqrt{(\dot{l}+\dot{m}+1)(\dot{l}-\dot{m})},\nonumber\\
b^{k^\prime
k;\dot{k}^\prime\dot{k}}_{l+1,l,m,m;\dot{l}+1,\dot{l},\dot{m}+1,
\dot{m}}&=& -\frac{i}{2}c^{k^\prime
k;\dot{k}^\prime\dot{k}}_{l+1,l;\dot{l}+1,\dot{l}}
\sqrt{((l+1)^2-m^2)(\dot{l}+\dot{m}+1)(\dot{l}+\dot{m}+2)}.\label{L2T}
\end{eqnarray}
For the dual basis we have the following non-null elements of $\overset{\ast}{\sL}{}^{l\dot{l}}_j$:
\begin{equation}\label{L3T'}
{\renewcommand{\arraystretch}{1.5}
\overset{\ast}{\sL}{}^{l\dot{l}}_3:\quad\left\{\begin{array}{ccc}
\overset{\ast}{c}{}^{k^\prime k;\dot{k}^\prime\dot{k}}_{l-1,l,m;
\dot{l}-1,\dot{l},\dot{m}} &=&\overset{\ast}{c}{}^{k^\prime
k;\dot{k}^\prime\dot{k}}_{l-1,l;\dot{l}-1,\dot{l}}
\sqrt{(l^2-m^2)(\dot{l}^2-\dot{m}^2)},\\
\overset{\ast}{c}{}^{k^\prime k;\dot{k}^\prime\dot{k}}_{ll,m;
\dot{l}\dot{l},\dot{m}}&=& \overset{\ast}{c}{}^{k^\prime
k;\dot{k}^\prime\dot{k}}_{ll;
\dot{l}\dot{l}}m\dot{m},\\
\overset{\ast}{c}{}^{k^\prime k;\dot{k}^\prime\dot{k}}_{l+1,l,m;
\dot{l}+1,\dot{l},\dot{m}} &=&\overset{\ast}{c}{}^{k^\prime
k;\dot{k}^\prime\dot{k}}_{l+1,l; \dot{l}+1,\dot{l}}
\sqrt{((l+1)^2-m^2)((\dot{l}+1)^2-\dot{m}^2)}.
\end{array}\right.}
\end{equation}
$\overset{\ast}{\sL}{}^{l\dot{l}}_1$:
\begin{eqnarray}
\overset{\ast}{a}{}^{k^\prime k;\dot{k}^\prime\dot{k}}_{l-1,l,m-1,m;
\dot{l}-1,\dot{l},\dot{m}\dot{m}}&=&
\frac{1}{2}\overset{\ast}{c}{}^{k^\prime
k;\dot{k}^\prime\dot{k}}_{l-1,l;
\dot{l}-1,\dot{l}}\sqrt{(l+m)(l+m-1)(\dot{l}^2-\dot{m}^2)},\nonumber\\
\overset{\ast}{a}{}^{k^\prime k;\dot{k}^\prime\dot{k}}_{ll,m-1,m;
\dot{l}\dot{l},\dot{m}\dot{m}}&=&
-\frac{1}{2}\overset{\ast}{c}{}^{k^\prime
k;\dot{k}^\prime\dot{k}}_{ll;
\dot{l}\dot{l}}\dot{m}\sqrt{(l+m)(l-m+1)},\nonumber\\
\overset{\ast}{a}{}^{k^\prime k;\dot{k}^\prime\dot{k}}_{l+1,l,m-1,m;
\dot{l}+1,\dot{l},\dot{m}\dot{m}}&=&
-\frac{1}{2}\overset{\ast}{c}{}^{k^\prime
k;\dot{k}^\prime\dot{k}}_{l+1,l;
\dot{l}+1,\dot{l}}\sqrt{(l-m+1)(l-m+2)((\dot{l}+1)^2-\dot{m}^2)},\nonumber
\end{eqnarray}
\begin{eqnarray}
\overset{\ast}{a}{}^{k^\prime k;\dot{k}^\prime\dot{k}}_{l-1,l,m+1,m;
\dot{l}-1,\dot{l},\dot{m}\dot{m}}&=&
-\frac{1}{2}\overset{\ast}{c}{}^{k^\prime
k;\dot{k}^\prime\dot{k}}_{l-1,l;
\dot{l}-1,\dot{l}}\sqrt{(l-m)(l-m-1)(\dot{l}^2-\dot{m}^2)},\nonumber\\
\overset{\ast}{a}{}^{k^\prime k;\dot{k}^\prime\dot{k}}_{ll,m+1,m;
\dot{l}\dot{l},\dot{m}\dot{m}}&=&
-\frac{1}{2}\overset{\ast}{c}{}^{k^\prime
k;\dot{k}^\prime\dot{k}}_{ll;
\dot{l}\dot{l}}\dot{m}\sqrt{(l+m+1)(l-m)},\nonumber\\
\overset{\ast}{a}{}^{k^\prime k;\dot{k}^\prime\dot{k}}_{l+1,l,m+1,m;
\dot{l}+1,\dot{l},\dot{m}\dot{m}}&=&
\frac{1}{2}\overset{\ast}{c}{}^{k^\prime
k;\dot{k}^\prime\dot{k}}_{l+1,l;
\dot{l}+1,\dot{l}}\sqrt{(l+m+1)(l+m+2)((\dot{l}+1)^2-\dot{m}^2)},\nonumber
\end{eqnarray}
\begin{eqnarray}
\overset{\ast}{a}{}^{k^\prime k;\dot{k}^\prime\dot{k}}_{l-1,l,m,m;
\dot{l}-1,\dot{l},\dot{m}-1,\dot{m}}&=&
-\frac{1}{2}\overset{\ast}{c}{}^{k^\prime
k;\dot{k}^\prime\dot{k}}_{l-1,l;
\dot{l}-1,\dot{l}}\sqrt{(l^2-m^2)(\dot{l}+\dot{m})(\dot{l}+\dot{m}-1)},\nonumber\\
\overset{\ast}{a}{}^{k^\prime k;\dot{k}^\prime\dot{k}}_{ll,m,m;
\dot{l}\dot{l},\dot{m}-1,\dot{m}}&=&
\frac{1}{2}\overset{\ast}{c}{}^{k^\prime
k;\dot{k}^\prime\dot{k}}_{ll;
\dot{l}\dot{l}}m\sqrt{(\dot{l}+\dot{m})(\dot{l}-\dot{m}+1)},\nonumber\\
\overset{\ast}{a}{}^{k^\prime
k;\dot{k}^\prime\dot{k}}_{l+1,l,m,m;\dot{l}+1,
\dot{l},\dot{m}-1,\dot{m}}&=&
\frac{1}{2}\overset{\ast}{c}{}^{k^\prime
k;\dot{k}^\prime\dot{k}}_{l+1,l; \dot{l}+1,\dot{l}}
\sqrt{((l+1)^2-m^2)(\dot{l}-\dot{m}+1)(\dot{l}-\dot{m}+2)},\nonumber
\end{eqnarray}
\begin{eqnarray}
\overset{\ast}{a}{}^{k^\prime
k;\dot{k}^\prime\dot{k}}_{l-1,l,m,m;\dot{l}-1,
\dot{l},\dot{m}+1,\dot{m}}&=&
\frac{1}{2}\overset{\ast}{c}{}^{k^\prime
k;\dot{k}^\prime\dot{k}}_{l-1,l; \dot{l}-1,\dot{l}}
\sqrt{(l^2-m^2)(\dot{l}-\dot{m})(\dot{l}-\dot{m}-1)},\nonumber\\
\overset{\ast}{a}{}^{k^\prime k;\dot{k}^\prime\dot{k}}_{ll,m,m;
\dot{l}\dot{l},\dot{m}+1,\dot{m}}&=&
\frac{1}{2}\overset{\ast}{c}{}^{k^\prime
k;\dot{k}^\prime\dot{k}}_{ll;
\dot{l}\dot{l}}m\sqrt{(\dot{l}+\dot{m}+1)(\dot{l}-\dot{m})},\nonumber\\
\overset{\ast}{a}{}^{k^\prime k;\dot{k}^\prime\dot{k}}_{l+1,l,m,m;
\dot{l}+1,\dot{l},\dot{m}+1,\dot{m}}&=&
-\frac{1}{2}\overset{\ast}{c}{}^{k^\prime
k;\dot{k}^\prime\dot{k}}_{l+1,l; \dot{l}+1,\dot{l}}
\sqrt{((l+1)^2-m^2)(\dot{l}+\dot{m}+1)(\dot{l}+\dot{m}+2)}.\label{L1T'}
\end{eqnarray}
$\overset{\ast}{\sL}{}^{l\dot{l}}_2$:
\begin{eqnarray}
\overset{\ast}{b}{}^{k^\prime k;\dot{k}^\prime\dot{k}}_{l-1,l,m-1,m;
\dot{l}-1,\dot{l},\dot{m}\dot{m}}&=&
\frac{i}{2}\overset{\ast}{c}{}^{k^\prime
k;\dot{k}^\prime\dot{k}}_{l-1,l;
\dot{l}-1,\dot{l}}\sqrt{(l+m)(l+m-1)(\dot{l}^2-\dot{m}^2)},\nonumber\\
\overset{\ast}{b}{}^{k^\prime k;\dot{k}^\prime\dot{k}}_{ll,m-1,m;
\dot{l}\dot{l},\dot{m}\dot{m}}&=&
-\frac{i}{2}\overset{\ast}{c}{}^{k^\prime
k;\dot{k}^\prime\dot{k}}_{ll;
\dot{l}\dot{l}}\dot{m}\sqrt{(l+m)(l-m+1)},\nonumber\\
\overset{\ast}{b}{}^{k^\prime k;\dot{k}^\prime\dot{k}}_{l+1,l,m-1,m;
\dot{l}+1,\dot{l},\dot{m}\dot{m}}&=&
-\frac{i}{2}\overset{\ast}{c}{}^{k^\prime
k;\dot{k}^\prime\dot{k}}_{l+1,l;
\dot{l}+1,\dot{l}}\sqrt{(l-m+1)(l-m+2)((\dot{l}+1)^2-\dot{m}^2)},\nonumber
\end{eqnarray}
\begin{eqnarray}
\overset{\ast}{b}{}^{k^\prime k;\dot{k}^\prime\dot{k}}_{l-1,l,m+1,m;
\dot{l}-1,\dot{l},\dot{m}\dot{m}}&=&
\frac{i}{2}\overset{\ast}{c}{}^{k^\prime
k;\dot{k}^\prime\dot{k}}_{l-1,l;
\dot{l}-1,\dot{l}}\sqrt{(l-m)(l-m-1)(\dot{l}^2-\dot{m}^2)},\nonumber\\
\overset{\ast}{b}{}^{k^\prime k;\dot{k}^\prime\dot{k}}_{ll,m+1,m;
\dot{l}\dot{l},\dot{m}\dot{m}}&=&
\frac{i}{2}\overset{\ast}{c}{}^{k^\prime
k;\dot{k}^\prime\dot{k}}_{ll;
\dot{l}\dot{l}}\dot{m}\sqrt{(l+m+1)(l-m)},\nonumber\\
\overset{\ast}{b}{}^{k^\prime k;\dot{k}^\prime\dot{k}}_{l+1,l,m+1,m;
\dot{l}+1,\dot{l},\dot{m}\dot{m}}&=&
-\frac{i}{2}\overset{\ast}{c}{}^{k^\prime
k;\dot{k}^\prime\dot{k}}_{l+1,l;
\dot{l}+1,\dot{l}}\sqrt{(l+m+1)(l+m+2)((\dot{l}+1)^2-\dot{m}^2)},\nonumber
\end{eqnarray}
\begin{eqnarray}
\overset{\ast}{b}{}^{k^\prime k;\dot{k}^\prime\dot{k}}_{l-1,l,m,m;
\dot{l}-1,\dot{l},\dot{m}-1,\dot{m}}&=&
-\frac{i}{2}\overset{\ast}{c}{}^{k^\prime
k;\dot{k}^\prime\dot{k}}_{l-1,l; \dot{l}-1,\dot{l}}
\sqrt{(l^2-m^2)(\dot{l}+\dot{m})(\dot{l}+\dot{m}-1)},\nonumber\\
\overset{\ast}{b}{}^{k^\prime k;\dot{k}^\prime\dot{k}}_{ll,m,m;
\dot{l}\dot{l},\dot{m}-1,\dot{m}}&=&
\frac{i}{2}\overset{\ast}{c}{}^{k^\prime
k;\dot{k}^\prime\dot{k}}_{ll;
\dot{l}\dot{l}}m\sqrt{(\dot{l}+\dot{m})(\dot{l}-\dot{m}+1)},\nonumber\\
\overset{\ast}{b}{}^{k^\prime k;\dot{k}^\prime\dot{k}}_{l+1,l,m,m;
\dot{l}+1,\dot{l},\dot{m}-1,\dot{m}}&=&
\frac{i}{2}\overset{\ast}{c}{}^{k^\prime
k;\dot{k}^\prime\dot{k}}_{l+1,l; \dot{l}+1,\dot{l}}
\sqrt{((l+1)^2-m^2)(\dot{l}-\dot{m}+1)(\dot{l}-\dot{m}+2)},\nonumber
\end{eqnarray}
\begin{eqnarray}
\overset{\ast}{b}{}^{k^\prime k;\dot{k}^\prime\dot{k}}_{l-1,l,m,m;
\dot{l}-1,\dot{l},\dot{m}+1,\dot{m}}&=&
-\frac{i}{2}\overset{\ast}{c}{}^{k^\prime
k;\dot{k}^\prime\dot{k}}_{l-1,l; \dot{l}-1,\dot{l}}
\sqrt{(l^2-m^2)(\dot{l}-\dot{m})(\dot{l}-\dot{m}-1)},\nonumber\\
\overset{\ast}{b}{}^{k^\prime
k;\dot{k}^\prime\dot{k}}_{ll,m,m;\dot{l}
\dot{l},\dot{m}+1,\dot{m}}&=&
-\frac{i}{2}\overset{\ast}{c}{}^{k^\prime
k;\dot{k}^\prime\dot{k}}_{ll;
\dot{l}\dot{l}}m\sqrt{(\dot{l}+\dot{m}+1)(\dot{l}-\dot{m})},\nonumber\\
\overset{\ast}{b}{}^{k^\prime k;\dot{k}^\prime\dot{k}}_{l+1,l,m,m;
\dot{l}+1,\dot{l},\dot{m}+1,\dot{m}}&=&
\frac{i}{2}\overset{\ast}{c}{}^{k^\prime
k;\dot{k}^\prime\dot{k}}_{l+1,l; \dot{l}+1,\dot{l}}
\sqrt{((l+1)^2-m^2)(\dot{l}+\dot{m}+1)(\dot{l}+\dot{m}+2)}.\label{L2T'}
\end{eqnarray}
Solutions of (\ref{BS2}) are defined via series in generalized hyperspherical functions \cite{Var07}
\[
\fM^{l\dot{l}}_{mn;\dot{m}\dot{n}}(\varphi,\epsilon,\theta,\tau,0,0)=e^{-n(\epsilon+i\varphi)-\dot{n}(\epsilon-i\varphi)}
\fZ^{l\dot{l}}_{mn;\dot{m}\dot{n}}(\theta,\tau).
\]
\section{Dirac and Maxwell equations}
In this section we consider two levels of matter spectrum which correspond to Dirac and Maxwell equations. The first level describes electron state and its expression in the form of Dirac equation is well-known. The second level describes photon state which leads to Dirac-like form of Maxwell equations. Such intriguing similarity of these equations has a simple explanation: both levels of matter have a common nature (they are differed only by tensor dimension). This similarity is expressed in spinor form.
The spinor form of Dirac equations was first given by
Van der Waerden \cite{Wa29}, he showed that Dirac's theory can be
expressed completely in this form. In turn, a
spinor formulation of Maxwell equations was studied by Laporte and
Uhlenbeck \cite{LU31}. In 1936, Rumer \cite{Rum36} showed that
spinor forms of Dirac and Maxwell equations look very
similar.
Further, Majorana \cite{ERMB} and Oppenheimer
\cite{Opp31} proposed to consider the Maxwell theory of
electromagnetism as the wave mechanics of the photon. They
introduced a wave function of the form $\psi=\bE-i\bB$ satisfying
the massless Dirac-like equations\footnote{In contrast to the
Gupta-Bleuler method, where the non-observable
four-potential $A_\mu$ is quantized, the main advantage of the
Majorana-Oppenheimer formulation of electrodynamics lies in the
fact that it deals directly with observable quantities, such as
the electric and magnetic fields.}. Maxwell equations in the Dirac
form considered during long time by many authors
(see \cite{Var05b} and references therein).
The interest to the Majorana-Oppenheimer formulation
of electrodynamics has grown in recent years \cite{BKOR10,Esp98}.
\subsection{Dirac equation}
The spin chain (fundamental doublet)
\begin{equation}\label{EChain}
\unitlength=1mm
\begin{picture}(20,13)
\put(0,5){$\overset{(0,\frac{1}{2})}{\bullet}$}
\put(20,5){$\overset{(\frac{1}{2},0)}{\bullet}$}
\put(3,6){\line(1,0){20}}
\put(-0.5,0){$-\frac{1}{2}$}
\put(22.5,0){$\frac{1}{2}$}
\put(6,0){$\cdots$}
\put(11.5,0){$\cdots$}
\put(17.5,0){$\cdots$}
\end{picture}
\end{equation}
on the Fig.\,1, defined within the representation $\boldsymbol{\tau}_{0,1/2}\oplus\boldsymbol{\tau}_{1/2,0}$, leads to a linear superposition of the two spin states $1/2$ and $-1/2$, that describes electron and corresponds to Dirac equation
\begin{equation}\label{Dirac}
\gamma_\mu\frac{\partial\psi}{\partial x_\mu}+m_e\psi=0.
\end{equation}
In accordance with general definition (\ref{TenAlg}), with the electron spin chain (\ref{EChain}) we have a chain of algebras
\[
\C_2\longleftrightarrow\overset{\ast}{\C}_2
\]
and associated spinspace
\[
\dS_2\oplus\dot{\dS}_2,
\]
where $\dS_2$ is a space of (2-component) spinors, and $\dot{\dS}_2$ is a space of cospinors. Hence it follows that the wave function of (\ref{Dirac}) is a \textit{bispinor} $\psi=\{\xi,\dot{\eta}\}$, where $\xi\in\dS_2$, $\dot{\eta}\in\dot{\dS}_2$, and $\psi$ can be written as
\begin{equation}\label{Bispinor}
\psi=\xi^1\boldsymbol{\varepsilon}_1+\xi^2\boldsymbol{\varepsilon}_2+\eta_{\dot{1}}\boldsymbol{\varepsilon}^{\dot{1}}+
\eta_{\dot{2}}\boldsymbol{\varepsilon}^{\dot{2}},
\end{equation}
or
\[
\psi=\begin{array}{||c||}
\xi^1\\
\xi^2\\
\eta_{\dot{1}}\\
\eta_{\dot{2}}
\end{array}.
\]
Let us consider a superposition defined by the fundamental doublet (\ref{EChain}), that corresponds to electron\footnote{It should be noted an interesting analogy between electron and qubit. As is known, \textit{qubit} is a nonseparable (entangled) state defined by a superposition $\left|\boldsymbol{\psi}\right\rangle=a\left|\boldsymbol{0}\right\rangle+b\left|\boldsymbol{1}\right\rangle,$
where $a,b\in\C$. Qubit is a minimally possible (elementary) state vector. Any state vector can be represented by a set of such elementary vectors. For that reason qubit is an original `building block' for all other state vectors of any dimension. A quantum state of $N$ qubits can be described by the vector of the Hilbert space of dimension $2^N$. It is obvious that this space coincides with the spinspace $\dS_{2^N}$. We can choose as an orthonormal basis for this space the states in which each qubit has a definite value, either $\left|\boldsymbol{0}\right\rangle$ or $\left|\boldsymbol{1}\right\rangle$. These can be labeled by binary strings such as
$\left|\boldsymbol{0}\boldsymbol{1}\boldsymbol{1}\boldsymbol{1}\boldsymbol{0}\boldsymbol{0}\boldsymbol{1}
\boldsymbol{0}\,\cdots\,\boldsymbol{1}\boldsymbol{0}\boldsymbol{0}\boldsymbol{1}\right\rangle$. A general normalized vector can be expressed in this basis as $\sum^{2^N-1}_{x=0}a_x\left|x\right\rangle$, where $a_x$ are complex numbers satisfying $\sum_x|a_x|^2=1$. Here we have a deep analogy between qubits and 2-component spinors. Just like the qubits, 2-component spinors are `building blocks' of the underlying spinor structure (via the tensor products of $\C_2$ and $\overset{\ast}{\C}_2$, see (\ref{TenAlg}) and (\ref{SpinSpace})). Moreover, vectors of the Hilbert space $\bsH^S\otimes\bsH^Q\otimes\bsH_\infty$, which represent actualized particle states (`binary strings'), are constructed via the same way \cite{Var15}.}. So, electron is a superposition $U$ of state vectors in nonseparable Hilbert space $\bsH^S_2\otimes\bsH^-\otimes\bsH_\infty$. In this case we have two spin states: the state 1/2, described by a representation $\boldsymbol{\tau}_{1/2,0}$ on the spin-1/2 line, and the state -1/2, described by a representation $\boldsymbol{\tau}_{0,1/2}$ on the dual spin-1/2 line. Representations $\boldsymbol{\tau}_{1/2,0}$ and $\boldsymbol{\tau}_{0,1/2}$ act in the spaces $\Sym_{(1,0)}$ and $\Sym_{(0,1)}$, respectively. Therefore, there are two state vectors: ket-vector  $\left|\Psi\right\rangle=\left.\left|\boldsymbol{\tau}_{1/2,0},\,\Sym_{(1,0)},\,\C_2,\,\dS_{2}
\right.\right\rangle$ and bra-vector $\langle\dot{\dS}_2,\,\overset{\ast}{\C}_2,\,\Sym_{(0,1)},\,
\boldsymbol{\tau}_{0,1/2}|=\langle\dot{\Psi}|$. Spinor structure is defined by a direct sum $\C_2\oplus\overset{\ast}{\C}_2$.
After the mapping of (\ref{Dirac}) into bivector space $\R^6$, we obtain equations
(particular case of the system (\ref{BS2}) for the spin chain $\boldsymbol{\tau}_{1/2,0}\leftrightarrow\boldsymbol{\tau}_{0,1/2}$)
\[
\sum^3_{j=1}\overset{\ast}{\Lambda}{}^{0,\frac{1}{2}}_j\frac{\partial\dot{\psi}}
{\partial\widetilde{a}_j}+i\sum^3_{j=1}\overset{\ast}{\Lambda}{}^{0,\frac{1}{2}}_j
\frac{\partial\dot{\psi}}{\partial\widetilde{a}^\ast_j}+
m_e\psi=0,
\]
\begin{equation}\label{BDirac}
\sum^3_{j=1}\Lambda^{\frac{1}{2},0}_j\frac{\partial\psi}{\partial a_j}-
i\sum^3_{j=1}\Lambda^{\frac{1}{2},0}_j\frac{\partial\psi}{\partial a^\ast_j}+
m_e\dot{\psi}=0,
\end{equation}
where
\[
\sL^{\frac{1}{2},0}_1=\frac{1}{2}c_{\frac{1}{2}\frac{1}{2}}\begin{array}{||cc||}
0 & 1\\
1 & 0
\end{array},\quad
\sL^{\frac{1}{2},0}_2=\frac{1}{2}c_{\frac{1}{2}\frac{1}{2}}\begin{array}{||cc||}
0 & -i\\
i & 0
\end{array},\quad
\sL^{\frac{1}{2},0}_3=\frac{1}{2}c_{\frac{1}{2}\frac{1}{2}}\begin{array}{||cc||}
1 & 0\\
0 & -1
\end{array},
\]
\[
\overset{\ast}{\sL}{}^{0,\frac{1}{2}}_1=\frac{1}{2}\dot{c}_{\frac{1}{2}\frac{1}{2}}\begin{array}{||cc||}
0 & 1\\
1 & 0
\end{array},\quad
\overset{\ast}{\sL}{}^{0,\frac{1}{2}}_2=\frac{1}{2}\dot{c}_{\frac{1}{2}\frac{1}{2}}\begin{array}{||cc||}
0 & -i\\
i & 0
\end{array},\quad
\overset{\ast}{\sL}{}^{0,\frac{1}{2}}_3=\frac{1}{2}\dot{c}_{\frac{1}{2}\frac{1}{2}}\begin{array}{||cc||}
1 & 0\\
0 & -1
\end{array}.
\]
It is easy to see that these matrices coincide with the Pauli
matrices $\sigma_i$ when $c_{\frac{1}{2}\frac{1}{2}}=2$. At the reduction of the superposition we have $U\rightarrow\boldsymbol{\tau}_{1/2,0}$ (or $U\rightarrow\boldsymbol{\tau}_{0,1/2}$) and $\bsH^S_2\otimes\bsH^-\otimes\bsH_\infty\rightarrow\sH_E\simeq\Sym_{(1,0)}$ (or $\bsH^S_2\otimes\bsH^-\otimes\bsH_\infty\rightarrow\sH_E\simeq\Sym_{(0,1)}$). Eigenvector subspace $\sH_E$ of the energy operator $H$ presents a mass level with the mass value defined by the formula (\ref{MGY}). Let $s=l=1/2$ and let $m_e=\mu^0\left(l+\frac{1}{2}\right)=\mu^0$ be an electron mass. From the mass formula (\ref{MGY}) it follows directly that the electron mass is a minimal rest mass $\mu^0$.

The wave function in (\ref{BDirac}) has the form
\[
\boldsymbol{\psi}(a_1,a_2,a_3,a^\ast_1,a^\ast_2,a^\ast_3)=\sum_{l,m,\dot{l},\dot{m}}\psi_{lm;\dot{l}\dot{m}}
(a_1,a_2,a_3,a^\ast_1,a^\ast_2,a^\ast_3)\left| lm;\dot{l}\dot{m}\right\rangle,
\]
where $l=1/2$, $\dot{l}=1/2$. Or, more explicitly,
\[
\boldsymbol{\psi}=\psi_{\frac{1}{2}\frac{1}{2};\dot{\frac{1}{2}}\dot{\frac{1}{2}}}
\left| \tfrac{1}{2}\tfrac{1}{2};\dot{\tfrac{1}{2}}\dot{\tfrac{1}{2}}\right\rangle+
\psi_{\frac{1}{2}-\frac{1}{2};\dot{\frac{1}{2}}\dot{\frac{1}{2}}}
\left| \tfrac{1}{2}-\tfrac{1}{2};\dot{\tfrac{1}{2}}\dot{\tfrac{1}{2}}\right\rangle+
\psi_{\frac{1}{2}\frac{1}{2};\dot{\frac{1}{2}}-\dot{\frac{1}{2}}}
\left| \tfrac{1}{2}\tfrac{1}{2};\dot{\tfrac{1}{2}}-\dot{\tfrac{1}{2}}\right\rangle+
\psi_{\frac{1}{2}-\frac{1}{2};\dot{\frac{1}{2}}-\dot{\frac{1}{2}}}
\left| \tfrac{1}{2}-\tfrac{1}{2};\dot{\tfrac{1}{2}}-\dot{\tfrac{1}{2}}\right\rangle.
\]
Comparing with (\ref{Bispinor}), we see that
\[
\xi^1\boldsymbol{\varepsilon}_1\;\sim\;\psi_{\frac{1}{2}\frac{1}{2};\dot{\frac{1}{2}}\dot{\frac{1}{2}}}
\left| \tfrac{1}{2}\tfrac{1}{2};\dot{\tfrac{1}{2}}\dot{\tfrac{1}{2}}\right\rangle,\quad
\xi^2\boldsymbol{\varepsilon}_2\;\sim\;\psi_{\frac{1}{2}-\frac{1}{2};\dot{\frac{1}{2}}\dot{\frac{1}{2}}}
\left| \tfrac{1}{2}-\tfrac{1}{2};\dot{\tfrac{1}{2}}\dot{\tfrac{1}{2}}\right\rangle,
\]
\[
\eta_{\dot{1}}\boldsymbol{\varepsilon}^{\dot{1}}\;\sim\;\psi_{\frac{1}{2}\frac{1}{2};\dot{\frac{1}{2}}-\dot{\frac{1}{2}}}
\left| \tfrac{1}{2}\tfrac{1}{2};\dot{\tfrac{1}{2}}-\dot{\tfrac{1}{2}}\right\rangle,\quad
\eta_{\dot{1}}\boldsymbol{\varepsilon}^{\dot{1}}\;\sim\;\psi_{\frac{1}{2}-\frac{1}{2};\dot{\frac{1}{2}}-\dot{\frac{1}{2}}}
\left| \tfrac{1}{2}-\tfrac{1}{2};\dot{\tfrac{1}{2}}-\dot{\tfrac{1}{2}}\right\rangle.
\]
\subsection{Maxwell equations}
The spin-1 chain (spin triplet)
\begin{equation}\label{PhChain}
\unitlength=1mm
\begin{picture}(40,13)
\put(0,5){$\overset{(0,1)}{\bullet}$}
\put(20,5){$\overset{(\frac{1}{2},\frac{1}{2})}{\bullet}$}
\put(40,5){$\overset{(1,0)}{\bullet}$}
\put(3,6){\line(1,0){20}}
\put(23,6){\line(1,0){20}}
\put(-0.5,0){$-1$}
\put(22.5,0){$0$}
\put(42.5,0){$1$}
\put(6,0){$\cdots$}
\put(11.5,0){$\cdots$}
\put(17.5,0){$\cdots$}
\put(26,0){$\cdots$}
\put(31.5,0){$\cdots$}
\put(37.5,0){$\cdots$}
\end{picture}
\end{equation}
on the Fig.\,1, defined within the representation $\boldsymbol{\tau}_{0,1}\oplus\boldsymbol{\tau}_{\frac{1}{2}\frac{1}{2}}\oplus\boldsymbol{\tau}_{1,0}$, leads to Maxwell equations.

As is known, Maxwell equations can be written in a Dirac-like form
\begin{equation}\label{Maxwell}
\left(\frac{i\hbar}{c}\frac{\partial}{\partial t}\begin{array}{||cc||}
0 & I\\
I & 0
\end{array}-i\hbar\frac{\partial}{\partial\bx}\begin{array}{||cc||}
0 & -\alpha_i\\
\alpha_i & 0
\end{array}\right)\begin{array}{||c||}
\psi(x)\\
\psi^\ast(x)
\end{array}=0,
\end{equation}
where `bispinor' (Majorana-Oppenheimer bispinor) is
\begin{equation}\label{MOspinor}
\begin{array}{||c||}
\psi(x)\\
\psi^\ast(x)
\end{array}=\begin{array}{||c||}
\bE-i\bB\\
\bE+i\bB
\end{array}=\begin{array}{||c||}
E_1-iB_1\\
E_2-iB_2\\
E_3-iE_3\\
E_1+iB_1\\
E_2+iB_2\\
E_3+iB_3
\end{array}
\end{equation}
and
\[
\alpha_1=\begin{array}{||ccc||}
0 & 0 & 0\\
0 & 0 & i\\
0 & -i& 0
\end{array},\quad\alpha_2=\begin{array}{||ccc||}
0 & 0 & -i\\
0 & 0 & 0\\
i & 0 & 0
\end{array},\quad\alpha_3=\begin{array}{||ccc||}
0 & i & 0\\
-i& 0 & 0\\
0 & 0 & 0
\end{array}.
\]
From the equation (\ref{Maxwell}) it follows that
\begin{eqnarray}
&&\left(\frac{i\hbar}{c}\frac{\partial}{\partial t}-
i\hbar\alpha_i\frac{\partial}{\partial\bx}\right)\psi(x)=0,\label{ME1}\\
&&\left(\frac{i\hbar}{c}\frac{\partial}{\partial t}+
i\hbar\alpha_i\frac{\partial}{\partial\bx}\right)\psi^\ast(x)=0.\label{ME2}
\end{eqnarray}
The latter equations with allowance for transversality conditions
($\bp\cdot\psi=0$, $\bp\cdot\psi^\ast=0$) coincide with the Maxwell
equations. Indeed, taking into account that
$(\bp\cdot\alpha)\psi=\hbar\nabla\times\psi$, we obtain
\begin{eqnarray}
\frac{i\hbar}{c}\frac{\partial\psi}{\partial
t}&=&-\hbar\nabla\times\psi,
\label{Tr1}\\
-i\hbar\nabla\cdot\psi&=&0.\label{Tr1'}
\end{eqnarray}
Whence
\begin{eqnarray}
\nabla\times(\bE-i\bB)&=&-\frac{i}{c}\frac{\partial(\bE-i\bB)}{\partial
t},
\nonumber\\
\nabla\cdot(\bE-i\bB)&=&0\nonumber
\end{eqnarray}
(the constant $\hbar$ is cancelled). Separating the real and
imaginary parts, we obtain Maxwell equations
\begin{eqnarray}
\nabla\times\bE&=&-\frac{1}{c}\frac{\partial\bB}{\partial t},\nonumber\\
\nabla\times\bB&=&\frac{1}{c}\frac{\partial\bE}{\partial t},\nonumber\\
\nabla\cdot\bE&=&0,\nonumber\\
\nabla\cdot\bB&=&0.\nonumber
\end{eqnarray}
It is easy to verify that we come again to Maxwell equations
starting from the equations
\begin{eqnarray}
\left(\frac{i\hbar}{c}\frac{\partial}{\partial t}+
i\hbar\alpha_i\frac{\partial}{\partial\bx}\right)\psi^\ast(x)&=&0,\label{Tr2}\\
-i\hbar\nabla\cdot\psi^\ast(x)&=&0.\label{Tr2'}
\end{eqnarray}
In spite of the fact that equations (\ref{ME1}) and (\ref{ME2}) lead to the same Maxwell equations, the physical interpretation of
these equations is different (see \cite{Ger98}). Namely, the
equations (\ref{ME1}) and (\ref{ME2}) are equations with negative
and positive helicity, respectively. Thus, the equation (\ref{Maxwell}) can be considered as a wave equation for the photon.

Similarity of Dirac and Maxwell equations can be shown with the most evidence in spinor form. Indeed, massless Dirac equations in spinor form look like \cite{Rum36}
\begin{equation}\label{SpinorD}
\begin{array}{ccc}
\partial_{1\dot{1}}\xi^1+\partial_{1\dot{2}}\xi^2&=&0,\\
\partial_{2\dot{1}}\xi^1+\partial_{2\dot{2}}\xi^2&=&0,\\
\partial^{1\dot{1}}\eta_{\dot{1}}+
\partial^{2\dot{1}}\eta_{\dot{2}}&=&0,\\
\partial^{1\dot{2}}\eta_{\dot{1}}+
\partial^{2\dot{1}}\eta_{\dot{2}}&=&0.
\end{array}
\end{equation}
In turn, spinor form of Maxwell equations in vacuum is
\begin{equation}\label{SpinorM}
\begin{array}{ccc}
\partial_{1\dot{1}}f_{11}+\partial_{1\dot{2}}f_{12}&=&0,\\
\partial_{2\dot{1}}f_{11}+\partial_{2\dot{2}}f_{12}&=&0,\\
\partial^{1\dot{1}}f^{\dot{1}\dot{1}}+
\partial^{2\dot{1}}f^{\dot{1}\dot{2}}&=&0,\\
\partial^{1\dot{2}}f^{\dot{1}\dot{1}}+
\partial^{2\dot{1}}f^{\dot{1}\dot{2}}&=&0.
\end{array}
\end{equation}
Comparison of (\ref{SpinorD}) and (\ref{SpinorM}) shows that these equations are differed by tensor dimension only. Namely, spinors $\xi^\lambda$ and
$\eta_{\dot{\mu}}$ are transformed within
$\boldsymbol{\tau}_{\frac{1}{2},0}$ and
$\boldsymbol{\tau}_{0,\frac{1}{2}}$ representations, whereas the
spintensors $f^{\lambda\mu}$ and $f^{\dot{\lambda}\dot{\mu}}$ are
transformed within $\boldsymbol{\tau}_{1,0}$ and
$\boldsymbol{\tau}_{0,1}$ representations of the Lorentz group.

Furher, with the spin-1 chain (\ref{PhChain}) we have a chain of algebras
\begin{equation}\label{AChain}
\C_2\otimes\C_2\longleftrightarrow\C_2\otimes\overset{\ast}{\C}_2\longleftrightarrow\overset{\ast}{\C}_2\otimes
\overset{\ast}{\C}_2
\end{equation}
and associated spinspace
\begin{equation}\label{SChain}
\dS_2\otimes\dS_2\bigoplus\dS_2\otimes\dot{\dS}_2\bigoplus\dot{\dS}_2\otimes\dot{\dS}_2.
\end{equation}
The representation $\boldsymbol{\tau}_{1,0}\oplus\boldsymbol{\tau}_{\frac{1}{2}\frac{1}{2}}\oplus\boldsymbol{\tau}_{0,1}$ acts in a 10-dimensional symmetric space
\begin{equation}\label{SymChain}
\Sym_{(2,0)}\oplus\Sym_{(1,1)}\oplus\Sym_{(0,2)}.
\end{equation}
When we consider a massless spin-1 particle (photon) the 4-dimensional representation $\boldsymbol{\tau}_{\frac{1}{2}\frac{1}{2}}$, related with the spin state 0, should be omitted\footnote{In 1957, Wigner \cite{Wig57} pointed out physical reasons for absence of the spin state 0 for the massless particles.}. Therefore, (\ref{AChain}), (\ref{SChain}) and (\ref{SymChain}) are reduced to
\[
\C_2\otimes\C_2\longleftrightarrow\overset{\ast}{\C}_2\otimes\overset{\ast}{\C}_2,
\]
\[
\dS_2\otimes\dS_2\bigoplus\dot{\dS}_2\otimes\dot{\dS}_2,
\]
\[
\Sym_{(2,0)}\oplus\Sym_{(0,2)}.
\]
Hence it follows that a photon is described by the massless field of type $(1,0)\oplus(0,1)$.

In the bivector space $\R^6$ the massless field $(1,0)\oplus(0,1)$ satisfies the following equations:
\[
\sum^3_{j=1}\Lambda^{1,0}_j\frac{\partial\boldsymbol{\psi}}{\partial a_j}-
i\sum^3_{j=1}\Lambda^{1,0}_j\frac{\partial\boldsymbol{\psi}}{\partial a^\ast_j}=0,
\]
\begin{equation}\label{BMaxwell}
\sum^3_{j=1}\Lambda^{0,1}_j\frac{\partial\dot{\boldsymbol{\psi}}}{\partial\widetilde{a}_j}+
i\sum^3_{j=1}\Lambda^{0,1}_j\frac{\partial\dot{\boldsymbol{\psi}}}{\partial\widetilde{a}^\ast_j}=0,
\end{equation}
where
\[
\Lambda^{1,0}_1=\frac{\sqrt{2}}{2}c_{11}\begin{array}{||ccc||} 0 & 1 & 0\\
1 & 0 & 1\\
0 & 1 & 0
\end{array},\quad
\Lambda^{1,0}_2=\frac{\sqrt{2}}{2}c_{11}\begin{array}{||ccc||} 0 & -i & 0\\
i & 0 & -i\\
0 & i & 0
\end{array},\quad
\Lambda^{1,0}_3=c_{11}\begin{array}{||ccc||} 1 & 0 & 0\\
0 & 0 & 0\\
0 & 0 & -1
\end{array}.
\]
System (\ref{BMaxwell}) describes photon as a superposition $U$ of state vectors in nonseparable Hilbert space $\bsH^S_2\otimes\bsH^{\overline{0}}\otimes\bsH_\infty$. $\bsH^S_2\otimes\bsH^{\overline{0}}\otimes\bsH_\infty$ is a \textit{truly neutral subspace} of $\bsH^S\otimes\bsH^Q\otimes\bsH_\infty$. In this case there are two spin states: the state 1, described by the representation $\boldsymbol{\tau}_{1,0}$ on the spin-1 line, and the state -1, described by the representation $\boldsymbol{\tau}_{0,1}$ on the dual spin-1 line. Representations $\boldsymbol{\tau}_{1,0}$ and $\boldsymbol{\tau}_{0,1}$ act in the spaces $\Sym_{(2,0)}$ and $\Sym_{(0,2)}$, respectively. Therefore, there exist two state vectors: ket-vector $\left|\Psi\right\rangle=\left.\left|\boldsymbol{\tau}_{1,0},\,\Sym_{(2,0)},\,\C_2\otimes\C_2,\,\dS_{2}\otimes\dS_2
\right.\right\rangle$ (left polarization) and bra-vector $\langle\dot{\dS}_2\otimes\dot{\dS}_2,\,\overset{\ast}{\C}_2\otimes\overset{\ast}{\C}_2,\,\Sym_{(0,2)},\,
\boldsymbol{\tau}_{0,1}|=\langle\dot{\Psi}|$ (right polarization). At the reduction of $U$ we have $U\rightarrow\boldsymbol{\tau}_{1,0}$ (or $U\rightarrow\boldsymbol{\tau}_{0,1}$) and $\bsH^S_2\otimes\bsH^{\overline{0}}\otimes\bsH_\infty\rightarrow\sH_E\simeq\Sym_{(2,0)}$ (or $\bsH^S_2\otimes\bsH^{\overline{0}}\otimes\bsH_\infty\rightarrow\sH_E\simeq\Sym_{(0,2)}$).

The wave function in (\ref{BMaxwell}) has the form
\begin{multline}
\boldsymbol{\psi}=\psi_{11;\dot{1}\dot{1}}\left| 11;\dot{1}\dot{1}\right\rangle+
\psi_{10;\dot{1}\dot{1}}\left| 10;\dot{1}\dot{1}\right\rangle+
\psi_{1-1;\dot{1}\dot{1}}\left| 1-1;\dot{1}\dot{1}\right\rangle+\\
\psi_{11;\dot{1}-\dot{1}}\left| 11;\dot{1}-\dot{1}\right\rangle+
\psi_{11;\dot{1}\dot{0}}\left| 11;\dot{1}\dot{0}\right\rangle+
\psi_{1-1;\dot{1}-\dot{1}}\left| 1-1;\dot{1}-\dot{1}\right\rangle.\nonumber
\end{multline}
In the case of photon this function can be identified with the Majorana-Oppenheimer bispinor (\ref{MOspinor}):
\[
\boldsymbol{\psi}=\begin{array}{||c||}
\psi_{11;\dot{1}\dot{1}}\\
\psi_{10;\dot{1}\dot{1}}\\
\psi_{1-1;\dot{1}\dot{1}}\\
\psi_{11;\dot{1}-\dot{1}}\\
\psi_{11;\dot{1}\dot{0}}\\
\psi_{1-1;\dot{1}-\dot{1}}
\end{array}=\begin{array}{||c||}
E_1-iB_1\\
E_2-iB_2\\
E_3-iE_3\\
E_1+iB_1\\
E_2+iB_2\\
E_3+iB_3
\end{array}.
\]

\section{The structure of RWE}
In this section we will study structure of indecomposable equations (\ref{BS2}) related by a general interlocking scheme\footnote{Structure of decomposable equations of the type (\ref{BS2}), that correspond to unstable particles, will be studied in a separate work.}. This interlocking scheme corresponds to a some nonseparable state (superposition) in the space $\bsH^S\otimes\bsH^Q\otimes\bsH_\infty$. The each state vector of the given superposition describes one from the possible (at the reduction) actualized states of some particle of the spin $s=|l-\dot{l}|$. In virtue of the commutation relations (\ref{Commut}) the operator $\boldsymbol{\Lambda}=(\Lambda^{l\dot{l}}_1,\Lambda^{l\dot{l}}_2,\Lambda^{l\dot{l}}_3)$ has a common system of eigenfunctions with the energy operator $H$ (and operators $\sX_l$, $\sY_l$ of the complex momentum) for the each value $s=l-\dot{l}$. Moreover, hence it follows that $\boldsymbol{\Lambda}$ is a Hermitian operator. The following theorem gives a structure of the operator $\boldsymbol{\Lambda}$ (that is, a structure of the energy spectrum) in dependence on the structure of elementary divisors.
\begin{thm}
1) The fields of type $(l,0)\oplus(0,\dot{l})$.\\
Let $\boldsymbol{\Lambda}=(\Lambda^{l,0}_1,\Lambda^{l,0}_2,\Lambda^{l,0}_3)$ and $\overset{\ast}{\boldsymbol{\Lambda}}=(\Lambda^{0,\dot{l}}_1,\Lambda^{0,\dot{l}}_2,\Lambda^{0,\dot{l}}_3)$ be linear Hermitian operators, and all the roots of characteristic equation $\boldsymbol{\Delta}=0$ belong to the field $\F=\R$. In the case of $(l,0)\oplus(0,\dot{l})$ all elementary divisors of the operators $\boldsymbol{\Lambda}$ and $\overset{\ast}{\boldsymbol{\Lambda}}$ are simple, and Jordan form of the matrices of $\boldsymbol{\Lambda}$ and $\overset{\ast}{\boldsymbol{\Lambda}}$ is a diagonal matrix
\begin{equation}\label{SimpleJ}
\bJ(\Lambda)=-\bJ(\overset{\ast}{\Lambda})=\text{{\rm diag}}\left(\lambda_1,\lambda_2,\ldots,\lambda_{2l+1}\right),
\end{equation}
where $\lambda_1,\lambda_2,\ldots,\lambda_{2l+1}$ are eigenvalues of the operators $\boldsymbol{\Lambda}$ and $\overset{\ast}{\boldsymbol{\Lambda}}$.\\
2) The fields of type $(l,\dot{l})\oplus(\dot{l},l)$.\\
In this case linear Hermitian operators $\boldsymbol{\Lambda}=(\Lambda^{l\dot{l}}_1,\Lambda^{l\dot{l}}_2,\Lambda^{l\dot{l}}_3)$ and $\overset{\ast}{\boldsymbol{\Lambda}}=(\Lambda^{\dot{l}l}_1,\Lambda^{\dot{l}l}_2,\Lambda^{\dot{l}l}_3)$ have elementary divisors with multiple roots, and Jordan form of the matrices of $\boldsymbol{\Lambda}$ and $\overset{\ast}{\boldsymbol{\Lambda}}$ is a block-diagonal matrix
\begin{equation}\label{MultipleJ}
\bJ(\Lambda)=-\bJ(\overset{\ast}{\Lambda})=\text{{\rm diag}}\left(\bJ({}_1\Lambda), \bJ({}_2\Lambda),\ldots,\bJ({}_p\Lambda),\bJ({}_{-1}\Lambda),\bJ({}_{-2}\Lambda),\ldots,\bJ({}_{-p}\Lambda)\right),
\end{equation}
where
\[
\bJ({}_p\Lambda)=\begin{Vmatrix}
\lambda_p & 1 & 0 & \dots & 0\\
0 & \lambda_p & 1 & \dots & 0\\
.&&\hdotsfor[3]{3}\\
.&&&\hdotsfor[3]{2}\\
.&&&.&1\\
0 & 0 & 0 & \dots & \lambda_p
\end{Vmatrix}
\]
is a Jordan cell, corresponding to elementary divisor $(\lambda-\lambda_p)^{m_p}$, $2<m_p<m^p_{\text{max}}$, $\lambda_p$ is an eigenvalue of $\Lambda^{s,\frac{k}{2}}$ ($\Lambda^{\frac{k}{2},s}$), $m^p_{\text{max}}$ is a number of products among $(2s+1)(k+1)$ products of the $sk$-basis, which equal to eigenvalue $\lambda_p$.
\end{thm}
\begin{proof} I) \textit{The fields of type $(l,0)\oplus(0,l)$}.
Let us consider a Jordan form $\bJ(\Lambda^{l\dot{l}}_j)$ of the matrices $\Lambda^{l\dot{l}}_j$ for the first simplest cases of the fields of type $(l,0)\oplus(0,l)$.

1) The field $(1/2,0)\oplus(0,1/2)$. This qubit-like field is formulated within the spin-1/2 chain (\ref{EChain}). An explicit form of $\Lambda^{\frac{1}{2}}_j$ we obtain from (\ref{L3})--(\ref{L2}) at $l=1/2$:
\[
\Lambda^{\frac{1}{2},0}_1=\frac{1}{2}c_{\frac{1}{2}\frac{1}{2}}\begin{array}{||cc||} 0 & 1\\
1 & 0
\end{array},\quad
\Lambda^{\frac{1}{2},0}_2=\frac{1}{2}c_{\frac{1}{2}\frac{1}{2}}\begin{array}{||cc||} 0 & -i\\
i & 0
\end{array},\quad
\Lambda^{\frac{1}{2},0}_3=\frac{1}{2}c_{\frac{1}{2}\frac{1}{2}}\begin{array}{||cc||} 1 & 0\\
0 & -1
\end{array}.
\]
Characteristic polynomials are (at $c_{\frac{1}{2}\frac{1}{2}}=1$)
\[
\boldsymbol{\Delta}_1(\lambda)=\left|\lambda\boldsymbol{1}_2-\Lambda^{\frac{1}{2},0}_1\right|=\begin{vmatrix}
\lambda & -\frac{1}{2}\\
-\frac{1}{2} & \lambda
\end{vmatrix}=\lambda^2-\frac{1}{4}=(\lambda-1/2)(\lambda+1/2),
\]
\[
\boldsymbol{\Delta}_2(\lambda)=\left|\lambda\boldsymbol{1}_2-\Lambda^{\frac{1}{2},0}_2\right|=\begin{vmatrix}
\lambda & \frac{i}{2}\\
-\frac{i}{2} & \lambda
\end{vmatrix}=\lambda^2-\frac{1}{4}=(\lambda-1/2)(\lambda+1/2),
\]
\[
\boldsymbol{\Delta}_3(\lambda)=\left|\lambda\boldsymbol{1}_2-\Lambda^{\frac{1}{2},0}_3\right|=\begin{vmatrix}
\lambda-\frac{1}{2} & 0\\
0 & \lambda+\frac{1}{2}
\end{vmatrix}=\lambda^2-\frac{1}{4}=(\lambda-1/2)(\lambda+1/2).
\]
Therefore,
\[
\bJ(\Lambda^{\frac{1}{2},0}_1)=\bJ(\Lambda^{\frac{1}{2},0}_2)=\bJ(\Lambda^{\frac{1}{2},0}_3)=
\begin{array}{||cc||}
\frac{1}{2} & 0\\
0 & -\frac{1}{2}
\end{array}.
\]

2) The field $(1,0)\oplus(0,1)$. This field is formulated within the spin-1 chain (\ref{PhChain}). An explicit form of $\Lambda^{1,0}_j$ we obtain from (\ref{L3})--(\ref{L2}) at $l=1$:
\[
\Lambda^{1,0}_1=\frac{\sqrt{2}}{2}c_{11}\begin{array}{||ccc||} 0 & 1 & 0\\
1 & 0 & 1\\
0 & 1 & 0
\end{array},\quad
\Lambda^{1,0}_2=\frac{\sqrt{2}}{2}c_{11}\begin{array}{||ccc||} 0 & -i & 0\\
i & 0 & -i\\
0 & i & 0
\end{array},\quad
\Lambda^{1,0}_3=c_{11}\begin{array}{||ccc||} 1 & 0 & 0\\
0 & 0 & 0\\
0 & 0 & -1
\end{array}.
\]
In this case characteristic polynomials are (at $c_{11}=\sqrt{2}$)
\[
\boldsymbol{\Delta}_1(\lambda)=\left|\lambda\boldsymbol{1}_3-\Lambda^{1,0}_1\right|=\begin{vmatrix}
\lambda & -1 & 0\\
-1 & \lambda & -1\\
0 & -1 & \lambda
\end{vmatrix}=\lambda(\lambda-\sqrt{2})(\lambda+\sqrt{2}),
\]
\[
\boldsymbol{\Delta}_2(\lambda)=\left|\lambda\boldsymbol{1}_3-\Lambda^{1,0}_2\right|=\begin{vmatrix}
\lambda & i & 0\\
-i & \lambda & i\\
0 & -i & \lambda
\end{vmatrix}=\lambda(\lambda-\sqrt{2})(\lambda+\sqrt{2}),
\]
\[
\boldsymbol{\Delta}_3(\lambda)=\left|\lambda\boldsymbol{1}_3-\Lambda^{1,0}_3\right|=\begin{vmatrix}
\lambda-\sqrt{2} & 0 & 0\\
0 & \lambda & 0\\
0 & 0 & \lambda+\sqrt{2}
\end{vmatrix}=\lambda(\lambda-\sqrt{2})(\lambda+\sqrt{2}).
\]
Therefore, after dividing on $\sqrt{2}$, we have
\[
\bJ(\Lambda^{1,0}_1)=\bJ(\Lambda^{1,0}_2)=\bJ(\Lambda^{1,0}_3)=
\begin{array}{||ccc||}
1 & 0 & 0\\
0 & 0 & 0\\
0 & 0 &-1
\end{array}.
\]

3) The field $(3/2,0)\oplus(0,3/2)$. In its turn, this field is defined within the following spin-3/2 chain (spin quadruplet):
\begin{equation}\label{3Chain}
\unitlength=1mm
\begin{picture}(60,13)
\put(0,5){$\overset{(0,\frac{3}{2})}{\bullet}$}
\put(20,5){$\overset{(\frac{1}{2},1)}{\bullet}$}
\put(40,5){$\overset{(1,\frac{1}{2})}{\bullet}$}
\put(60,5){$\overset{(\frac{3}{2},0)}{\bullet}$}
\put(3,6){\line(1,0){20}}
\put(23,6){\line(1,0){20}}
\put(43,6){\line(1,0){20}}
\put(-0.5,0){$-\frac{3}{2}$}
\put(19.5,0){$-\frac{1}{2}$}
\put(42.5,0){$\frac{1}{2}$}
\put(62.5,0){$\frac{3}{2}$}
\put(6,0){$\cdots$}
\put(11.5,0){$\cdots$}
\put(17.5,0){$\cdot$}
\put(26,0){$\cdots$}
\put(31.5,0){$\cdots$}
\put(37.5,0){$\cdots$}
\put(46,0){$\cdots$}
\put(51.5,0){$\cdots$}
\put(57.5,0){$\cdots$}
\end{picture}
\end{equation}
From (\ref{L3})--(\ref{L2}) it follows that matrices $\Lambda^{\frac{3}{2},0}_j$ have the form
\[
\Lambda^{\frac{3}{2},0}_1=c_{\frac{3}{2}\frac{3}{2}}\begin{array}{||cccc||} 0 &
\frac{\sqrt{3}}{2} & 0 & 0\\
\frac{\sqrt{3}}{2} & 0 & 1 & 0\\
0 & 1 & 0 & \frac{\sqrt{3}}{2}\\
0 & 0 & \frac{\sqrt{3}}{2} & 0
\end{array},\quad
\Lambda^{\frac{3}{2},0}_2=c_{\frac{3}{2}\frac{3}{2}}\begin{array}{||cccc||} 0 &
-i\frac{\sqrt{3}}{2} & 0 & 0\\
i\frac{\sqrt{3}}{2} & 0 & -i & 0\\
0 & i & 0 & -i\frac{\sqrt{3}}{2}\\
0 & 0 & i\frac{\sqrt{3}}{2} & 0
\end{array},
\]
\[
\Lambda^{\frac{3}{2},0}_3=\frac{1}{2}c_{\frac{3}{2}\frac{3}{2}}\begin{array}{||cccc||} 3 & 0 & 0 & 0\\
0 & 1 & 0 & 0\\
0 & 0 & -1 & 0\\
0 & 0 & 0 & -3
\end{array}.
\]
Characteristic polynomials are (at $c_{\frac{3}{2}\frac{3}{2}}=1$)
\[
\boldsymbol{\Delta}_1(\lambda)=\left|\lambda\boldsymbol{1}_4-\Lambda^{\frac{3}{2},0}_1\right|=\begin{vmatrix}
\lambda & -\frac{\sqrt{3}}{2} & 0 & 0\\
-\frac{\sqrt{3}}{2} & \lambda & -1 & 0\\
0 & -1 & \lambda & -\frac{\sqrt{3}}{2}\\
0 & 0 & -\frac{\sqrt{3}}{2} & \lambda
\end{vmatrix}=\lambda^4-\frac{5}{4}\lambda^2+\frac{9}{16},
\]
\[
\boldsymbol{\Delta}_2(\lambda)=\left|\lambda\boldsymbol{1}_4-\Lambda^{\frac{3}{2},0}_2\right|=\begin{vmatrix}
\lambda & i\frac{\sqrt{3}}{2} & 0 & 0\\
-i\frac{\sqrt{3}}{2} & \lambda & i & 0\\
0 & -i & \lambda & i\frac{\sqrt{3}}{2}\\
0 & 0 & -i\frac{\sqrt{3}}{2} & \lambda
\end{vmatrix}=\lambda^4-\frac{5}{4}\lambda^2+\frac{9}{16},
\]
\[
\boldsymbol{\Delta}_3(\lambda)=\left|\lambda\boldsymbol{1}_4-\Lambda^{\frac{3}{2},0}_3\right|=\begin{vmatrix}
\lambda-\frac{3}{2} & 0 & 0 & 0\\
0 & \lambda-\frac{1}{2} & 0 & 0\\
0 & 0 & \lambda+\frac{1}{2} & 0\\
0 & 0 & 0 & \lambda+\frac{3}{2}
\end{vmatrix}=(\lambda-\tfrac{3}{2})(\lambda-\tfrac{1}{2})(\lambda+\tfrac{1}{2})(\lambda+\tfrac{3}{2}).
\]
Multiplying by 16 the equation $\boldsymbol{\Delta}_1(\lambda)=0$ (or $\boldsymbol{\Delta}_2(\lambda)=0$), we obtain
\[
16\lambda^4-40\lambda^2+9=0.
\]
The roots of this biquadratic equation are $\lambda_1=3/2$, $\lambda_2=-3/2$, $\lambda_3=1/2$, $\lambda_4=-1/2$. Therefore,
\[
\bJ(\Lambda^{\frac{3}{2},0}_1)=\bJ(\Lambda^{\frac{3}{2},0}_2)=\bJ(\Lambda^{\frac{3}{2},0}_3)=
\begin{array}{||cccc||}
\frac{3}{2} & 0 & 0 & 0\\
0 & \frac{1}{2} & 0 & 0\\
0 & 0 & -\frac{1}{2} & 0\\
0 & 0 & 0 & -\frac{3}{2}
\end{array}.
\]

4) The field $(2,0)\oplus(0,2)$. This field is defined within the following spin-2 chain (spin 5-plet):
\begin{equation}\label{5Chain}
\unitlength=1mm
\begin{picture}(80,13)
\put(0,5){$\overset{(0,2)}{\bullet}$}
\put(20,5){$\overset{(\frac{1}{2},\frac{3}{2})}{\bullet}$}
\put(40,5){$\overset{(1,1)}{\bullet}$}
\put(60,5){$\overset{(\frac{3}{2},\frac{1}{2})}{\bullet}$}
\put(80,5){$\overset{(2,0)}{\bullet}$}
\put(3,6){\line(1,0){20}}
\put(23,6){\line(1,0){20}}
\put(43,6){\line(1,0){20}}
\put(63,6){\line(1,0){20}}
\put(-0.5,0){$-2$}
\put(19.5,0){$-1$}
\put(42.5,0){$0$}
\put(62.5,0){$1$}
\put(82.5,0){$2$}
\put(6,0){$\cdots$}
\put(11.5,0){$\cdots$}
\put(17.5,0){$\cdot$}
\put(26,0){$\cdots$}
\put(31.5,0){$\cdots$}
\put(37.5,0){$\cdots$}
\put(46,0){$\cdots$}
\put(51.5,0){$\cdots$}
\put(57.5,0){$\cdots$}
\put(66,0){$\cdots$}
\put(71.5,0){$\cdots$}
\put(77.5,0){$\cdots$}
\end{picture}
\end{equation}
In this case the matrices $\Lambda^{2,0}_j$ have the form
\[
\Lambda^{2,0}_1=c_{22}\begin{array}{||ccccc||} 0 & 1 & 0 & 0
& 0\\
1 & 0 & \frac{\sqrt{6}}{2} & 0 & 0\\
0 & \frac{\sqrt{6}}{2} & 0 & \frac{\sqrt{6}}{2} & 0\\
0 & 0 & \frac{\sqrt{6}}{2} & 0 & 1\\
0 & 0 & 0 & 1 & 0
\end{array},\quad
\Lambda^{2,0}_2=c_{22}\begin{array}{||ccccc||}0 & -i & 0 & 0
& 0\\
i & 0 & -i\frac{\sqrt{6}}{2} & 0 & 0\\
0 & i\frac{\sqrt{6}}{2} & 0 & -i\frac{\sqrt{6}}{2} & 0\\
0 & 0 & i\frac{\sqrt{6}}{2} & 0 & -i\\
0 & 0 & 0 & i & 0
\end{array},
\]
\[
\Lambda^{2,0}_3=c_{22}\begin{array}{||ccccc||} 2 & 0 & 0 & 0
& 0\\
0 & 1 & 0 & 0 & 0\\
0 & 0 & 0 & 0 & 0\\
0 & 0 & 0 & -1 & 0\\
0 & 0 & 0 & 0 & -2
\end{array}.
\]
At $c_{22}=1$ for the characteristic polynomials we have
\[
\boldsymbol{\Delta}_1(\lambda)=\left|\lambda\boldsymbol{1}_5-\Lambda^{2,0}_1\right|=\begin{vmatrix}
\lambda & -1 & 0 & 0 & 0\\
-1 & \lambda & -\frac{\sqrt{6}}{2} & 0 & 0\\
0 & -\frac{\sqrt{6}}{2} & \lambda & -\frac{\sqrt{6}}{2} & 0\\
0 & 0 & -\frac{\sqrt{6}}{2} & \lambda & -1\\
0 & 0 & 0 & -1 & \lambda
\end{vmatrix}=\lambda^5-5\lambda^3+4\lambda,
\]
\[
\boldsymbol{\Delta}_2(\lambda)=\left|\lambda\boldsymbol{1}_5-\Lambda^{2,0}_2\right|=\begin{vmatrix}
\lambda & i & 0 & 0 & 0\\
-i & \lambda & i\frac{\sqrt{6}}{2} & 0 & 0\\
0 & -i\frac{\sqrt{6}}{2} & \lambda & i\frac{\sqrt{6}}{2} & 0\\
0 & 0 & -i\frac{\sqrt{6}}{2} & \lambda & i\\
0 & 0 & 0 & -i & \lambda
\end{vmatrix}=\lambda^5-5\lambda^3+4\lambda,
\]
\[
\boldsymbol{\Delta}_3(\lambda)=\left|\lambda\boldsymbol{1}_5-\Lambda^{2,0}_3\right|=\begin{vmatrix}
\lambda-2 & 0 & 0 & 0 & 0\\
0 & \lambda-1 & 0 & 0 & 0\\
0 & 0 & \lambda & 0 & 0\\
0 & 0 & 0 & \lambda+1 & 0\\
0 & 0 & 0 & 0 & \lambda+2
\end{vmatrix}=(\lambda-2)(\lambda-1)\lambda(\lambda+1)(\lambda+2).
\]
In this case, the equation $\boldsymbol{\Delta}_1(\lambda)=0$ (or $\boldsymbol{\Delta}_2(\lambda)=0$) leads to
\[
\lambda^5-5\lambda^3+4\lambda=0.
\]
The roots of this equation are $\lambda_1=0$, $\lambda_2=-1$, $\lambda_3=1$, $\lambda_4=-2$, $\lambda_5=2$. Thus,
\[
\bJ(\Lambda^{2,0}_1)=\bJ(\Lambda^{2,0}_2)=\bJ(\Lambda^{2,0}_3)=
\begin{array}{||ccccc||}
2 & 0 & 0 & 0 & 0\\
0 & 1 & 0 & 0 & 0\\
0 & 0 & 0 & 0 & 0\\
0 & 0 & 0 & -1& 0\\
0 & 0 & 0 & 0 & -2
\end{array}.
\]

5) The field $(5/2,0)\oplus(0,5/2)$. In its turn, this field is formulated within the following spin-5/2 chain (spin 6-plet):
\begin{equation}\label{6Chain}
\unitlength=1mm
\begin{picture}(100,13)
\put(0,5){$\overset{(0,\frac{5}{2})}{\bullet}$}
\put(20,5){$\overset{(\frac{1}{2},2)}{\bullet}$}
\put(40,5){$\overset{(1,\frac{3}{2})}{\bullet}$}
\put(60,5){$\overset{(\frac{3}{2},1)}{\bullet}$}
\put(80,5){$\overset{(2,\frac{1}{2})}{\bullet}$}
\put(100,5){$\overset{(\frac{5}{2},0)}{\bullet}$}
\put(3,6){\line(1,0){20}}
\put(23,6){\line(1,0){20}}
\put(43,6){\line(1,0){20}}
\put(63,6){\line(1,0){20}}
\put(83,6){\line(1,0){20}}
\put(-0.5,0){$-\frac{5}{2}$}
\put(19.5,0){$-\frac{3}{2}$}
\put(39.5,0){$-\frac{1}{2}$}
\put(62.5,0){$\frac{1}{2}$}
\put(82.5,0){$\frac{3}{2}$}
\put(102.5,0){$\frac{5}{2}$}
\put(6,0){$\cdots$}
\put(11.5,0){$\cdots$}
\put(17.5,0){$\cdot$}
\put(26,0){$\cdots$}
\put(31.5,0){$\cdots$}
\put(37.5,0){$\cdot$}
\put(46,0){$\cdots$}
\put(51.5,0){$\cdots$}
\put(57.5,0){$\cdots$}
\put(66,0){$\cdots$}
\put(71.5,0){$\cdots$}
\put(77.5,0){$\cdots$}
\put(86,0){$\cdots$}
\put(91.5,0){$\cdots$}
\put(97.5,0){$\cdots$}
\end{picture}
\end{equation}
An explicit form of $\Lambda^{\frac{5}{2},0}_j$ we obtain from (\ref{L3})--(\ref{L2}) at $l=5/2$:
\[
\Lambda^{\frac{5}{2},0}_1=c_{\frac{5}{2}\frac{5}{2}}\begin{array}{||cccccc||} 0 &
\frac{\sqrt{5}}{2} & 0 & 0 & 0 & 0\\
\frac{\sqrt{5}}{2} & 0 & \sqrt{2} & 0 & 0 & 0\\
0 & \sqrt{2} & 0 & \frac{3}{2} & 0 & 0\\
0 & 0 & \frac{3}{2} & 0 & \sqrt{2} & 0\\
0 & 0 & 0 & \sqrt{2} & 0 & \frac{\sqrt{5}}{2}\\
0 & 0 & 0 & 0 & \frac{\sqrt{5}}{2} & 0
\end{array},\quad
\Lambda^{\frac{5}{2},0}_2=c_{\frac{5}{2}\frac{5}{2}}\begin{array}{||cccccc||} 0 &
-i\frac{\sqrt{5}}{2} & 0 & 0 & 0 & 0\\
i\frac{\sqrt{5}}{2} & 0 & -i\sqrt{2} & 0 & 0 & 0\\
0 & i\sqrt{2} & 0 & -i\frac{3}{2} & 0 & 0\\
0 & 0 & i\frac{3}{2} & 0 & -i\sqrt{2} & 0\\
0 & 0 & 0 & i\sqrt{2} & 0 & -i\frac{\sqrt{5}}{2}\\
0 & 0 & 0 & 0 & i\frac{\sqrt{5}}{2} & 0
\end{array},
\]
\[
\Lambda^{\frac{5}{2},0}_3=\frac{c_{\frac{5}{2}\frac{5}{2}}}{2}\begin{array}{||cccccc||} 5 & 0 & 0 &
0 & 0 & 0\\
0 & 3 & 0 & 0 & 0 & 0\\
0 & 0 & 1 & 0 & 0 & 0\\
0 & 0 & 0 & -1& 0 & 0\\
0 & 0 & 0 & 0 & -3& 0\\
0 & 0 & 0 & 0 & 0 & -5
\end{array}.
\]
Characteristic polynomials are (at $c_{\frac{5}{2}\frac{5}{2}}=1$)
\[
\boldsymbol{\Delta}_1(\lambda)=\left|\lambda\boldsymbol{1}_6-\Lambda^{\frac{5}{2},0}_1\right|=\begin{vmatrix}
\lambda & -\frac{\sqrt{5}}{2} & 0 & 0 & 0 & 0\\
-\frac{\sqrt{5}}{2} & \lambda & -\sqrt{2} & 0 & 0 & 0\\
0 & -\sqrt{2} & \lambda & -\frac{3}{2} & 0 & 0\\
0 & 0 & -\frac{3}{2} & \lambda & -\sqrt{2} & 0\\
0 & 0 & 0 & -\sqrt{2} & \lambda & -\frac{\sqrt{5}}{2}\\
0 & 0 & 0 & 0 & -\frac{\sqrt{5}}{2} & \lambda
\end{vmatrix}=\lambda^6-\frac{35}{4}\lambda^4+\frac{259}{16}\lambda^2-\frac{225}{64},
\]
\[
\boldsymbol{\Delta}_2(\lambda)=\left|\lambda\boldsymbol{1}_6-\Lambda^{\frac{5}{2},0}_2\right|=\begin{vmatrix}
\lambda & i\frac{\sqrt{5}}{2} & 0 & 0 & 0 & 0\\
-i\frac{\sqrt{5}}{2} & \lambda & i\sqrt{2} & 0 & 0 & 0\\
0 & -i\sqrt{2} & \lambda & i\frac{3}{2} & 0 & 0\\
0 & 0 & -i\frac{3}{2} & \lambda & i\sqrt{2} & 0\\
0 & 0 & 0 & -i\sqrt{2} & \lambda & i\frac{\sqrt{5}}{2}\\
0 & 0 & 0 & 0 & -i\frac{\sqrt{5}}{2} & \lambda
\end{vmatrix}=\lambda^6-\frac{35}{4}\lambda^4+\frac{259}{16}\lambda^2-\frac{225}{64},
\]
\begin{multline}
\boldsymbol{\Delta}_3(\lambda)=\left|\lambda\boldsymbol{1}_6-\Lambda^{\frac{5}{2},0}_3\right|=\begin{vmatrix}
\lambda-\frac{5}{2} & 0 & 0 & 0 & 0 & 0\\
0 & \lambda-\frac{3}{2} & 0 & 0 & 0 & 0\\
0 & 0 & \lambda-\frac{1}{2} & 0 & 0 & 0\\
0 & 0 & 0 & \lambda+\frac{1}{2} & 0 & 0\\
0 & 0 & 0 & 0 & \lambda+\frac{3}{2} & 0\\
0 & 0 & 0 & 0 & 0 & \lambda+\frac{5}{2}
\end{vmatrix}=\\
(\lambda-\tfrac{5}{2})(\lambda-\tfrac{3}{2})(\lambda-\tfrac{1}{2})(\lambda+\tfrac{1}{2})
(\lambda+\tfrac{3}{2})(\lambda+\tfrac{5}{2}).\nonumber
\end{multline}
Multiplying by 64 the equation $\boldsymbol{\Delta}_1(\lambda)=0$ (or $\boldsymbol{\Delta}_2(\lambda)=0$), we obtain
\[
64\lambda^6-560\lambda^4+1036\lambda^2-225=0.
\]
The roots of this equation are $\lambda_1=5/2$, $\lambda_2=3/2$, $\lambda_3=1/2$, $\lambda_4=-1/2$, $\lambda_5=-3/2$, $\lambda_6=-5/2$. Therefore,
\[
\bJ(\Lambda^{\frac{5}{2},0}_1)=\bJ(\Lambda^{\frac{5}{2},0}_2)=\bJ(\Lambda^{\frac{5}{2},0}_3)=
\begin{array}{||cccccc||}
\frac{5}{2} & 0 & 0 & 0 & 0 & 0\\
0 & \frac{3}{2} & 0 & 0 & 0 & 0\\
0 & 0 & \frac{1}{2} & 0 & 0 & 0\\
0 & 0 & 0 & -\frac{1}{2}& 0 & 0\\
0 & 0 & 0 & 0 & -\frac{3}{2}& 0\\
0 & 0 & 0 & 0 & 0 & -\frac{5}{2}
\end{array}.
\]

Further, moving up on the representation cone (Fig.\,1), we see that for the fields of type $(l,0)$ the characteristic equation
\[
\boldsymbol{\Delta}_1(\lambda)=\left|\lambda\boldsymbol{1}_{2l+1}-\Lambda^{l,0}_1\right|=0
\]
leads to
\begin{equation}\label{Aeq}
a_m\lambda^m+a_{m-2}\lambda^{m-2}+a_{m-4}\lambda^{m-4}+\ldots+a_0=0
\end{equation}
when $l$ is a half-integer number, $m=2l+1$ is an even number (dimensionality of the space $\Sym_{(k,0)}$), and to
\begin{equation}\label{Beq}
b_m\lambda^m+b_{m-2}\lambda^{m-2}+b_{m-4}\lambda^{m-4}+\ldots+b_1\lambda=0
\end{equation}
when $l$ is an integer number, $m=2l+1$ is an odd number. In both cases the roots of (\ref{Aeq}) and (\ref{Beq}) are $-l$, $-l+1$, $-l+2$, $\ldots$, $l$. Therefore, Jordan form of the field $(l,0)$ is defined by (\ref{SimpleJ}). It is obvious that the same Jordan form we have for the field $(0,\dot{l})$.

II) \textit{The fields of type $(l,\dot{l})$}.
Let us consider in detail the structure of the operators $\boldsymbol{\Lambda}=(\Lambda^{l\dot{l}}_1,\Lambda^{l\dot{l}}_2,\Lambda^{l\dot{l}}_3)$ and $\overset{\ast}{\boldsymbol{\Lambda}}=(\Lambda^{\dot{l}l}_1,\Lambda^{\dot{l}l}_2,\Lambda^{\dot{l}l}_3)$ for the fields of type $(l,\dot{l})$, $(\dot{l},l)$ and $(l,\dot{l})\oplus(\dot{l},l)$ (arbitrary spin chains, that is, the fields which occurring inside the representation cone on the Fig.\,1). For simplicity we consider structure of the operators $\boldsymbol{\Lambda}$ and $\overset{\ast}{\boldsymbol{\Lambda}}$ on the spin-1/2 and dual spin-1/2 lines (for other spin lines all the calculations are analogous). Inside the representation cone we see that spin-1/2 and dual spin-1/2 lines form two different spin chains (spin doublets): $\left(s,\tfrac{k}{2}\right)\longleftrightarrow\left(\tfrac{k}{2},s\right)$ and $\left(\tfrac{k}{2},s\right)\longleftrightarrow\left(s,\tfrac{k}{2}\right)$. Helicity basis for the state $\left(s,\tfrac{k}{2}\right)$ consists of the following ket-vectors:
\[
\left| ss;\tfrac{k}{2}\tfrac{k}{2}\right\rangle,\;\left| s,s-1;\tfrac{k}{2}\tfrac{k}{2}\right\rangle,\;\ldots,\;
\left| s,1;\tfrac{k}{2}\tfrac{k}{2}\right\rangle,\;\left| s,0;\tfrac{k}{2}\tfrac{k}{2}\right\rangle,\;\ldots
\left| s,-s;\tfrac{k}{2}\tfrac{k}{2}\right\rangle;
\]
\[
\left| ss;\tfrac{k}{2}\tfrac{k-2}{2}\right\rangle,\;\left| s,s-1;\tfrac{k}{2}\tfrac{k-2}{2}\right\rangle,\;\ldots,\;
\left| s,1;\tfrac{k}{2}\tfrac{k-2}{2}\right\rangle,\;\left| s,0;\tfrac{k}{2}\tfrac{k-2}{2}\right\rangle,\;\ldots
\left| s,-s;\tfrac{k}{2}\tfrac{k-2}{2}\right\rangle;
\]
\[
\ldots\ldots\ldots\ldots\ldots\ldots\ldots\ldots\ldots\ldots\ldots\ldots\ldots\ldots\ldots\ldots
\]
\[
\left| ss;\tfrac{k}{2}\tfrac{1}{2}\right\rangle,\;\left| s,s-1;\tfrac{k}{2}\tfrac{1}{2}\right\rangle,\;\ldots,\;
\left| s,1;\tfrac{k}{2}\tfrac{1}{2}\right\rangle,\;\left| s,0;\tfrac{k}{2}\tfrac{1}{2}\right\rangle,\;\ldots
\left| s,-s;\tfrac{k}{2}\tfrac{1}{2}\right\rangle;
\]
\[
\left| ss;\tfrac{k}{2},-\tfrac{1}{2}\right\rangle,\;\left| s,s-1;\tfrac{k}{2},-\tfrac{1}{2}\right\rangle,\;\ldots,\;
\left| s,1;\tfrac{k}{2},-\tfrac{1}{2}\right\rangle,\;\left| s,0;\tfrac{k}{2},-\tfrac{1}{2}\right\rangle,\;\ldots
\left| s,-s;\tfrac{k}{2},-\tfrac{1}{2}\right\rangle;
\]
\[
\ldots\ldots\ldots\ldots\ldots\ldots\ldots\ldots\ldots\ldots\ldots\ldots\ldots\ldots\ldots\ldots
\]
\begin{equation}\label{Ket_basis}
\left| ss;\tfrac{k}{2},-\tfrac{k}{2}\right\rangle,\;\left| s,s-1;\tfrac{k}{2},-\tfrac{k}{2}\right\rangle,\;\ldots,\;
\left| s,1;\tfrac{k}{2},-\tfrac{k}{2}\right\rangle,\;\left| s,0;\tfrac{k}{2},-\tfrac{k}{2}\right\rangle,\;\ldots
\left| s,-s;\tfrac{k}{2},-\tfrac{k}{2}\right\rangle.
\end{equation}
In its turn, helicity basis for the state $\left(\tfrac{k}{2},s\right)$ consists of the following bra-vectors:
\[
\left\langle\tfrac{k}{2}\tfrac{k}{2};ss\right|,\;\left\langle\tfrac{k}{2}\tfrac{k-2}{2};ss\right|,\;\ldots,\;
\left\langle\tfrac{k}{2},\tfrac{1}{2};ss\right|,\;\left\langle\tfrac{k}{2},-\tfrac{1}{2};ss\right|,\;\ldots,\;
\left\langle\tfrac{k}{2},-\tfrac{k}{2};ss\right|;
\]
\[
\left\langle\tfrac{k}{2}\tfrac{k}{2};s,s-1\right|,\;\left\langle\tfrac{k}{2}\tfrac{k-2}{2};s,s-1\right|,\;\ldots,\;
\left\langle\tfrac{k}{2},\tfrac{1}{2};s,s-1\right|,\;\left\langle\tfrac{k}{2},-\tfrac{1}{2};s,s-1\right|,\;\ldots,\;
\left\langle\tfrac{k}{2},-\tfrac{k}{2};s,s-1\right|;
\]
\[
\ldots\ldots\ldots\ldots\ldots\ldots\ldots\ldots\ldots\ldots\ldots\ldots\ldots\ldots\ldots\ldots
\]
\[
\left\langle\tfrac{k}{2}\tfrac{k}{2};s,0\right|,\;\left\langle\tfrac{k}{2}\tfrac{k-2}{2};s,0\right|,\;\ldots,\;
\left\langle\tfrac{k}{2},\tfrac{1}{2};s,0\right|,\;\left\langle\tfrac{k}{2},-\tfrac{1}{2};s,0\right|,\;\ldots,\;
\left\langle\tfrac{k}{2},-\tfrac{k}{2};s,0\right|;
\]
\[
\left\langle\tfrac{k}{2}\tfrac{k}{2};s,-1\right|,\;\left\langle\tfrac{k}{2}\tfrac{k-2}{2};s,-1\right|,\;\ldots,\;
\left\langle\tfrac{k}{2},\tfrac{1}{2};s,-1\right|,\;\left\langle\tfrac{k}{2},-\tfrac{1}{2};s,-1\right|,\;\ldots,\;
\left\langle\tfrac{k}{2},-\tfrac{k}{2};s,-1\right|;
\]
\[
\ldots\ldots\ldots\ldots\ldots\ldots\ldots\ldots\ldots\ldots\ldots\ldots\ldots\ldots\ldots\ldots
\]
\begin{equation}\label{Bra_basis}
\left\langle\tfrac{k}{2}\tfrac{k}{2};s,-s\right|,\;\left\langle\tfrac{k}{2}\tfrac{k-2}{2};s,-s\right|,\;\ldots,\;
\left\langle\tfrac{k}{2},\tfrac{1}{2};s,-s\right|,\;\left\langle\tfrac{k}{2},-\tfrac{1}{2};s,-s\right|,\;\ldots,\;
\left\langle\tfrac{k}{2},-\tfrac{k}{2};s,-s\right|.
\end{equation}
The matrix $\Lambda^{s\frac{k}{2}}_3$ in the ket-basis (\ref{Ket_basis}) has the form
\[
\Lambda^{s\frac{k}{2}}_3=\text{{\rm diag}}\left({}_1\Lambda^{s\frac{k}{2}}_3,\;{}_2\Lambda^{s\frac{k}{2}}_3,\;
{}_3\Lambda^{s\frac{k}{2}}_3,\;\ldots,\;{}_{\frac{k+1}{2}}\Lambda^{s\frac{k}{2}}_3,\;
-{}_{\frac{k+1}{2}}\Lambda^{s\frac{k}{2}}_3,\;\ldots,\;-{}_3\Lambda^{s\frac{k}{2}}_3,-{}_2\Lambda^{s\frac{k}{2}}_3,\;
-{}_1\Lambda^{s\frac{k}{2}}_3\right),
\]
where
\[
{}_1\Lambda^{s\frac{k}{2}}_3=\text{{\rm diag}}\left(\frac{sk}{2},\frac{(s-1)k}{2},\frac{(s-2)k}{2},\ldots,\frac{k}{2},0,-\frac{k}{2},\ldots,\frac{(-s+2)k}{2},
\frac{(-s+1)k}{2},-\frac{sk}{2}\right),
\]
\begin{multline}
{}_2\Lambda^{s\frac{k}{2}}_3=\text{{\rm diag}}\left(\frac{s(k-2)}{2},\frac{(s-1)(k-2)}{2},\frac{(s-2)(k-2)}{2},\ldots,\frac{k-2}{2},0,\right.\\
\left.-\frac{k-2}{2},\ldots,
\frac{(-s+2)(k-2)}{2},\frac{(-s+1)(k-2)}{2},-\frac{s(k-2)}{2}\right),\nonumber
\end{multline}
\begin{multline}
{}_3\Lambda^{s\frac{k}{2}}_3=\text{{\rm diag}}\left(\frac{s(k-4)}{2},\frac{(s-1)(k-4)}{2},\frac{(s-2)(k-4)}{2},\ldots,\frac{k-4}{2},0,\right.\\
\left.-\frac{k-4}{2},\ldots,
\frac{(-s+2)(k-4)}{2},\frac{(-s+1)(k-4)}{2},-\frac{s(k-4)}{2}\right),\nonumber
\end{multline}
\[
\ldots\ldots\ldots\ldots\ldots\ldots\ldots\ldots\ldots\ldots\ldots\ldots\ldots\ldots\ldots\ldots
\]
\[
{}_{\frac{k+1}{2}}\Lambda^{s\frac{k}{2}}_3=\text{{\rm diag}}\left(\frac{s}{2},\frac{s-1}{2},\frac{s-2}{2},\ldots,\frac{1}{2},0,-\frac{1}{2},\ldots,\frac{-s+2}{2},
\frac{-s+1}{2},-\frac{s}{2}\right).
\]
An explicit form of $\Lambda^{\frac{k}{2}s}_3$ in the bra-basis (\ref{Bra_basis}) is
\[
\Lambda^{\frac{k}{2}s}_3=\text{{\rm diag}}\left({}_1\Lambda^{\frac{k}{2}s}_3,\;{}_2\Lambda^{\frac{k}{2}s}_3,\;
{}_3\Lambda^{\frac{k}{2}s}_3,\;\ldots,\;{}_{s}\Lambda^{\frac{k}{2}s}_3,\;\bO_{k+1},
-{}_{s}\Lambda^{\frac{k}{2}s}_3,\;\ldots,\;-{}_3\Lambda^{\frac{k}{2}s}_3,-{}_2\Lambda^{\frac{k}{2}s}_3,\;
-{}_1\Lambda^{\frac{k}{2}s}_3\right),
\]
where
\[
{}_1\Lambda^{\frac{k}{2}s}_3=\text{{\rm diag}}\left(\frac{ks}{2},\frac{(k-2)s}{2},\frac{(k-4)s}{2},\ldots,\frac{s}{2},-\frac{s}{2},\ldots,\frac{(-k+4)s}{2},
\frac{(-k+2)s}{2},-\frac{ks}{2}\right),
\]
\begin{multline}
{}_2\Lambda^{\frac{k}{2}s}_3=\text{{\rm diag}}\left(\frac{k(s-1)}{2},\frac{(k-2)(s-1)}{2},\frac{(k-4)(s-1)}{2},\ldots,\frac{s-1}{2},\right.\\
\left.-\frac{s-1}{2},\ldots,\frac{(-k+4)(s-1)}{2},
\frac{(-k+2)(s-1)}{2},-\frac{k(s-1)}{2}\right),\nonumber
\end{multline}
\begin{multline}
{}_3\Lambda^{\frac{k}{2}s}_3=\text{{\rm diag}}\left(\frac{k(s-2)}{2},\frac{(k-2)(s-2)}{2},\frac{(k-4)(s-2)}{2},\ldots,\frac{s-2}{2},\right.\\
\left.-\frac{s-2}{2},\ldots,\frac{(-k+4)(s-2)}{2},
\frac{(-k+2)(s-2)}{2},-\frac{k(s-2)}{2}\right),\nonumber
\end{multline}
\[
\ldots\ldots\ldots\ldots\ldots\ldots\ldots\ldots\ldots\ldots\ldots\ldots\ldots\ldots\ldots\ldots
\]
\[
{}_s\Lambda^{\frac{k}{2}s}_3=\text{{\rm diag}}\left(\frac{k}{2},\frac{k-2}{2},\frac{k-4}{2},\ldots,\frac{1}{2},-\frac{1}{2},\ldots,\frac{-k+4}{2},
\frac{-k+2}{2},-\frac{k}{2}\right).
\]
Here $\bO_{k+1}$ is $(k+1)\times(k+1)$ zero matrix.
Characteristic polynomial of $\Lambda^{s\frac{k}{2}}_3$ is
\begin{multline}
\boldsymbol{\Delta}(\Lambda^{s\frac{k}{2}}_3)=(\lambda-sk/2)^{m_1}(\lambda-(1-s)k/2)^{m_2}(\lambda-(2-s)k/2)^{m_3}
\ldots(\lambda-1/2)^{m_p}\lambda^{2s}\times\\
\times(\lambda+sk/2)^{m_{-1}}(\lambda+(1-s)k/2)^{m_{-2}}(\lambda+(2-s)k/2)^{m_{-3}}\ldots(\lambda+1/2)^{m_{-p}},
\nonumber
\end{multline}
where
\[
2\leq m_1\leq m^1_{\text{{\rm max}}},
\]
\[
2\leq m_2\leq m^2_{\text{{\rm max}}},
\]
\[
\ldots\ldots\ldots\ldots\ldots
\]
\[
2\leq m_{-p}\leq m^{-p}_{\text{{\rm max}}},
\]
here $m^1_{\text{max}}$ is a number of products among $(2s+1)(k+1)$ products of the $sk$-basis, which equal to eigenvalue $sk/2$; $m^2_{\text{max}}$ is a number of products among $(2s+1)(k+1)$ products of the $sk$-basis, which equal to eigenvalue $(1-s)k/2$ and so on.

Jordan form of $\Lambda^{s\frac{k}{2}}_3$ is
\[
\bJ(\Lambda^{s\frac{k}{2}}_3)=\text{{\rm diag}}\left(\bJ({}_{s\frac{k}{2}}\Lambda^{s\frac{k}{2}}_3), \bJ({}_{\frac{(1-s)k}{2}}\Lambda^{s\frac{k}{2}}_3),\ldots,\bJ({}_{\frac{1}{2}}\Lambda^{s\frac{k}{2}}_3),
\bJ({}_{0}\Lambda^{s\frac{k}{2}}_3),\bJ({}_{-\frac{1}{2}}\Lambda^{s\frac{k}{2}}_3),\ldots,
\bJ({}_{-\frac{sk}{2}}\Lambda^{s\frac{k}{2}}_3)\right),
\]
where
\[
\bJ({}_{\frac{sk}{2}}\Lambda^{s\frac{k}{2}}_3)=\begin{Vmatrix}
\frac{sk}{2} & 1 & 0 & \dots & 0\\
0 & \frac{sk}{2} & 1 & \dots & 0\\
.&&\hdotsfor[3]{3}\\
.&&&\hdotsfor[3]{2}\\
.&&&.&1\\
0 & 0 & 0 & \dots & \frac{sk}{2}
\end{Vmatrix}
\]
is a Jordan cell corresponding to elementary divisor $(\lambda-sk/2)^{m_1}$ and so on.

Analogously, $\Lambda^{\frac{k}{2}s}_3$ has the same structure:
\begin{multline}
\boldsymbol{\Delta}(\Lambda^{\frac{k}{2}s}_3)=(\lambda-sk/2)^{m_1}(\lambda-(1-s)k/2)^{m_2}(\lambda-(2-s)k/2)^{m_3}
\ldots(\lambda-1/2)^{m_p}\lambda^{k+1}\times\\
\times(\lambda+sk/2)^{m_{-1}}(\lambda+(1-s)k/2)^{m_{-2}}(\lambda+(2-s)k/2)^{m_{-3}}\ldots(\lambda+1/2)^{m_{-p}},
\nonumber
\end{multline}
\[
2\leq m_1\leq m^1_{\text{{\rm max}}},
\]
\[
2\leq m_2\leq m^2_{\text{{\rm max}}},
\]
\[
\ldots\ldots\ldots\ldots\ldots
\]
\[
2\leq m_{-p}\leq m^{-p}_{\text{{\rm max}}},
\]
\[
\bJ(\Lambda^{\frac{k}{2}s}_3)=\text{{\rm diag}}\left(\bJ({}_{\frac{sk}{2}}\Lambda^{\frac{k}{2}s}_3), \bJ({}_{\frac{(1-s)k}{2}}\Lambda^{\frac{k}{2}s}_3),\ldots,\bJ({}_{\frac{1}{2}}\Lambda^{\frac{k}{2}s}_3),
\bJ({}_{0}\Lambda^{\frac{k}{2}s}_3),\bJ({}_{-\frac{1}{2}}\Lambda^{\frac{k}{2}s}_3),\ldots,
\bJ({}_{-\frac{sk}{2}}\Lambda^{\frac{k}{2}s}_3)\right),
\]
For the spin chain of type $\left(\tfrac{k}{2},s\right)\longleftrightarrow\left(s,\tfrac{k}{2}\right)$ all the calculations are analogous. At this point, ket-basis (\ref{Ket_basis}) and bra-basis (\ref{Bra_basis}) are replaced for the states $\left(\tfrac{k}{2},s\right)$ and $\left(s,\tfrac{k}{2}\right)$.

In conclusion of the proof let us consider several examples of the spin chains formed by the spin-1/2 and dual spin-1/2 lines.

1) Spin chain $\left(1,\tfrac{1}{2}\right)\longleftrightarrow\left(\tfrac{1}{2},1\right)$. This chain is a part of the spin quadruplet (\ref{3Chain}). The state $\left(1,\tfrac{1}{2}\right)$ presents a 6-dimensional representation in the space $\Sym_{(2,1)}$. Spinor structure for this state is defined by the product $\C_2\otimes\C_2\bigotimes\overset{\ast}{\C}_2$ with the spinspace $\dS_2\otimes\dS_2\bigotimes\dot{\dS}_2\simeq\dS_{2^3}$. Ket-vectors of the helicity basis are
\[
\left| 11;\tfrac{1}{2}\tfrac{1}{2}\right\rangle,\;\left| 10;\tfrac{1}{2}\tfrac{1}{2}\right\rangle,\;\left| 1,-1;\tfrac{1}{2}\tfrac{1}{2}\right\rangle,\;\left| 11;\tfrac{1}{2},-\tfrac{1}{2}\right\rangle,\;\left| 10;\tfrac{1}{2},-\tfrac{1}{2}\right\rangle,\;\left| 1,-1;\tfrac{1}{2},-\tfrac{1}{2}\right\rangle.
\]
Using this ket-basis, we obtain an explicit form of the matrix $\Lambda^{1\frac{1}{2}}_3$ from (\ref{L3T}) at $l=1$ and $\dot{l}=1/2$:
\[
\Lambda^{1\frac{1}{2}}_3=\frac{1}{2}c^{k^\prime
k;\dot{k}^\prime\dot{k}}_{11;\frac{1}{2}\frac{1}{2}}\begin{array}{||cccccc||}
1 & 0 & 0 & 0 & 0 & 0\\
0 & 0 & 0 & 0 & 0 & 0\\
0 & 0 & -1& 0 & 0 & 0\\
0 & 0 & 0 & -1& 0 & 0\\
0 & 0 & 0 & 0 & 0 & 0\\
0 & 0 & 0 & 0 & 0 & 1
\end{array}.
\]
Characteristic polynomial for $\Lambda^{1\frac{1}{2}}_3$ (at $c^{k^\prime
k;\dot{k}^\prime\dot{k}}_{11;\frac{1}{2}\frac{1}{2}}=1$) is
\[
\boldsymbol{\Delta}(\Lambda^{1\frac{1}{2}}_3)=(\lambda-1/2)^2\lambda^2(\lambda+1/2)^2
\]
and for the Jordan form we have
\[
\bJ(\Lambda^{1\frac{1}{2}}_3)=\text{{\rm diag}}\left(\bJ({}_{\frac{1}{2}}\Lambda^{1\frac{1}{2}}_3),
\bJ({}_0\Lambda^{1\frac{1}{2}}_3),\bJ({}_{-\frac{1}{2}}\Lambda^{1\frac{1}{2}}_3)\right),
\]
where
\[
\bJ({}_{\frac{1}{2}}\Lambda^{1\frac{1}{2}}_3)=\begin{array}{||cc||}
\frac{1}{2} & 1\\
0 & \frac{1}{2}
\end{array},\quad
\bJ({}_0\Lambda^{1\frac{1}{2}}_3)=\begin{array}{||cc||}
0 & 1\\
0 & 0
\end{array},\quad
\bJ({}_{-\frac{1}{2}}\Lambda^{1\frac{1}{2}}_3)=\begin{array}{||cc||}
-\frac{1}{2} & 1\\
0 & -\frac{1}{2}
\end{array}.
\]
System (\ref{BS2}) for the spin quadruplet (\ref{3Chain}) takes the following form:
\[
\sum^3_{j=1}\Lambda^{\frac{3}{2},0}_j\frac{\partial\psi}{\partial a_j}-
i\sum^3_{j=1}\Lambda^{\frac{3}{2},0}_j\frac{\partial\psi}{\partial a^\ast_j}+
m^{(3/2)}\dot{\psi}=0,
\]
\[
\sum^3_{j=1}\Lambda^{1,\frac{1}{2}}_j\frac{\partial\psi}{\partial a_j}-
i\sum^3_{j=1}\Lambda^{1,\frac{1}{2}}_j\frac{\partial\psi}{\partial a^\ast_j}+
m^{(1/2)}\dot{\psi}=0,
\]
\[
\sum^3_{j=1}\overset{\ast}{\Lambda}{}^{\frac{1}{2},1}_j\frac{\partial\dot{\psi}}
{\partial\widetilde{a}_j}+i\sum^3_{j=1}\overset{\ast}{\Lambda}{}^{\frac{1}{2},1}_j
\frac{\partial\dot{\psi}}{\partial\widetilde{a}^\ast_j}+
m^{(1/2)}\psi=0,
\]
\[
\sum^3_{j=1}\overset{\ast}{\Lambda}{}^{0,\frac{3}{2}}_j\frac{\partial\dot{\psi}}
{\partial\widetilde{a}_j}+i\sum^3_{j=1}\overset{\ast}{\Lambda}{}^{0,\frac{3}{2}}_j
\frac{\partial\dot{\psi}}{\partial\widetilde{a}^\ast_j}+
m^{(3/2)}\psi=0.
\]

2) Spin chain $\left(\tfrac{3}{2},1\right)\longleftrightarrow\left(1,\tfrac{3}{2}\right)$. This chain is a part of the spin 6-plet (\ref{6Chain}). The state $\left(\tfrac{3}{2},1\right)$ presents a 12-dimensional representation in the space $\Sym_{(3,2)}$. Spinor structure for this state is defined by the tensor product $\C_2\otimes\C_2\otimes\C_2\bigotimes\overset{\ast}{\C}_2\otimes\overset{\ast}{\C}_2$ with the spinspace $\dS_2\otimes\dS_2\otimes\dS_2\bigotimes\dot{\dS}_2\otimes\dot{\dS}_2\simeq\dS_{2^5}$. In this case, ket-vectors of the helicity basis are
\[
\left|\tfrac{3}{2}\tfrac{3}{2};11\right\rangle,\;\left|\tfrac{3}{2}\tfrac{1}{2};11\right\rangle,\;\left|
\tfrac{3}{2},-\tfrac{1}{2};11\right\rangle,\;\left|\tfrac{3}{2},-\tfrac{3}{2};11\right\rangle,\;\left| \tfrac{3}{2}\tfrac{3}{2};10\right\rangle,\;\left|\tfrac{3}{2}\tfrac{1}{2};10\right\rangle.
\]
\[
\left|\tfrac{3}{2},-\tfrac{1}{2};10\right\rangle,\;\left|\tfrac{3}{2},-\tfrac{3}{2};10\right\rangle,\;\left|
\tfrac{3}{2}\tfrac{3}{2};1,-1\right\rangle,\;\left|\tfrac{3}{2}\tfrac{1}{2};1,-1\right\rangle,\;\left| \tfrac{3}{2},-\tfrac{1}{2};1,-1\right\rangle,\;\left|\tfrac{3}{2},-\tfrac{3}{2};1,-1\right\rangle.
\]
Using this ket-basis and (\ref{L3T}) at $l=3/2$ and $\dot{l}=1$, we obtain
\[
\Lambda^{\frac{3}{2}1}_3=\frac{1}{2}c^{k^\prime
k;\dot{k}^\prime\dot{k}}_{\frac{3}{2}\frac{3}{2};11}\begin{array}{||ccc||}
{}_1\Lambda^{\frac{3}{2}1}_3 & \bO_4 & \bO_4\\
\bO_4 & \bO_4 & \bO_4\\
\bO_4 & \bO_4 &-{}_1\Lambda^{\frac{3}{2}1}_3
\end{array},
\]
where
\[
{}_1\Lambda^{\frac{3}{2}1}_3=\begin{array}{||cccc||}
3 & 0 & 0 & 0\\
0 & 1 & 0 & 0\\
0 & 0 &-1 & 0\\
0 & 0 & 0 &-3
\end{array}
\]
and $\bO_4$ is $4\times 4$ zero matrix. Characteristic polynomial is
\[
\boldsymbol{\Delta}(\Lambda^{\frac{3}{2}1}_3)=(\lambda-3/2)^2(\lambda-1/2)^2\lambda^4(\lambda+1/2)^2(\lambda+3/2)^2.
\]
and for the Jordan form we obtain
\[
\bJ(\Lambda^{\frac{3}{2}1}_3)=\text{{\rm diag}}\left(\bJ({}_{\frac{3}{2}}\Lambda^{\frac{3}{2}1}_3),\bJ({}_{\frac{1}{2}}\Lambda^{\frac{3}{2}1}_3)
\bJ({}_0\Lambda^{\frac{3}{2}1}_3),\bJ({}_{-\frac{1}{2}}\Lambda^{\frac{3}{2}1}_3),
\bJ({}_{-\frac{3}{2}}\Lambda^{\frac{3}{2}1}_3),\right),
\]
where
\[
\bJ({}_{\frac{3}{2}}\Lambda^{\frac{3}{2}1}_3)=\begin{array}{||cc||}
\frac{3}{2} & 1\\
0 & \frac{3}{2}
\end{array},\quad
\bJ({}_{\frac{1}{2}}\Lambda^{\frac{3}{2}1}_3)=\begin{array}{||cc||}
\frac{1}{2} & 1\\
0 & \frac{1}{2}
\end{array},\quad
\bJ({}_0\Lambda^{\frac{3}{2}1}_3)=\begin{array}{||cccc||}
0 & 1 & 0 & 0\\
0 & 0 & 1 & 0\\
0 & 0 & 0 & 1\\
0 & 0 & 0 & 0
\end{array}.
\]
In its turn, system (\ref{BS2}) in the case of 6-plet (\ref{6Chain}) takes the following form:
\[
\sum^3_{j=1}\Lambda^{\frac{5}{2},0}_j\frac{\partial\psi}{\partial a_j}-
i\sum^3_{j=1}\Lambda^{\frac{5}{2},0}_j\frac{\partial\psi}{\partial a^\ast_j}+
m^{(5/2)}\dot{\psi}=0,
\]
\[
\sum^3_{j=1}\Lambda^{2,\frac{1}{2}}_j\frac{\partial\psi}{\partial a_j}-
i\sum^3_{j=1}\Lambda^{2,\frac{1}{2}}_j\frac{\partial\psi}{\partial a^\ast_j}+
m^{(3/2)}\dot{\psi}=0,
\]
\[
\sum^3_{j=1}\Lambda^{\frac{3}{2},1}_j\frac{\partial\psi}{\partial a_j}-
i\sum^3_{j=1}\Lambda^{\frac{3}{2},1}_j\frac{\partial\psi}{\partial a^\ast_j}+
m^{(1/2)}\dot{\psi}=0,
\]
\[
\sum^3_{j=1}\overset{\ast}{\Lambda}{}^{1,\frac{3}{2}}_j\frac{\partial\dot{\psi}}
{\partial\widetilde{a}_j}+i\sum^3_{j=1}\overset{\ast}{\Lambda}{}^{1,\frac{3}{2}}_j
\frac{\partial\dot{\psi}}{\partial\widetilde{a}^\ast_j}+
m^{(1/2)}\psi=0,
\]
\[
\sum^3_{j=1}\overset{\ast}{\Lambda}{}^{\frac{1}{2},2}_j\frac{\partial\dot{\psi}}
{\partial\widetilde{a}_j}+i\sum^3_{j=1}\overset{\ast}{\Lambda}{}^{\frac{1}{2},2}_j
\frac{\partial\dot{\psi}}{\partial\widetilde{a}^\ast_j}+
m^{(3/2)}\psi=0,
\]
\[
\sum^3_{j=1}\overset{\ast}{\Lambda}{}^{0,\frac{5}{2}}_j\frac{\partial\dot{\psi}}
{\partial\widetilde{a}_j}+i\sum^3_{j=1}\overset{\ast}{\Lambda}{}^{0,\frac{5}{2}}_j
\frac{\partial\dot{\psi}}{\partial\widetilde{a}^\ast_j}+
m^{(5/2)}\psi=0.
\]
3) Spin chain $\left(2,\tfrac{3}{2}\right)\longleftrightarrow\left(\tfrac{3}{2},2\right)$. The state $\left(2,\tfrac{3}{2}\right)$ presents a 20-dimensional representation of the group $\SL(2,\C)$ in the space $\Sym_{(4,3)}$. This representation finishes the first cell of spinorial chessboard (see Fig.\,1). Spinor structure for this state is defined by the tensor product $\C_2\otimes\C_2\otimes\C_2\otimes\C_2\bigotimes\overset{\ast}{\C}_2\otimes\overset{\ast}{\C}_2\otimes\overset{\ast}{\C}_2$ with the spinspace $\dS_2\otimes\dS_2\otimes\dS_2\otimes\dS_2\bigotimes\dot{\dS}_2\otimes\dot{\dS}_2\otimes\dot{\dS}_2\simeq\dS_{2^7}$. In this case, ket-vectors of the helicity basis are
\[
\left|2,2;\tfrac{3}{2},\tfrac{3}{2}\right\rangle,\;\left|2,1;\tfrac{3}{2},\tfrac{3}{2}\right\rangle,\;
\left|2,0;\tfrac{3}{2},\tfrac{3}{2}\right\rangle,\;\left|2,-1;\tfrac{3}{2},\tfrac{3}{2}\right\rangle,\;
\left|2,-2;\tfrac{3}{2},\tfrac{3}{2}\right\rangle,
\]
\[
\left|2,2;\tfrac{3}{2},\tfrac{1}{2}\right\rangle,\;\left|2,1;\tfrac{3}{2},\tfrac{1}{2}\right\rangle,\;
\left|2,0;\tfrac{3}{2},\tfrac{1}{2}\right\rangle,\;\left|2,-1;\tfrac{3}{2},\tfrac{1}{2}\right\rangle,\;
\left|2,-2;\tfrac{3}{2},\tfrac{1}{2}\right\rangle,
\]
\[
\left|2,2;\tfrac{3}{2},-\tfrac{1}{2}\right\rangle,\;\left|2,1;\tfrac{3}{2},-\tfrac{1}{2}\right\rangle,\;
\left|2,0;\tfrac{3}{2},-\tfrac{1}{2}\right\rangle,\;\left|2,-1;\tfrac{3}{2},-\tfrac{1}{2}\right\rangle,\;
\left|2,-2;\tfrac{3}{2},-\tfrac{1}{2}\right\rangle,
\]
\[
\left|2,2;\tfrac{3}{2},-\tfrac{3}{2}\right\rangle,\;\left|2,1;\tfrac{3}{2},-\tfrac{3}{2}\right\rangle,\;
\left|2,0;\tfrac{3}{2},-\tfrac{3}{2}\right\rangle,\;\left|2,-1;\tfrac{3}{2},-\tfrac{3}{2}\right\rangle,\;
\left|2,-2;\tfrac{3}{2},-\tfrac{3}{2}\right\rangle.
\]
Using this ket-basis and (\ref{L3T}) at $l=2$ and $\dot{l}=3/2$, we obtain
\[
\Lambda^{2\frac{3}{2}}_3=c^{k^\prime
k;\dot{k}^\prime\dot{k}}_{22;\frac{3}{2}\frac{3}{2}}\begin{array}{||cccc||}
{}_1\Lambda^{2\frac{3}{2}}_3 & \bO_5 & \bO_5 & \bO_5\\
\bO_5 & {}_2\Lambda^{2\frac{3}{2}}_3 & \bO_5 & \bO_5\\
\bO_5 & \bO_5 & -{}_2\Lambda^{2\frac{3}{2}}_3 & \bO_5\\
\bO_5 & \bO_5 & \bO_5 & -{}_1\Lambda^{2\frac{3}{2}}_3
\end{array},
\]
where
\[
{}_1\Lambda^{2\frac{3}{2}}_3=\begin{array}{||ccccc||} 3 & 0 & 0 & 0
& 0\\
0 & \frac{3}{2} & 0 & 0 & 0\\
0 & 0 & 0 & 0 & 0\\
0 & 0 & 0 & -\frac{3}{2} & 0\\
0 & 0 & 0 & 0 & -3
\end{array},\quad
{}_2\Lambda^{2\frac{3}{2}}_3=\begin{array}{||ccccc||} 1 & 0 & 0 & 0
& 0\\
0 & \frac{1}{2} & 0 & 0 & 0\\
0 & 0 & 0 & 0 & 0\\
0 & 0 & 0 & -\frac{1}{2} & 0\\
0 & 0 & 0 & 0 & -1
\end{array}.
\]
Characteristic polynomial at $c^{k^\prime
k;\dot{k}^\prime\dot{k}}_{22;\frac{3}{2}\frac{3}{2}}=1$ is
\[
\boldsymbol{\Delta}(\Lambda^{2\frac{3}{2}}_3)=\lambda^4(\lambda-3)^2(\lambda-3/2)^2(\lambda-1)^2(\lambda-1/2)^2
(\lambda+1/2)^2(\lambda+1)^2(\lambda+3/2)^2(\lambda+3)^2.
\]
4) Spin chain $\left(\tfrac{7}{2},3\right)\longleftrightarrow\left(3,\tfrac{7}{2}\right)$. This chain belongs to a second cell of the spinorial chessboard (see Fig.\,2). The state $\left(\tfrac{7}{2},3\right)$ is described within representation $\boldsymbol{\tau}_{7/2,3}$ of the degree 56, that corresponds to the energy level $\sH_E\simeq\Sym_{(8,7)}$. Spintensor substrate for this level is
\[
\underbrace{\C_2\otimes\C_2\otimes\cdots\otimes\C_2}_{8\;\text{times}}\bigotimes
\underbrace{\overset{\ast}{\C}_2\otimes\overset{\ast}{\C}_2\otimes\cdots\otimes
\overset{\ast}{\C}_2}_{7\;\text{times}}
\]
with the spinspace
\[
\underbrace{\dS_2\otimes\dS_2\otimes\cdots\otimes\dS_2}_{8\;\text{times}}\bigotimes
\underbrace{\dot{\dS}_2\otimes\dot{\dS}_2\otimes\cdots\otimes\dot{\dS}_2}_{7\;\text{times}}
\simeq\dS_{2^{15}}.
\]
Helicity basis in $\sH_E\simeq\Sym_{(8,7)}$ consists of the following 56 ket-vectors:
\[
\left|\tfrac{7}{2},\tfrac{7}{2};3,3\right\rangle,\;\left|\tfrac{7}{2},\tfrac{5}{2};3,3\right\rangle,\;
\ldots,\;\left|\tfrac{7}{2},\tfrac{1}{2};3,3\right\rangle,\;
\left|\tfrac{7}{2},-\tfrac{1}{2};3,3\right\rangle,\;\ldots,\;
\left|\tfrac{7}{2},-\tfrac{7}{2};3,3\right\rangle;
\]
\[
\left|\tfrac{7}{2},\tfrac{7}{2};3,2\right\rangle,\;\left|\tfrac{7}{2},\tfrac{5}{2};3,2\right\rangle,\;
\ldots,\;\left|\tfrac{7}{2},\tfrac{1}{2};3,2\right\rangle,\;
\left|\tfrac{7}{2},-\tfrac{1}{2};3,2\right\rangle,\;\ldots,\;
\left|\tfrac{7}{2},-\tfrac{7}{2};3,2\right\rangle;
\]
\[
\ldots\ldots\ldots\ldots\ldots\ldots\ldots\ldots\ldots\ldots\ldots\ldots\ldots\ldots\ldots\ldots
\]
\[
\left|\tfrac{7}{2},\tfrac{7}{2};3,-3\right\rangle,\;\left|\tfrac{7}{2},\tfrac{5}{2};3,-3\right\rangle,\;
\ldots,\;\left|\tfrac{7}{2},\tfrac{1}{2};3,-3\right\rangle,\;
\left|\tfrac{7}{2},-\tfrac{1}{2};3,-3\right\rangle,\;\ldots,\;
\left|\tfrac{7}{2},-\tfrac{7}{2};3,-3\right\rangle.
\]
Using this ket-basis, we obtain from (\ref{L3T}) at $l=7/2$ and $\dot{l}=3$ the following matrix:
\[
\Lambda^{\frac{7}{2}3}_3=c^{k^\prime
k;\dot{k}^\prime\dot{k}}_{\frac{7}{2}\frac{7}{2};33}\begin{array}{||ccccccc||}
{}_1\Lambda^{\frac{7}{2}3}_3 & \bO_8 & \bO_8 & \bO_8 & \bO_8 & \bO_8 & \bO_8\\
\bO_8 & {}_2\Lambda^{\frac{7}{2}3}_3 & \bO_8 & \bO_8 & \bO_8 & \bO_8 & \bO_8\\
\bO_8 & \bO_8 & {}_3\Lambda^{\frac{7}{2}3}_3 & \bO_8 & \bO_8 & \bO_8 & \bO_8\\
\bO_8 & \bO_8 & \bO_8 & \bO_8 & \bO_8 & \bO_8 & \bO_8\\
\bO_8 & \bO_8 & \bO_8 & \bO_8 & -{}_3\Lambda^{\frac{7}{2}3}_3 & \bO_8 & \bO_8\\
\bO_8 & \bO_8 & \bO_8 & \bO_8 & \bO_8 &
-{}_2\Lambda^{\frac{7}{2}3}_3 &
\bO_8\\
\bO_8 & \bO_8 & \bO_8 & \bO_8 & \bO_8 & \bO_8 &
-{}_1\Lambda^{\frac{7}{2}3}_3
\end{array},
\]
where
\[
{}_1\Lambda^{\frac{7}{2}3}_3=\text{{\rm diag}}\left(\frac{21}{2}, \frac{15}{2}, \frac{9}{2}, \frac{3}{2}, -\frac{3}{2}, -\frac{9}{2}, -\frac{15}{2}, -\frac{21}{2}\right),
\]
\[
{}_2\Lambda^{\frac{7}{2}3}_3=\text{{\rm diag}}\left(7, 5, 3, 1, -1, -3, -5, -7\right),
\]
\[
{}_3\Lambda^{\frac{7}{2}3}_3=\text{{\rm diag}}\left(\frac{7}{2}, \frac{5}{2}, \frac{3}{2}, \frac{1}{2}, -\frac{1}{2}, -\frac{3}{2}, -\frac{5}{2}, -\frac{7}{2}\right).
\]
Characteristic polynomial for $\Lambda^{\frac{7}{2}3}_3$ (at $c^{k^\prime
k;\dot{k}^\prime\dot{k}}_{\frac{7}{2}\frac{7}{2};33}=1$)  is
\[
\boldsymbol{\Delta}(\Lambda^{\frac{7}{2}3}_3)=\lambda^8(\lambda-3/2)^4(\lambda-21/2)^2\cdots(\lambda-1/2)^2
(\lambda+3/2)^4(\lambda+21/2)^2\cdots(\lambda+1/2)^2.
\]
\end{proof}
\section{Stability levels of the matter spectrum}
It is well known that a first stable energy level of the matter spectrum is an electron state with the mass $m_e$. The second stability level of the matter spectrum is defined by a proton state with the mass $m_p$. All the particles, occurring between these two levels (and beyond), are unstable formations. The first stability level is defined by the fundamental doublet (\ref{EChain}) (see also Fig.\,1). With the aim to find place (on the representation scheme) of the second stability level we will use the mass formula (\ref{MGY}). As is known, a mass ratio of these two levels is given by Houston estimation \cite{Hou27}:
\[
\frac{m_p}{m_e}=1836,57\pm 0,20.
\]
From the formula (\ref{MGY}) it follows that a degree of the representation $\boldsymbol{\tau}_{l\dot{l}}$ is approximately equal to 900, that corresponds to representations $\boldsymbol{\tau}_{15,\frac{29}{2}}$ and $\boldsymbol{\tau}_{\frac{29}{2},15}$ with the levels $\sH_E\simeq\Sym_{(31,30)}$ and $\sH_E\simeq\Sym_{(30,31)}$ of the energy operator $H$. On the Fig.\,2 these levels are placed within circles. In relation with the place of the second stability level (proton state) we have the following remark. This place in energy spectrum is definitely separated out from other energy levels. Namely, this level occurs near the boundary of the fractal of second order\footnote{Fractal structure of the energy spectrum is generated by the action of the Brauer-Wall group $BW_{\R}\simeq\dZ_8$. The action of $BW_{\R}\simeq\dZ_8$ relates different types of real Clifford algebras, which form a real substructure of spintensor substrate. A cyclic action of $BW_{\R}\simeq\dZ_8$ generates a fractal structure on the substrate (\ref{TenAlg}). This fractal structure is analogous to a Sierpi\'{n}ski carpet with the fractal (Besicovitch-Hausdorff) dimension $D=\ln 63/\ln 8\approx 1,9924$ (see \cite{Var15a}). Further, in virtue of the mapping $\cl_{p,q}\overset{\gamma}{\longrightarrow}\End_{\K}(\dS),\quad
u\longrightarrow\gamma(u),\quad
\gamma(u)\boldsymbol{s}=u\boldsymbol{s}$, where $\dS=\dS_{2^r}(\K)\simeq I_{p,q}=\cl_{p,q}f$ is a \textit{real spinspace},
$\boldsymbol{s}=\boldsymbol{s}^{\alpha_1\alpha_2\ldots\alpha_r}\in\dS_{2^r}$, fractal structure of the spintensor substrate is transferred on the representations of the group $\spin_+(1,3)\simeq\SL(2,\C)$.}. In its turn, the first stability level (electron state) occurs at the ground of the fractal of the first order (see Fig.\,2).
\begin{figure}[htbp]
\unitlength=1mm
\begin{center}
\begin{picture}(100,150)
\put(50,0){$\scr\overset{(0,0)}{\bullet}$}
\put(55,5){$\scr\overset{(\frac{1}{2},0)}{\bigcirc\hskip -5.9pt\bullet}$}
\put(45,5){$\scr\overset{(0,\frac{1}{2})}{\bigcirc\hskip -5.9pt\bullet}$}
\put(40,10){$\scr\overset{(0,1)}{\bullet}$}
\put(50,10){$\scr\overset{(\frac{1}{2},\frac{1}{2})}{\bullet}$}
\put(60,10){$\scr\overset{(1,0)}{\bullet}$}
\put(35,15){$\scr\overset{(0,\frac{3}{2})}{\bullet}$}
\put(45,15){$\scr\overset{(\frac{1}{2},1)}{\bullet}$}
\put(55,15){$\scr\overset{(1,\frac{1}{2})}{\bullet}$}
\put(65,15){$\scr\overset{(\frac{3}{2},0)}{\bullet}$}
\put(30,20){$\scr\overset{(0,2)}{\bullet}$}\put(10,20.5){\line(1,0){20}}\put(20,21.5){$\scr 1$}
\put(40,20){$\scr\overset{(\frac{1}{2},\frac{3}{2})}{\bullet}$}
\put(50,20){$\scr\overset{(1,1)}{\bullet}$}
\put(60,20){$\scr\overset{(\frac{3}{2},\frac{1}{2})}{\bullet}$}
\put(70,20){$\scr\overset{(2,0)}{\bullet}$}
\put(35,25){$\scr\overset{(\frac{1}{2},2)}{\bullet}$}
\put(45,25){$\scr\overset{(1,\frac{3}{2})}{\bullet}$}
\put(55,25){$\scr\overset{(\frac{3}{2},1)}{\bullet}$}
\put(65,25){$\scr\overset{(2,\frac{1}{2})}{\bullet}$}
\put(40,30){$\scr\overset{(1,2)}{\bullet}$}
\put(50,30){$\scr\overset{(\frac{3}{2},\frac{3}{2})}{\bullet}$}
\put(60,30){$\scr\overset{(2,1)}{\bullet}$}
\put(45,35){$\scr\overset{(\frac{3}{2},2)}{\bullet}$}
\put(55,35){$\scr\overset{(2,\frac{3}{2})}{\bullet}$}
\put(50,40){$\scr\overset{(2,2)}{\bullet}$}
\put(45,45){$\scr\overset{(2,\frac{5}{2})}{\bullet}$}
\put(55,45){$\scr\overset{(\frac{5}{2},2)}{\bullet}$}
\put(40,50){$\scr\overset{(2,3)}{\bullet}$}
\put(50,50){$\scr\overset{(\frac{5}{2},\frac{5}{2})}{\bullet}$}
\put(60,50){$\scr\overset{(3,2)}{\bullet}$}
\put(35,55){$\scr\overset{(2,\frac{7}{2})}{\bullet}$}
\put(45,55){$\scr\overset{(\frac{5}{2},3)}{\bullet}$}
\put(55,55){$\scr\overset{(3,\frac{5}{2})}{\bullet}$}
\put(65,55){$\scr\overset{(\frac{7}{2},2)}{\bullet}$}
\put(30,60){$\scr\overset{(2,4)}{\bullet}$}\put(10,60.5){\line(1,0){20}}\put(20,61.5){$\scr 2$}
\put(40,60){$\scr\overset{(\frac{5}{2},\frac{7}{2})}{\bullet}$}
\put(50,60){$\scr\overset{(3,3)}{\bullet}$}
\put(60,60){$\scr\overset{(\frac{7}{2},\frac{5}{2})}{\bullet}$}
\put(70,60){$\scr\overset{(4,2)}{\bullet}$}
\put(35,65){$\scr\overset{(\frac{5}{2},4)}{\bullet}$}
\put(45,65){$\scr\overset{(3,\frac{7}{2})}{\bullet}$}
\put(55,65){$\scr\overset{(\frac{7}{2},3)}{\bullet}$}
\put(65,65){$\scr\overset{(4,\frac{5}{2})}{\bullet}$}
\put(40,70){$\scr\overset{(3,4)}{\bullet}$}
\put(50,70){$\scr\overset{(\frac{7}{2},\frac{7}{2})}{\bullet}$}
\put(60,70){$\scr\overset{(4,3)}{\bullet}$}
\put(45,75){$\scr\overset{(\frac{7}{2},4)}{\bullet}$}
\put(55,75){$\scr\overset{(4,\frac{7}{2})}{\bullet}$}
\put(50,80){$\scr\overset{(4,4)}{\bullet}$}
\put(52,85){$\vdots$}
\put(52,90){$\vdots$}
\put(48.75,95){$\scr\overset{(14,14)}{\bullet}$}
\put(43.75,100){$\scr\overset{(14,\frac{29}{2})}{\bullet}$}
\put(53.75,100){$\scr\overset{(\frac{29}{2},14)}{\bullet}$}
\put(38.75,105){$\scr\overset{(14,15)}{\bullet}$}
\put(48.75,105){$\scr\overset{(\frac{29}{2},\frac{29}{2})}{\bullet}$}
\put(58.75,105){$\scr\overset{(15,14)}{\bullet}$}
\put(33.75,110){$\scr\overset{(14,\frac{31}{2})}{\bullet}$}
\put(43.75,110){$\scr\overset{(\frac{29}{2},15)}{\bigcirc\hskip -5.9pt\bullet}$}
\put(53.75,110){$\scr\overset{(15,\frac{29}{2})}{\bigcirc\hskip -5.9pt\bullet}$}
\put(63.75,110){$\scr\overset{(\frac{31}{2},14)}{\bullet}$}
\put(28.75,115){$\scr\overset{(14,16)}{\bullet}$}\put(10,115.5){\line(1,0){20}}\put(20,116.5){$\scr 8$}
\put(38.75,115){$\scr\overset{(\frac{29}{2},\frac{31}{2})}{\bullet}$}
\put(48.75,115){$\scr\overset{(15,15)}{\bullet}$}
\put(58.75,115){$\scr\overset{(\frac{31}{2},\frac{29}{2})}{\bullet}$}
\put(68.75,115){$\scr\overset{(16,14)}{\bullet}$}
\put(33.75,120){$\scr\overset{(\frac{29}{2},16)}{\bullet}$}
\put(43.75,120){$\scr\overset{(15,\frac{31}{2})}{\bullet}$}
\put(53.75,120){$\scr\overset{(\frac{31}{2},15)}{\bullet}$}
\put(63.75,120){$\scr\overset{(16,\frac{29}{2})}{\bullet}$}
\put(38.75,125){$\scr\overset{(15,16)}{\bullet}$}
\put(48.75,125){$\scr\overset{(\frac{31}{2},\frac{31}{2})}{\bullet}$}
\put(58.75,125){$\scr\overset{(16,15)}{\bullet}$}
\put(43.75,130){$\scr\overset{(\frac{31}{2},16)}{\bullet}$}
\put(53.75,130){$\scr\overset{(16,\frac{31}{2})}{\bullet}$}
\put(48.75,135){$\scr\overset{(16,16)}{\bullet}$}
\put(10,0.5){\line(1,0){42}}\put(50,0.5){\vector(1,0){42}}
\put(32,15){$\vdots$}
\put(32,11){$\vdots$}
\put(32,7){$\vdots$}
\put(32,3){$\vdots$}
\put(32,1){$\cdot$}
\put(29.5,-3){$\scr -2$}
\put(32,55){$\vdots$}
\put(32,51){$\vdots$}
\put(32,47){$\vdots$}
\put(32,43){$\vdots$}
\put(32,39){$\vdots$}
\put(32,35){$\vdots$}
\put(32,31){$\vdots$}
\put(32,27){$\vdots$}
\put(32,23){$\vdots$}
\put(32,110){$\vdots$}
\put(32,106){$\vdots$}
\put(32,102){$\vdots$}
\put(32,98){$\vdots$}
\put(32,94){$\vdots$}
\put(32,90){$\vdots$}
\put(32,86){$\vdots$}
\put(32,82){$\vdots$}
\put(32,78){$\vdots$}
\put(32,74){$\vdots$}
\put(32,70){$\vdots$}
\put(32,66){$\vdots$}
\put(32,120){$\vdots$}
\put(32,124){$\vdots$}
\put(32,128){$\vdots$}
\put(32,132){$\vdots$}
\put(32,136){$\vdots$}
\put(32,140){$\vdots$}
\put(72,15){$\vdots$}
\put(72,11){$\vdots$}
\put(72,7){$\vdots$}
\put(72,3){$\vdots$}
\put(72,1){$\cdot$}
\put(72,-3){$\scr 2$}
\put(72,55){$\vdots$}
\put(72,51){$\vdots$}
\put(72,47){$\vdots$}
\put(72,43){$\vdots$}
\put(72,39){$\vdots$}
\put(72,35){$\vdots$}
\put(72,31){$\vdots$}
\put(72,27){$\vdots$}
\put(72,23){$\vdots$}
\put(72,110){$\vdots$}
\put(72,106){$\vdots$}
\put(72,102){$\vdots$}
\put(72,98){$\vdots$}
\put(72,94){$\vdots$}
\put(72,90){$\vdots$}
\put(72,86){$\vdots$}
\put(72,82){$\vdots$}
\put(72,78){$\vdots$}
\put(72,74){$\vdots$}
\put(72,70){$\vdots$}
\put(72,66){$\vdots$}
\put(72,120){$\vdots$}
\put(72,124){$\vdots$}
\put(72,128){$\vdots$}
\put(72,132){$\vdots$}
\put(72,136){$\vdots$}
\put(72,140){$\vdots$}
\put(37.5,10){$\vdots$}
\put(37.5,6){$\vdots$}
\put(37.5,2){$\vdots$}
\put(37.5,1){$\cdot$}
\put(35,-3){$\scr -\frac{3}{2}$}
\put(37.5,50){$\vdots$}
\put(37.5,46){$\vdots$}
\put(37.5,42){$\vdots$}
\put(37.5,38){$\vdots$}
\put(37.5,34){$\vdots$}
\put(37.5,31){$\vdots$}
\put(37.5,105){$\vdots$}
\put(37.5,101){$\vdots$}
\put(37.5,97){$\vdots$}
\put(37.5,93){$\vdots$}
\put(37.5,89){$\vdots$}
\put(37.5,85){$\vdots$}
\put(37.5,81){$\vdots$}
\put(37.5,77){$\vdots$}
\put(37.5,73){$\vdots$}
\put(37.5,69){$\vdots$}
\put(37.5,124){$\vdots$}
\put(37.5,128){$\vdots$}
\put(37.5,132){$\vdots$}
\put(37.5,136){$\vdots$}
\put(37.5,140){$\vdots$}
\put(67.5,10){$\vdots$}
\put(67.5,6){$\vdots$}
\put(67.5,2){$\vdots$}
\put(67.5,1){$\cdot$}
\put(67,-3){$\scr\frac{3}{2}$}
\put(67.5,50){$\vdots$}
\put(67.5,46){$\vdots$}
\put(67.5,42){$\vdots$}
\put(67.5,38){$\vdots$}
\put(67.5,34){$\vdots$}
\put(67.5,31){$\vdots$}
\put(67.5,105){$\vdots$}
\put(67.5,101){$\vdots$}
\put(67.5,97){$\vdots$}
\put(67.5,93){$\vdots$}
\put(67.5,89){$\vdots$}
\put(67.5,85){$\vdots$}
\put(67.5,81){$\vdots$}
\put(67.5,77){$\vdots$}
\put(67.5,73){$\vdots$}
\put(67.5,69){$\vdots$}
\put(67.5,124){$\vdots$}
\put(67.5,128){$\vdots$}
\put(67.5,132){$\vdots$}
\put(67.5,136){$\vdots$}
\put(67.5,140){$\vdots$}
\put(42,5){$\vdots$}
\put(42,1){$\vdots$}
\put(40,-3){$\scr -1$}
\put(42,45){$\vdots$}
\put(42,41){$\vdots$}
\put(42,37){$\vdots$}
\put(42,33){$\vdots$}
\put(42,100){$\vdots$}
\put(42,96){$\vdots$}
\put(42,92){$\vdots$}
\put(42,88){$\vdots$}
\put(42,84){$\vdots$}
\put(42,81){$\vdots$}
\put(42,77){$\vdots$}
\put(42,73){$\vdots$}
\put(42,128){$\vdots$}
\put(42,132){$\vdots$}
\put(42,136){$\vdots$}
\put(42,140){$\vdots$}
\put(62,5){$\vdots$}
\put(62,1){$\vdots$}
\put(62,-3){$\scr 1$}
\put(62,45){$\vdots$}
\put(62,41){$\vdots$}
\put(62,37){$\vdots$}
\put(62,33){$\vdots$}
\put(62,100){$\vdots$}
\put(62,96){$\vdots$}
\put(62,92){$\vdots$}
\put(62,88){$\vdots$}
\put(62,84){$\vdots$}
\put(62,81){$\vdots$}
\put(62,77){$\vdots$}
\put(62,73){$\vdots$}
\put(62,128){$\vdots$}
\put(62,132){$\vdots$}
\put(62,136){$\vdots$}
\put(62,140){$\vdots$}
\put(47.5,1){$\vdots$}
\put(45,-3){$\scr -\frac{1}{2}$}
\put(47.5,40){$\vdots$}
\put(47.5,95){$\vdots$}
\put(47.5,91){$\vdots$}
\put(47.5,87){$\vdots$}
\put(47.5,83){$\vdots$}
\put(47.5,79){$\vdots$}
\put(47.5,136){$\vdots$}
\put(47.5,140){$\vdots$}
\put(57.5,1){$\vdots$}
\put(57,-3){$\scr \frac{1}{2}$}
\put(57.5,40){$\vdots$}
\put(57.5,95){$\vdots$}
\put(57.5,91){$\vdots$}
\put(57.5,87){$\vdots$}
\put(57.5,83){$\vdots$}
\put(57.5,79){$\vdots$}
\put(57.5,136){$\vdots$}
\put(57.5,140){$\vdots$}
\put(52,-3){$\scr 0$}
\put(52,140){$\vdots$}
\end{picture}
\end{center}
\vspace{0.3cm}
\begin{center}\begin{minipage}{30pc}{\small {\bf Fig.\,2:} The representation block of the second order of the group $\spin_+(1,3)$, that is, the main diagonal (8 cells) of spinorial chessboard of second order. This block is generated by the eight cycles of the group $BW_{\R}\simeq\dZ_8$. Two stability level are placed within circles.}\end{minipage}\end{center}
\end{figure}

The place of these energy levels at the boundaries of the fractal structure allows us to suppose that these levels have a nature of \textit{threshold scales}. In the fractal theory \cite{Mandel} notions of \textit{threshold scale} and \textit{effective dimension} play a key role. It means that at the value of threshold scale fractal acquires a certain new quality (quantity transits to quality), that is, a new stable level of matter is formed.

Let us study an elementary divisor structure of the energy level $\sH_E\simeq\Sym_{(31,30)}$. On this level we have representation $\boldsymbol{\tau}_{15,\frac{29}{2}}$ of the degree 930. Spintensor substrate for this level of the matter spectrum is
\[
\underbrace{\C_2\otimes\C_2\otimes\cdots\otimes\C_2}_{31\;\text{times}}\bigotimes
\underbrace{\overset{\ast}{\C}_2\otimes\overset{\ast}{\C}_2\otimes\cdots\otimes
\overset{\ast}{\C}_2}_{30\;\text{times}}
\]
and a spinspace for the level $\sH_E\simeq\Sym_{(31,30)}$ is
\[
\underbrace{\dS_2\otimes\dS_2\otimes\cdots\otimes\dS_2}_{31\;\text{times}}\bigotimes
\underbrace{\dot{\dS}_2\otimes\dot{\dS}_2\otimes\cdots\otimes\dot{\dS}_2}_{30\;\text{times}}
\simeq\dS_{2^{61}}.
\]
In this case helicity basis consists of 930 ket-vectors:
\[
\left| 15,15;\tfrac{29}{2},\tfrac{29}{2}\right\rangle,\;\left| 15,14;\tfrac{29}{2},\tfrac{29}{2}\right\rangle,\;\ldots,\;
\left| 15,1;\tfrac{29}{2},\tfrac{29}{2}\right\rangle,\;\left| 15,0;\tfrac{29}{2},\tfrac{29}{2}\right\rangle,\;\ldots
\left| 15,-15;\tfrac{29}{2},\tfrac{29}{2}\right\rangle;
\]
\[
\left| 15,15;\tfrac{29}{2},\tfrac{27}{2}\right\rangle,\;\left| 15,14;\tfrac{29}{2},\tfrac{27}{2}\right\rangle,\;\ldots,\;
\left| 15,1;\tfrac{29}{2},\tfrac{27}{2}\right\rangle,\;\left| 15,0;\tfrac{29}{2},\tfrac{27}{2}\right\rangle,\;\ldots
\left| 15,-15;\tfrac{29}{2},\tfrac{27}{2}\right\rangle;
\]
\[
\ldots\ldots\ldots\ldots\ldots\ldots\ldots\ldots\ldots\ldots\ldots\ldots\ldots\ldots\ldots\ldots
\]
\[
\left| 15,15;\tfrac{29}{2},\tfrac{1}{2}\right\rangle,\;\left| 15,14;\tfrac{29}{2},\tfrac{1}{2}\right\rangle,\;\ldots,\;
\left| 15,1;\tfrac{29}{2},\tfrac{1}{2}\right\rangle,\;\left| 15,0;\tfrac{29}{2},\tfrac{1}{2}\right\rangle,\;\ldots
\left| 15,-15;\tfrac{29}{2},\tfrac{1}{2}\right\rangle;
\]
\[
\left| 15,15;\tfrac{29}{2},-\tfrac{1}{2}\right\rangle,\;\left| 15,14;\tfrac{29}{2},-\tfrac{1}{2}\right\rangle,\;\ldots,\;
\left| 15,1;\tfrac{29}{2},-\tfrac{1}{2}\right\rangle,\;\left| 15,0;\tfrac{29}{2},-\tfrac{1}{2}\right\rangle,\;\ldots
\left| 15,-15;\tfrac{29}{2},-\tfrac{1}{2}\right\rangle;
\]
\[
\ldots\ldots\ldots\ldots\ldots\ldots\ldots\ldots\ldots\ldots\ldots\ldots\ldots\ldots\ldots\ldots
\]
\[
\left| 15,15;\tfrac{29}{2},-\tfrac{29}{2}\right\rangle,\;\left| 15,14;\tfrac{29}{2},-\tfrac{29}{2}\right\rangle,\;\ldots,\;
\left| 15,1;\tfrac{29}{2},-\tfrac{29}{2}\right\rangle,\;\left| 15,0;\tfrac{29}{2},-\tfrac{29}{2}\right\rangle,\;\ldots
\left| 15,-15;\tfrac{29}{2},-\tfrac{29}{2}\right\rangle.
\]
In this ket-basis the matrix $\Lambda^{15,\frac{29}{2}}_3$ has the following structure:
\begin{multline}
\Lambda^{15,\frac{29}{2}}_3=\text{{\rm diag}}\left({}_1\Lambda^{15,\frac{29}{2}}_3,\;{}_2\Lambda^{15,\frac{29}{2}}_3,\;
{}_3\Lambda^{15,\frac{29}{2}}_3,\;\ldots,\;{}_{15}\Lambda^{\frac{29}{2},14}_3,\;\right.\\
\left.-{}_{15}\Lambda^{15,\frac{29}{2}}_3,\;\ldots,\;-{}_3\Lambda^{15,\frac{29}{2}}_3,-{}_2\Lambda^{15,\frac{29}{2}}_3,\;
-{}_1\Lambda^{15,\frac{29}{2}}_3\right),\nonumber
\end{multline}
where
\[
{}_1\Lambda^{15,\frac{29}{2}}_3=\text{{\rm diag}}\left(\frac{435}{2}, 203, \frac{377}{2},\ldots, \frac{29}{2}, 0, -\frac{29}{2},\ldots, -\frac{377}{2}, -203, -\frac{435}{2}\right),
\]
\[
{}_2\Lambda^{15,\frac{29}{2}}_3=\text{{\rm diag}}\left(\frac{405}{2}, 189, \frac{351}{2},\ldots,\frac{27}{2}, 0,
-\frac{27}{2},\ldots, -\frac{351}{2},
-189, -\frac{405}{2}\right),
\]
\[
{}_3\Lambda^{15,\frac{29}{2}}_3=\text{{\rm diag}}\left(\frac{375}{2}, 175, \frac{325}{2}, \ldots, \frac{25}{2}, 0,
-\frac{25}{2}, \ldots, -\frac{325}{2},
-175, -\frac{375}{2}\right),
\]
\[
\ldots\ldots\ldots\ldots\ldots\ldots\ldots\ldots\ldots\ldots\ldots\ldots\ldots\ldots\ldots\ldots
\]
\[
{}_{15}\Lambda^{15,\frac{29}{2}}_3=\text{{\rm diag}}\left(\frac{15}{2}, 7, \frac{13}{2}, \ldots, \frac{1}{2}, 0, -\frac{1}{2}, \ldots, -\frac{13}{2},
-7, -\frac{15}{2}\right)
\]
are $31\times 31$ diagonal matrices.

Among the 930 eigenvectors of $\boldsymbol{\Lambda}$ in $\sH_E\simeq\Sym_{(31,30)}$ there are 329 eigenvectors with different eigenvalues. Therefore, characteristic polynomial for $\Lambda^{15,\frac{29}{2}}_3$ consists of 329 different factors (elementary divisors):
\begin{multline}
\boldsymbol{\Delta}(\Lambda^{15,\frac{29}{2}}_3)=\lambda^{30}(\lambda-45/2)^8(\lambda-15/2)^8(\lambda-135/2)^6\cdots
(\lambda-1/2)^2\\
(\lambda+45/2)^8(\lambda+15/2)^8(\lambda+135/2)^6\cdots
(\lambda+1/2)^2.\nonumber
\end{multline}
This polynomial contains four elementary divisors of 8-th order with eigenvalues\footnote{It means that these eigenvalues of $\boldsymbol{\Lambda}$ are 8-fold degenerate. It should be noted that $\boldsymbol{\tau}_{15,29/2}$ is the first representation on the spin-1/2 line, in which there are elementary divisors of eighth order with non-null eigenvalues. In preceding representation $\boldsymbol{\tau}_{29/2,14}$ we have divisors of sixth order only.}
\[
\pm\frac{15}{2},\;\pm\frac{45}{2};
\]
18 elementary divisors of 6-th order with eigenvalues
\[
\pm\frac{9}{2},\;\pm\frac{21}{2}\;\pm\frac{27}{2},\;\pm 15,\;\pm 21,\;\pm\frac{63}{2},\;\pm\frac{75}{2},\;
\pm\frac{105}{2},\;\pm\frac{135}{2}
\]
and 74 elementary divisors of 4-th order with eigenvalues
\[
\pm\frac{3}{2},\;\pm\frac{5}{2},\;\pm 3,\;\frac{7}{2},\;\pm 5,\;\pm\frac{11}{2},\;\pm 6,\;\pm\frac{13}{2},\;
\pm 7,\;\pm 9,\;\frac{25}{2},\;\pm\frac{33}{2},\;\pm\frac{35}{2},
\]
\[
\pm 18,\;\frac{39}{2},\;\pm 25,\;\pm 27,\;\pm\frac{55}{2},\;\pm 30,\pm\frac{65}{2},\;\pm 35,\;\pm\frac{77}{2},\;
\pm\frac{81}{2},\;\pm 42,\;\pm 45,
\]
\[
\pm\frac{91}{2},\;\pm\frac{99}{2},\;\pm 54,\;\pm\frac{117}{2},\;\pm 63,\;\pm\frac{143}{2},\;\pm 75,\;\frac{165}{2},\;
\pm\frac{189}{2},\;\pm\frac{195}{2},\;\pm 105,\;\pm\frac{225}{2}.
\]
All other non-null eigenvalues of $\boldsymbol{\Lambda}$ in $\sH_E\simeq\Sym_{(31,30)}$ are twofold degenerate.

As in the case of electron (see section 3.1), proton is a superposition of state vectors in nonseparable Hilbert space $\bsH^S_2\otimes\bsH^+\otimes\bsH_\infty$, that corresponds to a spin doublet
\[
\unitlength=1mm
\begin{picture}(20,13)
\put(0,5){$\overset{(\frac{29}{2},15)}{\bullet}$}
\put(20,5){$\overset{(15,\frac{29}{2})}{\bullet}$}
\put(5,6){\line(1,0){20}}
\put(0.5,0){$-\frac{1}{2}$}
\put(23.5,0){$\frac{1}{2}$}
\put(7,0){$\cdots$}
\put(12.5,0){$\cdots$}
\put(18.5,0){$\cdots$}
\end{picture}
\]
In this case we have two spin states: the state 1/2, described by the representation $\boldsymbol{\tau}_{15,29/2}$ on the spin-1/2 line, and the state -1/2, described by the representation $\boldsymbol{\tau}_{29/2,15}$ on the dual spin-1/2 line. Representations $\boldsymbol{\tau}_{15,29/2}$ and $\boldsymbol{\tau}_{29/2,15}$ act in the spaces $\Sym_{(31,30)}$ and $\Sym_{(30,31)}$, respectively. Therefore, there are two state vectors in $\bsH^S_2\otimes\bsH^+\otimes\bsH_\infty$: ket-vector  $\left|\Psi\right\rangle=\left.\left|\boldsymbol{\tau}_{15,29/2},\,\Sym_{(31,30)},\,\C_{122},\,\dS_{2^{61}}
\right.\right\rangle$ and bra-vector $\langle\dot{\dS}_{2^{61}},\,\overset{\ast}{\C}_{122},\,\Sym_{(30,31)},\,
\boldsymbol{\tau}_{29/2,15}|=\langle\dot{\Psi}|$. Equations of type (\ref{BS2}) for the spin chain $\left(15,29/2\right)\longleftrightarrow\left(29/2,15\right)$ are
\[
\sum^3_{j=1}\overset{\ast}{\Lambda}{}^{\frac{29}{2},15}_j\frac{\partial\dot{\psi}}
{\partial\widetilde{a}_j}+i\sum^3_{j=1}\overset{\ast}{\Lambda}{}^{\frac{29}{2},15}_j
\frac{\partial\dot{\psi}}{\partial\widetilde{a}^\ast_j}+
m_p\psi=0,
\]
\[
\sum^3_{j=1}\Lambda^{15,\frac{29}{2}}_j\frac{\partial\psi}{\partial a_j}-
i\sum^3_{j=1}\Lambda^{15,\frac{29}{2}}_j\frac{\partial\psi}{\partial a^\ast_j}+
m_p\dot{\psi}=0,
\]
where the matrices $\Lambda^{15,\frac{29}{2}}_1$ and $\Lambda^{15,\frac{29}{2}}_2$ have the same elementary divisor structure as $\Lambda^{15,\frac{29}{2}}_3$.

It is obvious that the next stability level of matter spectrum occurs as a threshold scale at the boundary of spinorial chessboard of third order\footnote{Of course, third stability level can be identified with a so-called `dark matter'.}.
\section{Summary}
We have presented group-theoretical description of \textit{matter spectrum} in terms of interlocking representations of the Lorentz group. The each level of this spectrum corresponds to subspace $\sH_E\simeq\Sym_{(k,r)}$ of the energy operator $H$. On the other hand, the each level of matter spectrum presents a state vector in the spin-charge nonseparable Hilbert space $\bsH^S\otimes\bsH^Q\otimes\bsH_\infty$. Superposition of state vectors in $\bsH^S\otimes\bsH^Q\otimes\bsH_\infty$ is an elementary particle. In such description there are no fundamental particles, all the levels of matter spectrum are equivalent (they are differed from each other by the value of energy as actualized particle states with the given values of spin and mass). \textit{Fundamental symmetry} of matter spectrum is defined by the Lorentz group, whereas \textit{dynamical symmetries} ($\SU(3)$, $\SU(6)$ and so on) can be defined in $\bsH^S\otimes\bsH^Q\otimes\bsH_\infty$ via central extension (for more details see \cite{Var15}). From this viewpoint, following to Heisenberg' opinion \cite{Heisen}, all the approximate dynamical symmetries should be considered as \textit{secondary symmetries}. As is known, Heisenberg divided all known symmetries in particle physics on the two categories: fundamental (primary) symmetries and dynamical (secondary) symmetries. In this context it should be noted that quark scheme is a derivative construction of dynamical $\SU(3)$-symmetry. Moreover, all basic facts of $\SU(3)$- and $\SU(6)$-theories, concerning systematization of hadron spectra, can be obtained without usage of quark hypothesis (see \cite{RF70}).

Among the infinite set of matter levels there are two levels which have a special meaning (at first, from historical viewpoint). These levels correspond to Dirac and Maxwell equations. As the first fermionic and bosonic terms of matter spectrum, Dirac and Maxwell equations have a very similar spinor form. Therefore, these fields should be considered on an equal footing, from one group-theoretical viewpoint. As a consequence of this, we come to Majorana-Oppenheimer formulation of quantum electrodynamics. All other high-energy terms of matter spectrum are studied by similar manner. RWE for all matter levels have a common structure, this structure is defined by interlocking representations of the Lorentz group. Classification of energy spectra for simple and arbitrary spin chains has been obtained with respect to elementary divisor structure of the operators $\boldsymbol{\Lambda}$ (RWE-operators). Among the all matter levels there are stability levels that correspond to threshold scales of the fractal structure associated with the interlocking scheme. Hence it follows that underlying spinor structure has a very complicate nature. In this connection it should be noted recent Finkelstein works about simplex (modular) spinor structure \cite{Fin1,Fin2}.

\end{document}